\documentclass[11pt]{article}
\usepackage{graphicx,amsthm,graphicx,fancyhdr,mathrsfs,bbm}
\usepackage{amsfonts}
\usepackage{amsmath}
\usepackage{amssymb}
\usepackage{url}
\usepackage[colorlinks=true,citecolor=blue]{hyperref}

\usepackage{fancyhdr}
\usepackage{indentfirst}
\usepackage{enumerate}
\usepackage{amsthm}
\usepackage{color}
\usepackage[normalem]{ulem}
\usepackage{comment}
\usepackage{dsfont}
\usepackage{natbib}
\usepackage[misc]{ifsym}
\usepackage[toc,page]{appendix}
\usepackage{multirow}
\usepackage{float}
\usepackage{bm}
\usepackage{booktabs}
\usepackage{todonotes}
\usepackage{subcaption}
\usepackage{graphicx}
\usepackage{tabularx}
\usepackage{longtable}
\usepackage{pdflscape}
\usepackage{afterpage}

\def\d{\mathrm{d}}
\newcommand{\VaR}{\mathrm{VaR}}

\newcommand{\ES}{\mathrm{ES}}

\renewcommand{\H}{\mathcal{H}}

\newcommand{\I}{\mathcal{I}}

\newcommand{\X}{\mathcal{X}}

\newcommand{\E}{\mathbb{E}}
\newcommand{\R}{\mathbb{R}}
\newcommand{\N}{\mathbb{N}}

\newcommand{\M}{\mathcal{M}}

\newcommand{\p}{\mathbb{P}}
\renewcommand{\P}{\mathbb{P}}

\newcommand{\id}{\mathds{1}}

\renewcommand{\ge}{\geqslant}
\renewcommand{\le}{\leqslant}
\renewcommand{\geq}{\geqslant}
\renewcommand{\leq}{\leqslant}
\def\laweq{\buildrel \mathrm{d} \over =}
\def\lawis{\buildrel \mathrm{d} \over \sim}

\newcommand{\T}{\mathbb{T}}

\theoremstyle{plain}
\newtheorem{theorem}{Theorem}

\newtheorem{lemma}{Lemma}
\newtheorem{proposition}{Proposition}
\theoremstyle{definition}
\newtheorem{definition}{Definition}
\newtheorem{example}{Example}

\newtheorem{assumption}{Assumption}

\newtheorem{remark}{Remark}

\usepackage{mathtools}
\DeclareMathOperator*{\argmax}{arg\,max}

\definecolor{DarkGreen}{rgb}{0.2,0.6,0.2}

\def\green#1{\textcolor{DarkGreen}{#1}}
\definecolor{purple}{rgb}{0.6,0.3,0.8}
\definecolor{SelectedRed}{RGB}{163,58,53}
\definecolor{WatermarkBlue}{RGB}{0,114,178}
\definecolor{HumanReport}{RGB}{213,94,0}
\definecolor{CaseInk}{RGB}{37,52,61}
\definecolor{CaseBorder}{RGB}{170,183,190}
\definecolor{CaseTint}{RGB}{248,250,251}

\usepackage{pgfplots}
\pgfplotsset{compat=1.18}
\usepackage{tikz}
\usetikzlibrary{arrows.meta,positioning}

\renewcommand{\cite}{\citet}

\makeatletter
\DeclareRobustCommand{\bsquare}{%
	\mathop{\vphantom{\sum}\mathpalette\bigstar@\relax}\slimits@
}

\newcommand{\bigstar@}[2]{%
	\vcenter{%
		\sbox\z@{$#1\sum$}%
		\hbox{\resizebox{.9\dimexpr\ht\z@+\dp\z@}{!}{$\m@th\dsquare$}}%
	}%
}
\makeatother

\providecommand{\FDR}{\operatorname{FDR}}

\newtheorem{assumptionalt}{Assumption}[assumption]

\newenvironment{assumption-c}[1]{
  \renewcommand\theassumptionalt{#1}
  \assumptionalt
}{\endassumptionalt}

\usepackage[onehalfspacing]{setspace}

\newcommand{\dsquare}{\mathop{  \square} \displaylimits}

\begin{document}

\title{Sequential resetting  procedures 
and false discovery rate}

\author{
Qiuqi Wang\thanks{Maurice R.~Greenberg School of Risk Science, Georgia State University, USA. \texttt{qwang30@gsu.edu}}
\and Ruodu Wang\thanks{Department of Statistics and Actuarial Science, University of Waterloo, Canada. \texttt{wang@uwaterloo.ca}}
\and Zhenyuan Zhang\thanks{Laboratory for Information and Decision Systems, Massachusetts Institute of Technology, USA.   \texttt{zzyzzy@mit.edu}}}

\date{October 13, 2026}
\maketitle

\begin{abstract}
\medskip
\noindent 
Data arrive sequentially, each associated with a null hypothesis. We develop testing procedures to locate intervals in which some null hypotheses fail, with false discovery rate (FDR) control. The new procedures are called sequential resetting procedures, and they are based on e-values and test supermartingales. We also develop a refined version of the procedures by dropping less informative data points before the block minimum of the test supermartingale in each rejection block. These procedures have explicit FDR bounds under two settings: a classic setting of independence and the more general setting of possible dependence across null data and non-null data. These FDR bounds are independent of the testing horizon, and the general one has anytime validity, but it has an extra logarithm factor compared with the standard FDR level. We present simulation studies and data experiments with applications of sequential resetting procedures to LLM watermark detection and financial backtesting.

~


\end{abstract}

\section{Introduction}
\label{sec:intro}

A sequential test, also known as an online test, examines observations as they arrive and rejects once there is sufficient evidence against the null hypothesis  \citep{Wald47,Siegmund85}. Modern approaches based on e-processes allow for continuous monitoring with test validity under optional stopping \citep{GHK24,RGVS23}. 
There are two clear features of this classic approach: 
\begin{enumerate}
    \item The null hypothesis of interest is global, which  applies to all data.
    \item The method does not  try to identify which data points caused the rejection, and how many times such violations occur.
\end{enumerate}

Modern statistical applications require relaxing both points. 
Especially, the tester may like to identify where the departure from the null occurred and continue looking for further departures. We briefly mention two concrete examples:
\begin{enumerate}
    \item A financial risk forecast method may be adequate for several years and then underestimate losses during a crisis \citep{WWZ26}.  In this context, we want to report periods of inadequate forecasts. 
    \item A document may combine human writing with passages generated by a large language model (LLM) \citep{LRWLS25}; see
Figure~\ref{fig:llm-mixed-passage}. In this context,  the goal is to continuously identify watermarked texts.  
\end{enumerate}

\begin{figure}[!t]
\centering
\includegraphics[width=\textwidth]{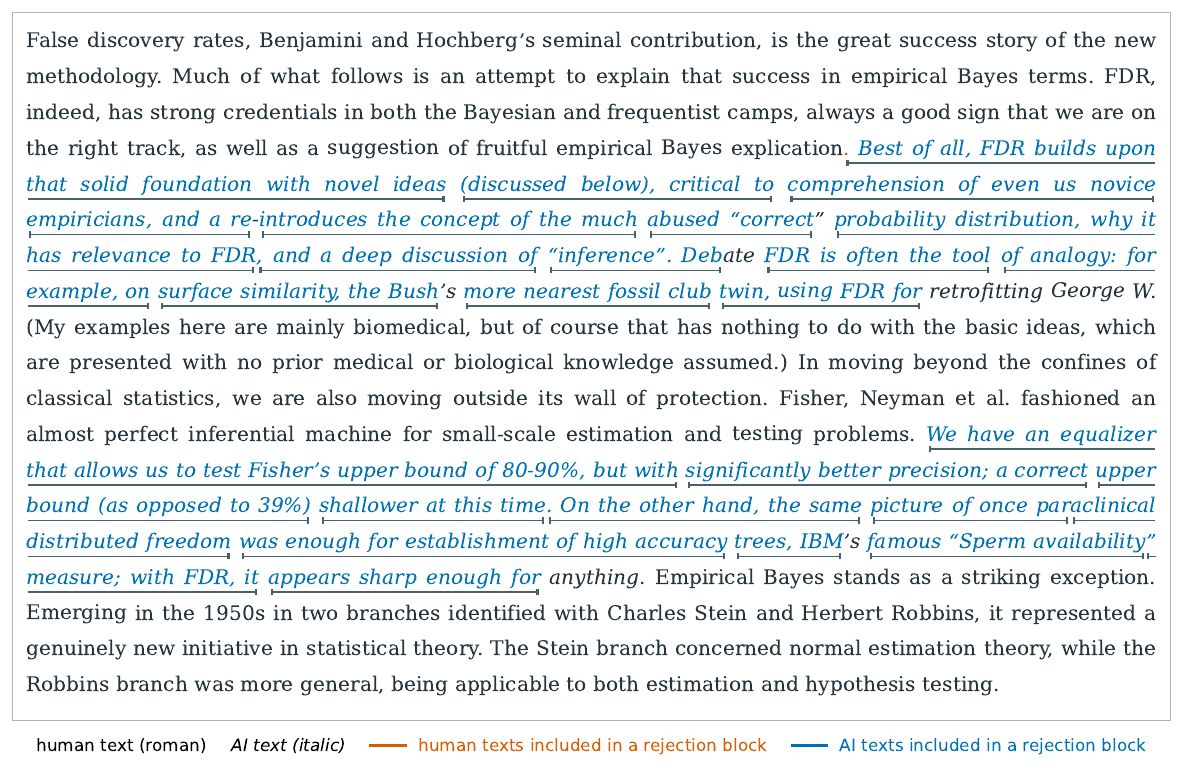}
\caption{The texts are from the prologue of Efron's \emph{Large-Scale Inference}~\citep{Efron10}, where two text blocks are replaced by watermarked model output (in italics) using Tournament-watermarked OPT-1.3B continuations. Each underbracket, illustrated by ``\(\protect\underbracket[0.4pt][3pt]{\,\mathrm{token}\,}\)'', represents one rejection block from our sequential resetting procedure. A report is false if its rejection block contains no watermarked position, and is true if there is at least one watermarked position. The above example reports 92.5\% of the AI-generated text, and none of the reported blocks is false.
See Section~\ref{sec:llm-mixed-case} for further details.}
\label{fig:llm-mixed-passage}
\end{figure}

With the applications above in mind, we associate a null hypothesis with each observation.   
Certainly, one data point is typically not enough to offer sufficient statistical evidence for decisions, and we should look at intervals. 
Therefore, we aim to develop methods that report several intervals of data, each with a violation of the null hypotheses corresponding to these data points. 
As we are interested in reporting multiple violations, the classic error metric of type-I error rate is no longer  suitable to use. There are several error metrics in multiple testing, and we consider the most popular one, the false discovery rate (FDR) introduced by \citet{BH95}, which suits our target applications in mind. FDR controls, on average, the proportion of rejections that are false, meaning that the null hypotheses   hold true during the period reported, with the proportion set to zero when there is no report. Interval reporting methods with FDR control have recently appeared in changepoint detection \citep{BM23,CGRWZ25, DW26}.

To describe the proposed methodology, our starting point is a conditionally valid e-value at each time step \citep{VW21,GHK24}.
This setting is not restrictive as an e-value can be obtained from any test or any p-value \citep{RW25}. 
Each individual e-value is typically computed from only one data point (or a batch of data points in some applications). When these individual values are modest, directly applying a multiple testing procedure such as e-BH \citep{WR22} can have little power.\footnote{E-value-based multiple testing procedures, such as those of \cite{WR22} and \cite{XSF25}, require (at least one of) the rejected e-values to be at least $1/\alpha$, where $\alpha$ is the FDR level.}
To design a useful method in this setting,  we take their cumulative products to obtain a test supermartingale \citep{SSVV11}, a nonnegative process that measures accumulating evidence against the null. This enables sequential testing without a prespecified stopping point or criterion. Intuitively, a test supermartingale can be interpreted as the capital of a bettor who starts with one unit and bets against the null. In this case, a large capital provides evidence for rejection. We refer to Section~\ref{sec:e} below for the precise definitions.

The methodology of testing via e-processes offers a useful starting point of our exploration for new methods. 
Fix a rejection threshold $\gamma=1/\alpha$, where $\alpha\in(0,1)$. The standard approach  based on e-processes should start at one and reject when the value of the process, also known as the capital with a betting interpretation \citep{Shafer21}, first reaches $\gamma$ \citep{RW25}.  This method has a controlled type-I error rate due to Ville's inequality.   Such a rejection, if it occurs,
is joint evidence from all the observations accumulated so far. Therefore, it is natural for a testing procedure to report all observations so far.

The main question is how to proceed once the first rejection has been made. 
We propose \emph{sequential resetting procedures}. 
Such a procedure resets  the capital to one and resumes testing, using the same threshold. Each subsequent rejection reports the observations since the preceding reset. Resetting the capital allows each new report to be based on new evidence, while the tester's betting strategy in constructing the e-process can be based on all earlier observations.

By default, for a sequential resetting procedure, \textit{a rejection is true as soon as the reported interval contains one alternative}, even if most of its observations satisfy the null. This makes the reports not as informative when the rejected intervals are long, in settings such as watermark detection. This motivates us to introduce a \emph{refinement rule}. At each rejection, we identify the last time at which the capital attained its lowest value during the completed run. Then, we report only the observations after that time through the rejection. Figure~\ref{fig:intro-localization} shows the resulting refined intervals. 

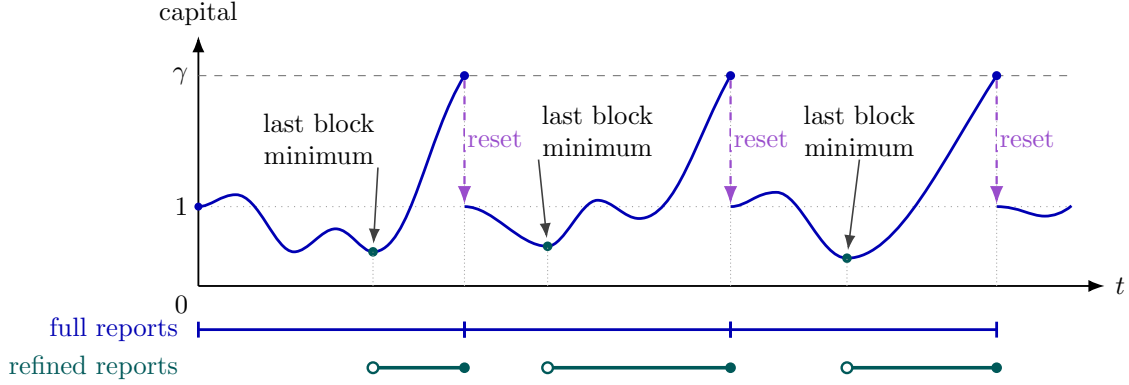
\begin{figure}[t]
\centering
\begin{tikzpicture}[x=1.10cm,y=1.05cm,>=Latex,font=\small]
\tikzset{
    capital/.style={blue!70!black,line width=1.05pt},
    localized/.style={teal!70!black,line width=1.35pt},
    guide/.style={black!35,densely dotted},
    resetarrow/.style={->,purple,dashed,line width=0.85pt},
    minarrow/.style={->,black!75,line width=0.6pt,shorten >=3pt}
}
\draw[->,line width=0.7pt] (0,0) -- (10.90,0) node[right] {$t$};
\draw[->,line width=0.7pt] (0,0) -- (0,3.15) node[above] {capital};
\node[below left] at (0,0) {$0$};
\draw[black!40,dotted] (0,1) -- (10.50,1);
\node[left] at (0,1) {$1$};
\draw[black!55,dashed] (0,2.65) -- (10.50,2.65);
\node[left] at (0,2.65) {$\gamma$};
\draw[capital] (0,1)
    .. controls (0.15,1.00) and (0.27,1.15) .. (0.45,1.15)
    .. controls (0.73,1.15) and (0.92,0.43) .. (1.15,0.43)
    .. controls (1.32,0.43) and (1.48,0.72) .. (1.65,0.72)
    .. controls (1.82,0.72) and (1.96,0.43) .. (2.10,0.43)
    .. controls (2.55,0.43) and (2.75,1.90) .. (3.20,2.65);
\draw[capital] (3.20,1)
    .. controls (3.55,1.00) and (3.85,0.50) .. (4.20,0.50)
    .. controls (4.40,0.50) and (4.60,1.08) .. (4.80,1.08)
    .. controls (4.98,1.08) and (5.12,0.85) .. (5.30,0.85)
    .. controls (5.75,0.85) and (6.00,1.92) .. (6.40,2.65);
\draw[capital] (6.40,1)
    .. controls (6.60,1.00) and (6.70,1.18) .. (6.95,1.18)
    .. controls (7.24,1.18) and (7.48,0.35) .. (7.80,0.35)
    .. controls (8.42,0.35) and (8.95,1.60) .. (9.60,2.65);
\draw[capital] (9.60,1)
    .. controls (9.87,1.00) and (10.00,0.88) .. (10.18,0.88)
    .. controls (10.28,0.88) and (10.40,0.92) .. (10.50,1.01);
\foreach \x in {3.20,6.40,9.60} {
    \draw[guide] (\x,0) -- (\x,2.65);
    \fill[blue!70!black] (\x,2.65) circle (1.7pt);
    \draw[resetarrow] (\x,2.59) -- (\x,1.03);
    \node[right,purple,fill=white,inner sep=1pt] at (\x,1.85) {reset};
}
\fill[blue!70!black] (0,1) circle (1.5pt);
\foreach \x/\y in {2.10/0.43,4.20/0.50,7.80/0.35} {
    \fill[teal!70!black] (\x,\y) circle (1.8pt);
    \draw[guide] (\x,0) -- (\x,\y);
}
\node[align=center,inner sep=1pt] (minimumone) at (1.45,1.85) {last block\\minimum};
\draw[minarrow] (minimumone.south east) -- (2.10,0.43);
\node[align=center,inner sep=1pt] (minimumtwo) at (4.80,1.95) {last block\\minimum};
\draw[minarrow] (minimumtwo.south west) -- (4.20,0.50);
\node[align=center,inner sep=1pt] (minimumthree) at (7.95,1.95) {last block\\minimum};
\draw[minarrow] (minimumthree.south) -- (7.80,0.35);
\node[anchor=east,blue!70!black] at (-0.12,-0.55) {full reports};
\node[anchor=east,teal!70!black] at (-0.12,-1.03) {refined reports};
\foreach \xstart/\xend in {0/3.20,3.20/6.40,6.40/9.60} {
    \draw[capital] (\xstart,-0.55) -- (\xend,-0.55);
    \draw[capital] (\xstart,-0.45) -- (\xstart,-0.65);
    \draw[capital] (\xend,-0.45) -- (\xend,-0.65);
}
\foreach \xstart/\xend in {2.10/3.20,4.20/6.40,7.80/9.60} {
    \draw[localized] (\xstart,-1.03) -- (\xend,-1.03);
    \draw[teal!70!black,fill=white,line width=0.9pt] (\xstart,-1.03) circle (2pt);
    \fill[teal!70!black] (\xend,-1.03) circle (2pt);
}
\end{tikzpicture}

\caption{Resetting and refinement. The capital of the sequential resetting procedure resets to one after each rejection at $\gamma$. A refined report retains observations after the last block minimum of the full report interval, at the time of rejection.}
\label{fig:intro-localization}
\end{figure}

The refinement rule does not change the rejection times (i.e., right endpoints of the reported intervals); the left endpoints are selected retrospectively when the rejection occurs.
Nevertheless, it is able to remain silent
when there is a long period during which the null hypotheses are true, thus providing more refined information for the decision maker. By definition, the FDR of the refined method is at least as large as that of the ordinary (unrefined) method.\footnote{We emphasize that, as a consequence, a true report need not consist entirely of alternative observations. The refinement makes a rejection easier to become false, because shortening a report may remove alternatives.}  
The two approaches serve different purposes.  The ordinary method only cares about how many times rejection happens and in which parts of the data; the union of the reported intervals covers the entire testing horizon except for the last segment. The refined method offers more precise locations of the rejections, reporting a non-covering collection of intervals. Surprisingly, in our main results, we show that the two versions satisfy almost the same level of FDR control, leading to our recommendation of the refined method in most applications, except when narrowing the reported intervals is not of interest.\footnote{More precisely, Example \ref{ex:infty} shows that there is a common tight bound for both methods; for the independent-null case \eqref{eq:FDR11}, the bound is tight up to the small factor of $(1-\alpha)^{-1}$. } 

\subsection{Main contributions}

Our main contributions are threefold. 
First, we propose a novel sequential resetting procedure and further enhance it by a new refinement rule to produce more informative reports (Section \ref{sec:pre}). There are several advantages with this approach: (i) The resetting and refinement rules for our procedure are simple, natural, and can be easily implemented. (ii) The resetting and refinement are computationally efficient: apart from necessary computations of the underlying e-process, this construction takes only $O(T)$ additional work through time $T$.\footnote{The reason is that the additional work compared to classic e-processes consists only of updating the value of the e-process and recording its running minimum and the time of the last minimum.} 
(iii) We provide FDR control with a constant rejection threshold that is horizon-independent, and our FDR bounds provide a direct choice of the constant threshold. 

Second, to the best of our knowledge, we are the first to establish nontrivial FDR bounds for repeated sequential procedures based on test supermartingales (Sections~\ref{sec:FDR indep} and~\ref{sec:FDR}). Write $\mathrm{FDR}^{\mathrm{ref}}_t$ as the FDR of the cumulative reports for the sequential resetting procedure by time $t$. If we assume that e-values on the null times are mutually independent and independent of the e-values on the alternative times \citep[Assumption~\ref{ass:ind}; similar to the independence condition for the BH procedure of][]{BH95}, we obtain (Theorem \ref{thm:localized-main})
\begin{align}
    \mathrm{FDR}^{\mathrm{ref}}_t\le \frac{\alpha}{1-\alpha}
    \qquad\text{for every deterministic time }t.\label{eq:FDR11}
\end{align}
Note that in the global-null setting, Ville's inequality gives an upper bound of $\alpha$, but the bound in \eqref{eq:FDR11} cannot be improved to $\alpha$ (Example \ref{ex:non-iid-alpha-counterexample}). 
Without this independence assumption, let $\mathrm{FDR}^{\mathrm{ref}}_\infty$ denote the expected worst-case false discovery proportion (FDP) over time. We prove that (Theorem \ref{thm:FDR_barrier})
\begin{align}
    \mathrm{FDR}^{\mathrm{ref}}_\infty\le \alpha\left(1+\log\frac{1}{\alpha}\right).\label{eq:FDR111}
\end{align}
The above results also hold for the ordinary unrefined case. The bound \eqref{eq:FDR111} can also be seen as an online FDR bound as it controls FDR at every finite stopping time. In addition, without the independence assumption, if the set of non-null times is deterministic, we derive a stochastic upper bound on the total number of reports. Applying this bound leads to finite-sample lower confidence bounds on the number of non-null observations.

The same expression in \eqref{eq:FDR111} appears in the FDR-linking theorem of \citet{S18}; Remark~\ref{rem:Su18} discusses the connection. Together, \eqref{eq:FDR11}, \eqref{eq:FDR111}, and Su's result illustrate an interesting feature:
\begin{center}
    \it The price of dependence between nulls and non-nulls is a multiplicative factor of $(1-\log \alpha)$.
\end{center}
However, the mathematical arguments in \eqref{eq:FDR111} and in \cite{S18} leading to the factor in this observation are quite different, so the   factor may be a curious coincidence (see Remark \ref{rem:Su18}).

Our proofs establish validity for the interval hypotheses selected by resetting and refinement, which is not supplied directly by existing online FDR results \citep{XR24,FXR24}, because the hypothesis for each report (especially its starting point) depends on the past data (i.e., history of the e-process), and refinement needs to be addressed. Tackling these difficulties requires at least two new inputs:
\begin{itemize}
    \item First, we adaptively identify relevant all-null time intervals and invent a novel weighting function that connects the starting and ending values of the e-process on these intervals with the cumulative FDR. Intuitively, this is necessary because the terminal capital affects subsequent rejection counts and therefore the FDR denominator.
    \item Second, we prove a refined version of Ville's inequality, which may be of independent interest. {Roughly speaking, we show that conditional on the past at the previous reset, the probability of issuing a false refined report is at most $\alpha$ (Lemma \ref{lem:ville}). Here, if we replace ``refined" with ``ordinary", this is the classic Ville's inequality. The key difficulty in proving this extension is that the location of the block minimum (i.e., the left endpoint of the refined interval) is not a stopping time.}
\end{itemize}
 Our proofs of the main results are contained in Appendix~\ref{app:sec:proofs}, and streamlined versions are sketched in Sections~\ref{sec:FDR indep} and~\ref{sec:FDR}.

Third, we discuss applications of sequential resetting procedures to LLM watermark detection and financial backtesting (Section~\ref{sec:num}). We demonstrate the first repeated online reporting of adaptively refined intervals containing LLM watermark evidence with nontrivial finite-sample FDR guarantees, which remain valid under optional stopping in the general setting.
Both Gumbel--max and Tournament watermarks are studied.\footnote{Tournament sampling underlies SynthID-Text \citep{DathathriEtAl24}, a version of which Anthropic has recently adopted for supported Claude models \citep{Anthropic26}. More recently, OpenAI introduced textGrain, an alternative watermarking scheme based on entropy-constrained optimal transport \citep{TextGrain26}.}  We also initiate the study of repeated sequential reporting in financial backtesting with an FDR guarantee. We extend existing sequential ES backtests to continue after rejection and report further periods containing forecast inadequacy, with FDR control across the accumulated warnings. Complementary experiments are presented in Appendix~\ref{sec:add_num}.

\subsection{Related works}

Our framework builds upon the theory of e-values and e-processes, as well as their ancestors, test supermartingales. \citet{SSVV11} study test supermartingales and their relation to p-values and Bayes factors. \citet{VW21} develop the theory of e-values and their combination. \citet{GHK24} provide foundational background for e-values and e-processes. \citet{RGVS23} present a review on e-values and game-theoretic statistics, and the monograph by \citet{RW25} gives a thorough treatment of hypothesis testing with e-values.

FDR was proposed and studied first by the celebrated \citet{BH95}, with many subsequent important developments, such as \cite{BY01}, \cite{S02}, and \cite{ST03}; see the recent summary of \cite{S25}. For FDR with e-values, \citet{WR22} introduce the e-BH procedure, which controls FDR under arbitrary dependence. Concrete further developments  are made by \citet{RB24},   \citet{XR24}, \citet{FXR24},  \cite{GS25},  and \cite{XSF25}, among many others. Standard e-BH and closed e-BH select individual hypotheses from the supplied e-values and control their FDR without requiring temporal independence. Our goal instead is to report intervals selected by sequential accumulation, resetting, and refinement, and to control FDR for these interval reports.

Resetting also appears in changepoint detection with a different focus. \citet{DR25} show that finite average run length (the expected time to the first alarm under the global null) precludes nontrivial FDR control uniformly over stopping times. Consequently, guarantees concern waiting times to false alarms and alarm frequencies \citep{Vovk21,VNG25}, which can be used to provide the rejection threshold \citep{SRR24}. Our procedure has infinite average run length and controls FDR across cumulative reports.

\section{Sequential resetting procedures and false discovery rate}
\label{sec:pre}

\subsection{Hypotheses and e-values}
\label{sec:mon}

We first set up some notation and convention.
We use the convention $\inf\emptyset=\min\emptyset=\infty$. For $x,y\in\R$, write $x\wedge y=\min\{x,y\}$, $x\vee y=\max\{x,y\}$, and $x_+=x\vee 0$. Throughout, ``increasing" and ``decreasing" are in the nonstrict sense. The notation $\mathcal L(\cdot)$ stands for the law (distribution) of a random variable.
Write $\N_0=\N\cup\{0\}$.
For a time horizon $T\in\N\cup\{\infty\}$, write $\T=\{0,1,\dots,T\}$,  $\T_{\geq l}=\T\setminus\{0,\dots,l-1\}$ for $l\in\T$ and $\T_{\geq \infty}=\emptyset$. 

{Consider a standard Borel state space  $(\X,\mathcal B)$. Let $\Omega=\X^{\T_{\ge 1}}$ be a sample space and define the projection maps $X_t:\Omega\to\X$ for $t\in\T_{\ge 1}$. 
The sequence $(X_t)_{t\in \T_{\ge1}}$ represents our observed data; we do not assume that they are independent.
Denote the information up to $t\in\T_{\ge 1}$ by the $\sigma$-algebra $\mathcal F_t=\sigma(X_1,\dots,X_t)$ with the convention $\mathcal F_0=\{\emptyset,\Omega\}$, and write $\mathcal F=\sigma (\bigcup_{t\in\T_{\ge 1}}\mathcal F_t )$. The path sample space is thus   $(\Omega,\mathcal F)$.   
Let $\p$ be the true data-generating probability measure on $(\Omega,\mathcal F)$, and denote by $\p_t$ a regular conditional law of $X_t$ given $\mathcal F_{t-1}$. 
For each $t\in\T_{\ge 1}$, a regulator monitors a mapping $\mathcal P_t:\X^{t-1}\to 2^{\mathcal P(\X)}$ given historical data up to time $t-1$ (we treat $\X^0$ as a singleton), where $\mathcal P(\X)$ is the set of probability measures on $(\X,\mathcal B)$. When there is no notational confusion, we write $\mathcal P_t=\mathcal P_t(X_1,\dots,X_{t-1})$ for simplicity.
}

For each $t\in\T_{\ge 1}$, the null hypothesis is $H_t=\{\p_t\in \mathcal P_t\}$, which is data-dependent; it can depend  on the first $t-1$ observations. For any probability measure \(Q\) on \((\Omega,\mathcal F)\), write
\(Q_t\) for a regular conditional law of \(X_t\) given
\(\mathcal F_{t-1}\) under \(Q\). Define the class of path laws under which all
null conditions hold as
\[
\mathfrak P_0
=
\Big\{
Q \in \mathcal P(\Omega):\ Q_t\in \mathcal P_t(X_1,\dots,X_{t-1})
\ \text{for all }t\in\T_{\ge 1},\ Q\text{-a.s.}
\Big\},
\]
where $ \mathcal P(\Omega)$ is the set of probability measures on $(\Omega,\mathcal F$). 

 An \emph{e-variable} for a set of probability measures $\mathcal Q\subseteq \mathcal P(\Omega)$ (representing a null hypothesis) is defined as a $\mathcal F$-measurable nonnegative random variable $E:\Omega \to[0,\infty]$ with $\E^Q[E]\le 1$ for each $Q\in\mathcal Q$. An e-variable is also called an \emph{e-value}. An \emph{e-test} rejects the null $\mathcal Q$ if an e-value is higher than a prespecified threshold.

For each $t\in\T_{\ge1}$, the data  generate a nonnegative, $\mathcal F_t$-measurable e-variable $E_t$ satisfying 
\begin{align}
\E^Q[E_t \mid\mathcal F_{t-1}]\le1 
\quad \mbox{ for } Q \in \mathcal P(\Omega) \text{ with }Q_t\in\mathcal P_t.
\label{eq:evalue}
\end{align} 
Such e-values can be obtained by calibrating conditionally valid p-values or a pivotal statistic for each $H_t$, but  in general e-values do not require p-values or pivotal statistics to be available. 
A \emph{test supermartingale} for
$\mathcal Q\subseteq \mathcal P(\Omega) $ 
is a nonnegative stochastic process
$(M_t)_{t\in \T}$ adapted to $(\mathcal F_t)_{t\in\T}$ such that $M_0\le 1$ and $\E^Q[M_t|\mathcal F_{t-1}]\le M_{t-1}$
for all $t\in \T_{\ge 1}$ and $Q\in \mathcal Q$.
We will use a particular construction:
Set $M_0=1$ and 
\begin{align}
    M_t=E_tM_{t-1},~~~t\in\T_{\ge 1}, \mbox{~with the convention $0\times \infty=0$.} \label{eq:increments}
\end{align}
The  process $M=(M_t)_{t\in\T}$ in \eqref{eq:increments} is a {test supermartingale} for
$\mathfrak P_0$ if \eqref{eq:evalue} holds.  If only $(M_t)_{t\in \T}$ is directly observed and it is finite valued,  we can construct e-values by setting $E_t=M_t/M_{t-1}$ for $t\in\T_{\ge 1}$ when $M_{t-1} >0 $ and $E_t=1$ when $M_{t-1}=0$. 
A nonnegative $(\mathcal F_t)_{t\in\T}$-adapted stochastic process $(M_t)_{t\in\T}$ is called an \emph{e-process} for $\mathcal Q$ if $\E^Q[M_\rho]\le 1$ for any stopping time $\rho$ valued in $\T$ and any $Q\in \mathcal Q$. A test supermartingale is an e-process for the same hypothesis; the converse is not true in general.

\subsection{Tests and false discovery rate}
\label{sec:FDR_def}

A key feature of e-processes is that a rejection at time $t_0$ provides evidence against $\bigcap_{s=1}^{t_0}H_s$. Continuing the same test after $t_0$ without resetting does not distinguish new violations after $t_0$ from evidence accumulated before that rejection.   This motivates us to consider rejections of intervals instead of each individual  points. 

\begin{definition}
    \label{def:testing-procedure}
   A \emph{(sequential interval) testing procedure} $\phi=(\phi_t)_{t\in \T} $ 
is such that:
\begin{enumerate}[(a)]
    \item for $t \in \T$, $\phi_t$ is defined on $\X^t$ and $\phi_t(x_1,\dots,x_t)$ is a  set of nonempty disjoint subsets of ${\{1,\dots,t\}}$ for each $(x_1,\dots,x_t)\in \X^t$; 
    \item for $t\in \T_{\ge 1}$ and $(x_1,\dots,x_{t})\in \X^{t}$, it holds  $\phi_{t-1}(x_1,\dots,x_{t-1})\subseteq \phi_{t}(x_1,\dots,x_{t})$.  
\end{enumerate}
\end{definition}

The output of $\phi_t$ on any observed data $(x_1,\dots,x_{t})$  
is the set of \emph{rejected intervals} based on the data. 
Note that for $t=0$, $\phi_0=\emptyset$, which has no input.
Intuitively, condition (a) means that we cannot reject future hypotheses before observing their corresponding data. Condition (b) means that once a rejection of an interval is made, it cannot be revoked with observing later data. 
Note that it is possible to reject more from $\{1,\dots,t\}$ when observing $X_{t+1}$.
Traditional sequential tests of multiple hypotheses can be seen as the special case of our Definition~\ref{def:testing-procedure} in which only singleton outputs are allowed.

In our setting, 
a rejection of an interval $I\subseteq \T_{\ge 1}$ is considered \emph{true} if $\p_t\not\in \mathcal P_t$ for some $t\in I$,
and it is \emph{false} if 
$\p_t \in \mathcal P_t$ for all $t\in I$.
Denote by 
\begin{align}
    V^{\phi}_t =\sum_{I\in \phi_t(X_1,\dots,X_t)} 
\id_{\bigcap_{s\in I}\{\p_s \in \mathcal P_s\}},\quad \mbox{and}\quad
R^{\phi}_t = |\phi_t(X_1,\dots,X_t)|,~~t\in\T,\label{eq:V-R}
\end{align}
which are the number of false rejections of $\phi$ 
and the total number of rejections  up to time $t$, respectively. 
Following \cite{BH95}, the  \emph{false discovery proportion} 
and the 
\emph{false discovery rate}  of $\phi$ up to time $t\in \T$   are defined as 
\begin{equation}
    \label{eq:FDP-FDR}
 \mathrm{FDP}^{\phi}_t=\frac{V_t^\phi}{R_t^\phi\vee 1}
\quad \mbox{and}\quad
\FDR^{\phi}_t = \E\left[\frac{V^\phi_t}{R^\phi_t\vee 1}\right],
\end{equation}
where the expectation is taken under the data-generating probability $\p$. 
A testing procedure $\phi$ has FDR guarantee at level $\alpha\in (0,1)$ if $\FDR^{\phi}_t\le \alpha$ for all choices of the null hypotheses $(\mathcal P_t)_{t\in \T_{\ge 1}}$.

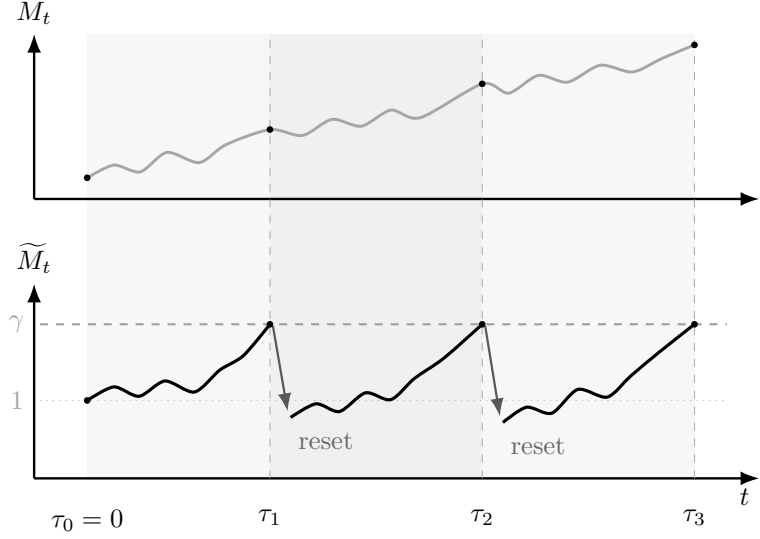
\begin{figure}[t!]
\centering
\begin{tikzpicture}[x=0.78cm,y=1.12cm,>=Latex,font=\small]
    \useasboundingbox (-0.95,-0.75) rectangle (12.95,5.65);
    \fill[white] (-0.95,-0.75) rectangle (12.95,5.65);

    \def\xzero{0.9}
    \def\xtauone{4.0}
    \def\xtautwo{7.6}
    \def\xtauthree{11.2}

    \def\ytopbase{3.30}
    \def\ytopmax{5.25}
    \def\ybotbase{0.00}
    \def\ybotone{0.92}
    \def\ybotthr{1.82}
    \def\ybotmax{2.30}

    \fill[gray!6] (\xzero,\ybotbase) rectangle (\xtauone,\ytopmax);
    \fill[gray!12] (\xtauone,\ybotbase) rectangle (\xtautwo,\ytopmax);
    \fill[gray!6] (\xtautwo,\ybotbase) rectangle (\xtauthree,\ytopmax);

    \draw[->,line width=0.9pt] (0,\ytopbase) -- (12.3,\ytopbase);
    \draw[->,line width=0.9pt] (0,\ybotbase) -- (12.3,\ybotbase);
    \draw[->,line width=0.9pt] (0,\ybotbase) -- (0,\ybotmax);
    \draw[->,line width=0.9pt] (0,\ytopbase) -- (0,\ytopmax);

    \node[below] at (12.05,\ybotbase) {$t$};
    \node[anchor=south] at (0,\ytopmax) {$M_t$};
    \node[anchor=south] at (0,\ybotmax) {$\widetilde M_t$};

    \draw[gray!75,dashed,line width=0.8pt] (0.1,\ybotthr) -- (11.75,\ybotthr);
    \node[gray!75,left] at (0,\ybotthr) {$\gamma$};
    \draw[gray!60,dotted] (0.1,\ybotone) -- (11.75,\ybotone);
    \node[gray!70,left] at (0,\ybotone) {$1$};

    \foreach \x/\lab in {\xtauone/$\tau_1$,\xtautwo/$\tau_2$,\xtauthree/$\tau_3$} {
        \draw[gray!65,dashed] (\x,\ybotbase) -- (\x,\ytopmax);
        \node[below=8pt] at (\x,\ybotbase) {\lab};
    }
    \node[below=8pt] at (\xzero,\ybotbase) {$\tau_0=0$};

    \draw[gray!70,line width=1.1pt]
        plot[smooth] coordinates {
            (\xzero,3.55)
            (1.35,3.70)
            (1.80,3.62)
            (2.25,3.85)
            (2.80,3.73)
            (3.25,3.94)
            (\xtauone,4.12)
            (4.55,4.05)
            (5.05,4.24)
            (5.55,4.16)
            (6.05,4.35)
            (6.55,4.26)
            (\xtautwo,4.66)
            (8.05,4.55)
            (8.55,4.76)
            (9.05,4.68)
            (9.60,4.88)
            (10.15,4.80)
            (10.65,4.96)
            (\xtauthree,5.12)
        };

    \draw[black,line width=1.2pt]
        plot[smooth] coordinates {
            (\xzero,\ybotone)
            (1.35,1.08)
            (1.78,0.97)
            (2.22,1.15)
            (2.72,1.02)
            (3.15,1.28)
            (3.55,1.45)
            (\xtauone,\ybotthr)
        };

    \draw[black,line width=1.2pt]
        plot[smooth] coordinates {
            (4.35,0.72)
            (4.78,0.88)
            (5.18,0.79)
            (5.62,1.01)
            (6.05,0.93)
            (6.45,1.18)
            (6.95,1.42)
            (\xtautwo,\ybotthr)
        };

    \draw[black,line width=1.2pt]
        plot[smooth] coordinates {
            (7.95,0.66)
            (8.35,0.84)
            (8.78,0.77)
            (9.22,1.05)
            (9.72,0.96)
            (10.12,1.21)
            (10.60,1.49)
            (\xtauthree,\ybotthr)
        };

    \draw[black!65,->,line width=0.85pt] (\xtauone+0.05,\ybotthr-0.02) -- (4.30,0.79);
    \draw[black!65,->,line width=0.85pt] (\xtautwo+0.05,\ybotthr-0.02) -- (7.90,0.73);

    \fill (\xzero,\ybotone) circle (1.25pt);
    \fill (\xtauone,\ybotthr) circle (1.25pt);
    \fill (\xtautwo,\ybotthr) circle (1.25pt);
    \fill (\xtauthree,\ybotthr) circle (1.25pt);
    \fill (\xzero,3.55) circle (1.25pt);
    \fill (\xtauone,4.12) circle (1.25pt);
    \fill (\xtautwo,4.66) circle (1.25pt);
    \fill (\xtauthree,5.12) circle (1.25pt);

    \node[black!60] at (4.95,0.44) {reset};
    \node[black!60] at (8.55,0.38) {reset};
\end{tikzpicture}
\caption{Schematic of the test supermartingale $M_t$ (top panel) and the corresponding resetting process $\widetilde M_t$ (bottom panel).}
\label{fig:refreshing-schematic}
\end{figure}

\subsection{Sequential resetting procedures and a refinement rule}
\label{sec:e}

We now propose our new testing procedures based on the ideas of resetting and refinement.  
For  $M$ in \eqref{eq:increments} and a pre-specified threshold $\gamma\in(1,\infty)$,  define the \emph{resetting process} $\widetilde M=(\widetilde M_t)_{t\in\T}$ recursively by $\widetilde M_0=Z_0=1$ and
\begin{equation}
\label{eq:M_tilde}
    \widetilde M_t=Z_{t-1}E_t,~~
    Z_t=\left\{\begin{array}{ll}
    1, & \mbox{if}~\widetilde M_{t}\ge \gamma; \\
    \widetilde M_t, & \mbox{otherwise},
\end{array}\right.
~~\mbox{for}~t\in\T_{\ge 1}.
\end{equation}
Here, $Z=(Z_t)_{t\in\T}$ is a post-reset capital process resetting to $1$ whenever the resetting process $\widetilde M$ exceeds the threshold $\gamma$.
See Figure \ref{fig:refreshing-schematic} for an illustration. 
One can easily check that the resulting process $\widetilde M$ is still a test supermartingale  for $\mathfrak P_0$. We recursively define stopping times $\tau_0=0$ and
\begin{equation}
\label{eq:tau}
    \tau_r=\tau_r(\widetilde M)=\min\{t\in\T_{\geq \tau_{r-1}+1}:\widetilde M_t\ge \gamma\},~~r\in\N.
\end{equation}

Throughout, intervals like $(\tau_{r-1},\tau_r]$ are interpreted as their intersections with $\N$.

\begin{definition} \label{def:ordinary}
   The \emph{ordinary sequential resetting procedure with rejection threshold $\gamma> 1$} is the testing procedure  $\phi^{\mathrm{ord}}=(\phi^{\mathrm{ord}}_t)_{t\in\T}$ that rejects the blocks $(\tau_{r-1},\tau_r]$, $r\in\N$, where $\tau_r$ is given in \eqref{eq:tau}. 
Formally, $$\phi^{\mathrm{ord}}_t(X_1,\dots,X_t)=\{(\tau_{r-1},\tau_r]: r\in \N,~ \tau_r\le t \},\quad t\in\T.
   $$
\end{definition}

By definition, for any $r_0\in\N$, we have $\tau_r=\infty$ for all $r\ge r_0$ if $\tau_{r_0}=\infty$, and $\tau_r<\infty$ for all $r\le r_0$ if $\tau_{r_0}<\infty$.
For every separate block $r\in\N$ such that $\tau_{r-1}<\infty$, define the \emph{block-normalized process} $M^{[r]}=(M^{[r]}_{s})_{s\in\T}$ as the normalized piece of $\widetilde{M}$: 
\begin{equation*}
    M^{[r]}_0=1,~~M^{[r]}_s=\widetilde M_{\tau_{r-1}+s}~~\mbox{for}~ s\in\T_{\ge 1}~\text{ such that }~s+\tau_{r-1}\in\T.
\end{equation*}

In the ordinary sequential testing, the rejection at time $\tau_r$ is interpreted relative to the entire resetting block $(\tau_{r-1},\tau_r]$. This can be conservative when the entire rejection block is long, where the early part of the block may contribute little to the eventual rejection. To obtain more informative rejections, we trim the beginning of each block so that less informative data points are dropped. Define the following random times:
\begin{equation}
        \sigma_r
        =
        \begin{cases}
        \displaystyle
        \tau_{r-1}
        +
        \max\left\{
        s\in\{0,\ldots,\tau_r-\tau_{r-1}\}:
        M^{[r]}_s
        =
        \min_{0\le v\le\tau_r-\tau_{r-1}}M^{[r]}_v
        \right\},
        &\tau_r<\infty,\\[1ex]
        \tau_{r-1},&\tau_r=\infty.
        \end{cases}
        \label{eq:sigma}
\end{equation}
By definition, the random time $\sigma_r$ satisfies $\tau_{r-1}\le \sigma_r\le \tau_r$ and is $\mathcal F_{\tau_r}$-measurable (but it is not a stopping time).
The random time $\sigma_r$ is the time where the process \emph{attains the within-block minimum for the last time}.
We can interpret the rejection at $\tau_r$ only through the refined interval $(\sigma_r,\tau_r]\subseteq (\tau_{r-1},\tau_r]$.
\begin{definition}
\label{def:refined} 
     The \emph{refined sequential resetting procedure with rejection threshold $\gamma> 1$} is the testing procedure $\phi^{\mathrm{ref}}=(\phi^{\mathrm{ref}}_t)_{t\in\T}$ that rejects the blocks $(\sigma_{r},\tau_r]$, $r\in\N$, where $\tau_r$ and $\sigma_r$ are given in \eqref{eq:tau} and \eqref{eq:sigma}.
Formally,       $$\phi^{\mathrm{ref}}_t(X_1,\dots,X_t)=\{(\sigma_{r},\tau_r]: r\in \N,~ \tau_r\le t \},\quad t\in\T.
   $$
\end{definition}
We also speak of ``{sequential resetting procedures}''  when we do not distinguish ordinary/refined procedures and the rejection threshold $\gamma$ is clear, and we sometimes also omit ``sequential resetting''.
It is straightforward to check that both ordinary and refined procedures satisfy conditions (a) and (b) in Definition~\ref{def:testing-procedure}.

Following Definitions \ref{def:ordinary} and \ref{def:refined}, up to time $t\in\T$, the FDP and FDR of the ordinary and refined procedures are denoted by $\mathrm{FDR}^{\mathrm{ord}}_t,~\mathrm{FDP}^{\mathrm{ord}}_t,~\mathrm{FDR}^{\mathrm{ref}}_t$, and $\mathrm{FDP}^{\mathrm{ref}}_t$ 
as in \eqref{eq:FDP-FDR},\footnote{For example, $\mathrm{FDR}^{\mathrm{ord}}_t=\mathrm{FDR}^{\phi^{\mathrm{ord}}}_t$.} and we also denote the numbers of false rejections by $V_t=V^{\phi^{\mathrm{ord}}}_t$ and $V^\sigma_t=V^{\phi^{\mathrm{ref}}}_t$ as in \eqref{eq:V-R}. Moreover, by definition, the total number of reported intervals coincides for the ordinary and refined procedures, and we denote this quantity by $R_t$, so $R_t=\sum_{r\in\N}\id_{\{\tau_r\le t\}}$. Correspondingly, we let the number of true rejections for both procedures be $K_t=R_t-V_t$ and $K^\sigma_t=R_t-V^\sigma_t$.

\suppressfloats[t]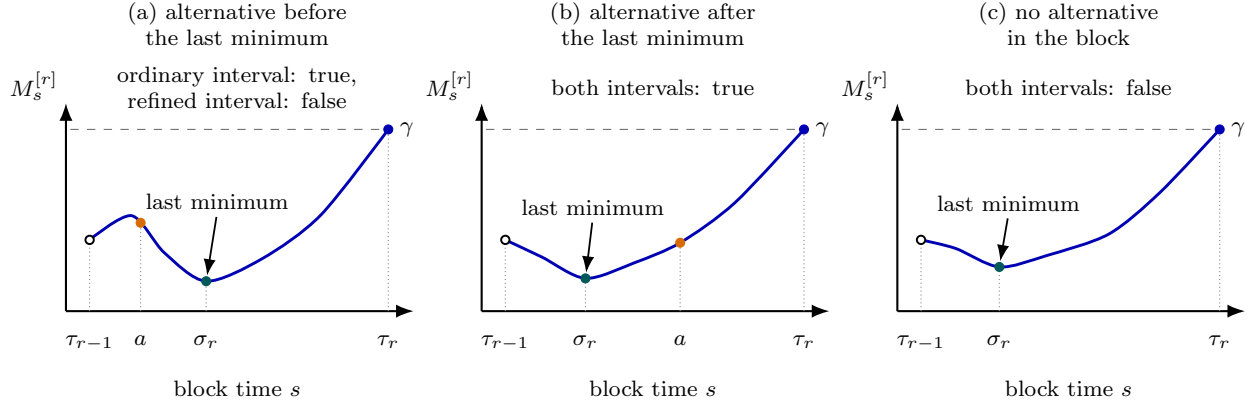
\begin{figure}[t]
\centering
\resizebox{\textwidth}{!}{%
\begin{tikzpicture}[x=0.46cm,y=0.90cm,>=Latex,font=\scriptsize]
\tikzset{
    product/.style={blue!70!black,line width=1.05pt},
    guide/.style={black!45,densely dotted},
    alt/.style={orange!85!black,fill=orange!85!black},
    reset/.style={blue!70!black,fill=blue!70!black},
    sigma/.style={teal!70!black,fill=teal!70!black},
    minarrow/.style={
        ->,draw=black,line width=0.65pt,
        shorten <=1.5pt,shorten >=2pt
    }
}

\foreach \shift/\paneltitle/\caselabel/\casex/\casey in {
    0/{(a) alternative before\\the last minimum}/
      {ordinary interval: true,\\refined interval: false}/4.70/3.14,
    5.25/{(b) alternative after\\the last minimum}/
      {both intervals: true}/4.70/3.14,
    10.50/{(c) no alternative\\in the block}/
      {both intervals: false}/4.70/3.14} {
    \begin{scope}[xshift=\shift cm]
        \node[
            anchor=south,
            align=center,
            text width=4.55cm,
            font=\scriptsize
        ] at (4.70,3.58) {\paneltitle};

        \draw[->,line width=0.75pt] (0,0) -- (9.55,0);
        \draw[->,line width=0.75pt] (0,0) -- (0,2.92);

        \node[anchor=south east,inner sep=1pt] at (-0.12,2.92)
            {\(M^{[r]}_s\)};
        \node at (4.70,-1.08) {block time \(s\)};

        \draw[black!55,dashed]
            (0.12,2.55) -- (8.88,2.55)
            node[right,text=black]{\(\gamma\)};

        \node[
            anchor=center,
            align=center,
            text width=4.20cm,
            inner sep=1pt
        ] at (\casex,\casey) {\caselabel};
    \end{scope}
}

\begin{scope}
\draw[product] plot[smooth] coordinates {
    (0.65,1.00) (1.80,1.34) (2.70,0.82) (3.85,0.42)
    (5.40,0.74) (7.00,1.35) (8.85,2.55)
};

\foreach \x/\y/\lab in {
    0.65/1.00/{\tau_{r-1}},
    2.05/1.24/{a},
    3.85/0.42/{\sigma_r},
    8.85/2.55/{\tau_r}} {
    \draw[guide] (\x,0) -- (\x,\y);
    \node[below=5pt] at (\x,0) {\(\lab\)};
}

\draw[fill=white,line width=0.7pt]
    (0.65,1.00) circle (1.5pt);
\draw[alt]   (2.05,1.24) circle (1.6pt);
\draw[sigma] (3.85,0.42) circle (1.6pt);
\draw[reset] (8.85,2.55) circle (1.6pt);

\node[anchor=south,inner sep=1pt]
    (lastmin-a) at (4.15,1.38) {last minimum};
\draw[minarrow] (lastmin-a.south) -- (3.85,0.42);
\end{scope}

\begin{scope}[xshift=5.25cm]
\draw[product] plot[smooth] coordinates {
    (0.65,1.00) (1.65,0.76) (2.85,0.46) (4.20,0.68)
    (5.45,0.96) (6.95,1.52) (8.85,2.55)
};

\foreach \x/\y/\lab in {
    0.65/1.00/{\tau_{r-1}},
    2.85/0.46/{\sigma_r},
    5.45/0.96/{a},
    8.85/2.55/{\tau_r}} {
    \draw[guide] (\x,0) -- (\x,\y);
    \node[below=5pt] at (\x,0) {\(\lab\)};
}

\draw[fill=white,line width=0.7pt]
    (0.65,1.00) circle (1.5pt);
\draw[sigma] (2.85,0.46) circle (1.6pt);
\draw[alt]   (5.45,0.96) circle (1.6pt);
\draw[reset] (8.85,2.55) circle (1.6pt);

\node[anchor=south,inner sep=1pt]
    (lastmin-b) at (3.05,1.30) {last minimum};
\draw[minarrow] (lastmin-b.south) -- (2.85,0.46);
\end{scope}

\begin{scope}[xshift=10.50cm]
\draw[product] plot[smooth] coordinates {
    (0.65,1.00) (1.60,0.88) (2.80,0.62) (4.10,0.78)
    (5.80,1.08) (7.20,1.66) (8.85,2.55)
};

\foreach \x/\y/\lab in {
    0.65/1.00/{\tau_{r-1}},
    2.80/0.62/{\sigma_r},
    8.85/2.55/{\tau_r}} {
    \draw[guide] (\x,0) -- (\x,\y);
    \node[below=5pt] at (\x,0) {\(\lab\)};
}

\draw[fill=white,line width=0.7pt]
    (0.65,1.00) circle (1.5pt);
\draw[sigma] (2.80,0.62) circle (1.6pt);
\draw[reset] (8.85,2.55) circle (1.6pt);

\node[anchor=south,inner sep=1pt]
    (lastmin-c) at (3.05,1.36) {last minimum};
\draw[minarrow] (lastmin-c.south) -- (2.80,0.62);
\end{scope}

\end{tikzpicture}%
}

\caption{Refined classification inside a rejection block. Positions of the alternatives are denoted by $a$.}
\label{fig:localized-classification}
\end{figure}

Figure \ref{fig:localized-classification} illustrates the difference between ordinary and refined sequential resetting procedures. 
The two versions of the procedures serve different practical purposes.
The refined procedure gives more precise information on the location of the rejected spots. 
The FDR guarantee for ordinary blocks remains useful when the primary interest is in the number and reliability of alarms (instead of their locations in settings such as watermark detection). For example, in quality monitoring or financial backtesting, each rejection usually triggers a review of the entire block, where retaining the full block preserves useful context for investigating a detected problem.
Potentially, with the same threshold, the ordinary procedure may have a smaller FDR than the refined procedure when applied to the same dataset.
Later we will show that they actually share the same provable FDR bounds, and therefore the refined procedures should be recommended in practice.

At a fixed horizon, e-BH and closed e-BH \citep{WR22,XSF25} can control FDR for individual e-values. Our goal is to control FDR for intervals selected by accumulating evidence over sequential observations, resetting after each rejection, and refining the reported blocks. Applying the pointwise procedures to individual e-values does not directly provide this interval-reporting guarantee, and it can have limited power when the individual evidence is modest; Appendix~\ref{subsec:ebh-comparison} illustrates the distinction numerically.%

In our setting, the (ordinary or refined) interval hypothesis for each report depends on the past data, so we do not specify valid e-values a priori for these (data-dependent) hypotheses. As a consequence, we cannot directly apply existing online FDR control techniques \citep{XR24,FXR24} to control the cumulative FDR over the reports. These works also use calibrated rejection thresholds in discovery counts, whereas our threshold is constant over the entire horizon.

For technical purpose, with $\tau_r$ in \eqref{eq:tau}, define
\begin{equation}\label{eq:alter}
    \nu_r=\min\{t\in\T_{\geq \tau_{r-1}+1}:\p_t\notin\mathcal P_t\}\quad\mbox{and}\quad\nu^\sigma_r=\min\{t\in\T_{\geq \sigma_{r}+1}:\p_t\notin\mathcal P_t\},~~r\in\N,
\end{equation}
representing the first time points in the block at which the null is violated for the ordinary and refined procedures, respectively. In particular, if $\tau_r<\infty$, $\nu_r\leq \tau_r$ (or $\nu^\sigma_r\leq \tau_r$) implies that the $r$-th rejection is a true rejection and vice versa. In this sense, we have $$V_t=\sum_{r\in\N}\id_{\{\nu_r>\tau_r,~\tau_r\le t\}}\quad\mbox{and}\quad V^\sigma_t=\sum_{r\in\N}\id_{\{\nu^\sigma_r>\tau_r,~\tau_r\le t\}},~~t\in\T.$$

\section{False discovery rate under   independence  for  a fixed horizon}
\label{sec:FDR indep}

In this section, we examine the false discovery rate bound for sequential resetting procedures under an independence assumption, which will be relaxed in the next section.

 \subsection{FDR bound and examples}

We first state the main assumption used in this section.  
This assumption may not hold in some applications, but it gives a useful benchmark  for our full theory.  
\begin{assumption-c}{IND}
\label{ass:ind}
 The set \(\mathcal H_0=\{u\in\T_{\ge 1}:\p_u\in\mathcal P_u\}\) of null times is deterministic.
Under the data-generating probability $\p$, the e-values $(E_u)_{u\in\H_0}$ are mutually independent and independent of the vector $(E_u)_{u\notin\H_0}$.
\end{assumption-c}

Here, we only set independence requirements for the e-values under the null hypotheses. Apart from nonnegativity and independence from the null vector, no condition is imposed on alternative e-values. 
This assumption parallels the independence condition for the BH procedure \citep{BH95}: the null p-values are mutually independent and independent of the non-null p-values. In our setting, the condition is imposed on the e-values used to construct the resetting process.

\begin{theorem}
\label{thm:localized-main}
Let $\alpha\in(0,1)$. Under Assumption \ref{ass:ind}, the  sequential resetting procedures with rejection threshold $1/\alpha$  satisfy
\begin{equation}
        \FDR^{\mathrm{ord}}_t\le \FDR^{\mathrm{ref}}_t\le \min\left\{1,\frac{\alpha}{1-\alpha}\right\},~~\text{for every deterministic }t\in\T.
        \label{eq:main-bound}
\end{equation}

\end{theorem}

 We first briefly explain the first inequality $\FDR^{\mathrm{ord}}_t\le \FDR^{\mathrm{ref}}_t$, which follows by construction. 
Indeed, the ordinary version has a smaller FDP than the refined version pathwise.  If the
rejection at \(\tau_r\) is true after refinement, then the alternative time $\nu^\sigma_r$ lies in
\((\sigma_r,\tau_r]\).  Because \((\sigma_r,\tau_r]\subseteq
(\tau_{r-1},\tau_r]\), the same ordinary rejection is true in the block $(\tau_{r-1},\tau_r]$. Equivalently,
every false rejection under the ordinary law in the block $(\tau_{r-1},\tau_r]$ is false under the refined law in the block $(\sigma_r,\tau_r]$.  Hence, for every sample path, $V_t\le V^\sigma_t$ for all $t\in\T$, and thus
$
        {V_t} /({R_t\vee1}) 
        \le
 {V^\sigma_t}/({R_t\vee1})$.

For a single hypothesis, a test supermartingale with rejection threshold $1/\alpha$ controls type-I error at level \(\alpha\) by Ville's inequality \citep{V39}:
\[
        \mathbb P\left(\sup_{t\in\T} M_t\ge \frac 1 \alpha \right)\le \alpha,~~~~\alpha \in (0,1).
\]
The FDR bound we obtained in \eqref{eq:main-bound} is higher than the level $\alpha$.
Indeed, one might hope for an upper bound $\alpha$ for 
$\mathrm{FDR}_t^{\rm ord}$ and $\mathrm{FDR}_t^{\rm ref}$. 
Example~\ref{ex:non-iid-alpha-counterexample} gives a concrete counterexample for which \(\FDR^{\mathrm{ord}}_t>\alpha\). 

\begin{example} 
\label{ex:non-iid-alpha-counterexample}
Take \(\alpha=1/5\), \(T=4\), and any $q\in (0,1/5)$. Take labels \((N,N,A,N)\), and independent
e-values given by
\[
\begin{gathered}
E_1=5\ \text{with probability }q,\qquad
E_1=\varepsilon\ \text{with probability }1-q,\\
E_2,E_4=5\ \text{with probability }1/5,\qquad
E_2,E_4=0 \ \text{with probability }4/5,\\
E_3=\infty \quad\text{deterministically},
\end{gathered}
\]
where $\varepsilon>0$ normalizes the mean of $E_1$ to be $1$.  The  FDR for this construction is
\begin{align*}
        \FDR^{\mathrm{ord}}_T
        &=q\left( \frac34\times \frac1{25} +\frac23\times \frac4{25} + 1 \times  \frac45 \right) +(1-q)\left(\frac12\times  \frac1{25}\right) =\frac{11}{12}q + \frac1{50} \to\frac{ 61}{300}  >\alpha,
\end{align*}
as $q\to 1/5$  (see Figure~\ref{fig:ex1} for an illustration). 
In fact, when $q$ is close to $1/5$, a finite $E_3=1/\varepsilon$ is sufficient for the above FDR computation,  but $E_3=\infty$ gives a better intuition for the role of $E_3$.
This example also answers negatively an open question in \cite[Section 5.2.2]{JWZ25}.
\end{example}

The next example shows that the finite-horizon results in Theorem \ref{thm:localized-main} do not hold in either of the following relaxations:
\begin{enumerate}[(i)]
    \item if we replace the finite horizon $t\in\T$ by an arbitrary stopping time $\tau_*$, even if we require that $\tau_*$ is adapted to the total rejection-count filtration;
    \item if we do not impose independence in Assumption \ref{ass:ind}, even if we still consider a deterministic set $\H_0$ and a finite horizon $t\in\T$.
\end{enumerate}
The example further shows that in both cases, the price to pay for the relaxation is a multiplicative factor of $-\log \alpha$. In Section \ref{sec:FDR}, we show that relaxing both conditions still leads to an asymptotically equivalent FDR upper bound as $\alpha\to 0$.

\begin{example}\label{ex:1}
(i) Fix
\(\varepsilon\in(0,1)\), $\alpha\in(0,e^{-1})$, and set $q_\varepsilon
    =
    \frac{1-\varepsilon}{1/\alpha-\varepsilon}
    =
    \frac{\alpha(1-\varepsilon)}{1-\alpha\varepsilon}
    \in(0,\alpha)$. Assume that $\p_{2j-1}\in\mathcal P_{2j-1}$ and $\p_{2j}\not\in\mathcal P_{2j}$ for $j\in\N$. 
At each null time $2j-1$, independently let
\[
        E_{2j-1}
        =
        \begin{cases}
        \frac{1}{\alpha}, & \text{with probability }q_\varepsilon;\\
        \varepsilon, & \text{with probability }1-q_\varepsilon.
        \end{cases}
\]
Then \(\E[E_{2j-1}]=1\).  At each alternative time $2j$, set deterministically
\(
        E_{2j}={\alpha^{-1}}/{\varepsilon}.
\)
Thus, if the null does not reject, the following alternative rejects truly;
if the null rejects, the rejection is false.

Let \(R_t\) be the total number of rejections until time $t$ and
\(\mathcal R_t=\sigma(R_1,\dots,R_t)\).  For $m\geq 1$, define the bounded
\((\mathcal R_t)\)-stopping time
$$\tau_m
        =
        \min\{2j-1:R_{2j-1}=j,\ 1\le j\le m\},$$ 
with \(\tau_m=2m\) if the set is empty. This records the time of the first false rejection. If the first null rejection occurs in
cycle \(j\), then \(j-1\) true rejections have occurred before it, and $\frac{V_{\tau_m}}{R_{\tau_m}\vee1}=\frac1j$. 
Therefore,
\begin{align}
     \lim_{m\to \infty,\,\varepsilon\to 0}   \FDR^{\mathrm{ord}}_{\tau_m}
        =\lim_{m\to \infty,\,\varepsilon\to 0}
        \sum_{j=1}^m
        q_\varepsilon(1-q_\varepsilon)^{j-1}\frac1j=\frac{\alpha}{1-\alpha}\log\frac1\alpha.\label{eq:mee}
\end{align}
  Consequently, for small enough \(\varepsilon\) and
large enough \(m\), $\FDR^{\mathrm{ord}}_{\tau_m}>\frac{\alpha}{1-\alpha}$ because $\alpha<e^{-1}$.

(ii) In the same setup as (i), consider the e-values given by
\[
\bar E_t=
\begin{cases}
E_t,&t\le\tau_m,\\
1,&t>\tau_m.
\end{cases}
\]
It follows that $(\bar E_t)_{t\in\N}$ does not satisfy Assumption \ref{ass:ind}, because the choice of each $\bar E_t$ depends on the past rejection records. Moreover, after time $\tau_m$, there will be no further rejection. Since $\tau_m\leq 2m$, the FDR for the procedure built from $(\bar E_t)_{t\in\N}$ satisfies $\overline{\mathrm{FDR}}^{\mathrm{ord}}_{2m}
=\mathrm{FDR}^{\mathrm{ord}}_{\tau_m}$.  It follows that, for each fixed time horizon $2m$, this construction yields an FDR of $\sum_{j=1}^m
        q_\varepsilon(1-q_\varepsilon)^{j-1}\frac1j$, which has limit of $\frac{\alpha}{1-\alpha}\log\frac1\alpha$ as $m\to\infty$ and $\varepsilon\to 0$ by \eqref{eq:mee}.
\end{example}

 \subsection{Proof sketch for Theorem \ref{thm:localized-main}}

We give a road map for proving the main part of Theorem~\ref{thm:localized-main}, which is the second inequality in \eqref{eq:main-bound}. The full proof can be found in Appendix \ref{app:sec:proof_thm1}. Fix a
terminal time \(t\in\T\).  All objects in the proof are constructed only up to
this fixed horizon \(t\).  Since \(t\) is arbitrary, the final result gives a pointwise bound at each deterministic \(t\in\T\).
The proof of Theorem \ref{thm:localized-main} has three main ideas.

\begin{enumerate}[(i)]
\item \textbf{Conditioning on the alternatives.}
Assumption~\ref{ass:ind} does not impose independence on the alternative times.  We handle this by conditioning
on the alternative e-values up to time \(t\), in which case the null e-values
remain mutually independent.  Thus, by taking mixtures, we may assume that e-values at all times are independent.

\item \textbf{Constructing refined all-null intervals for rejection characterization.}
In order to characterize true and false rejection blocks, we need to identify the positions of alternative time points. To this end, we define
\begin{equation*}
\begin{aligned}
        a(s)=\min\{u\in\{s+1,\ldots,t\}:u\notin\H_0\}~~
        \mbox{and}~~
        b(s)=(a(s)-1)\wedge t,~~s=0,\dots,t.
\end{aligned}
\end{equation*}
Here, $a(s)$ is the first alternative point after $s$, and thus \((s,b(s)]\) is the all-null interval immediately following \(s\).

Now we construct refined intervals containing true and false rejections. Set $\lambda_1=0$. Recursively, after setting a starting point $\lambda_\ell$ for $\ell\in\{1,\ldots,t\}$, write $\beta_\ell=b(\lambda_\ell)$. We construct an all-null interval $(\lambda_\ell,\beta_\ell]$.
Because the rejection blocks are refined by block minima, we set the starting point $\lambda_{\ell+1}$ of the next all-null interval $(\lambda_{\ell+1},\beta_{\ell+1}]$, assuming it exists (later, we say the interval is \emph{genuine} if it exists), as the earlier of the following: (i) the first time point $\theta_\ell$ of a rejection after $\beta_\ell$; (ii) the first time point $U_\ell$ at which the path reaches or falls below its running minimum $L_{\beta_\ell}=\min_{\tau_{R_{\beta_\ell}}\le u\le \beta_\ell}Z_{u}$. Here, we use the running minimum for technical purposes, which dominates the block minimum (see the proof of Lemma \ref{lem:classification}). By such construction, if $\theta_\ell$ is earlier than $U_\ell$, then the rejection at $\lambda_{\ell+1}=\theta_\ell$ is true under the refined law in the interval $(\beta_\ell,\lambda_{\ell+1}]$. On the other hand, if $U_\ell$ is earlier than $\theta_\ell$, then $\lambda_{\ell+1}=U_\ell$ and there are no rejections in the interval $(\beta_\ell,\lambda_{\ell+1}]$, which should be dropped by the refinement.

The all-null intervals $(\lambda_\ell,\beta_\ell]$, $\ell\in\{1,\ldots,t+1\}$, are what we call \emph{clean intervals}.
We record $N_\ell=R_{\beta_\ell}-R_{\lambda_\ell}$ as the total number of rejections in the $\ell$-th clean interval.
The construction is designed so that rejections inside clean intervals are exactly
the refined false rejections, whereas those outside clean intervals are refined true rejections. Lemma~\ref{lem:classification} proves
$V_t^\sigma=\sum_{\ell=1}^{t+1}N_\ell$.

By definition, it is straightforward that
\begin{equation}\label{eq:map1}
    \frac{N_\ell}{R_t\vee1}
        \le
        \frac{N_\ell}{R_{t}-N_\ell+1}.
\end{equation}
It remains to derive the bound of the right-hand side of \eqref{eq:map1} for each $\ell$.

\item \textbf{Obtaining the bound of weighted capital costs.}
Inside every clean interval $(\lambda_\ell,\beta_\ell]$, $\ell\in\{1,\ldots,t+1\}$, all corresponding hypotheses are null. Let $X_\ell=Z_{\lambda_\ell}$ and $Q_\ell=Z_{\beta_\ell}$ be the initial and terminal capital of the clean interval, respectively. A false rejection is costly:
to reject, the resetting process must cross the threshold \(1/\alpha\), and the
post-rejection capital is brought back to \(1\), corresponding to a capital cost of at least $1/\alpha-1$.
Thus, heuristically, the total capital cost of the entire clean interval, $X_\ell-Q_\ell$, should be roughly bounded below by $(1/\alpha-1)N_\ell$.

In light of this intuition, Lemma~\ref{lem:clean-start} shows a rigorous relation that for any
bounded decreasing weight \(w\) known at the start of a clean interval,
\begin{equation}\label{eq:map2}
        \E\left[
        w(Q_\ell)N_\ell
        \,\middle|\,
        \mathcal G_{\lambda_\ell}
        \right]
        \le
        \frac{\alpha}{1-\alpha}
        \E\left[
        w(Q_\ell)(X_\ell-Q_\ell)
        \,\middle|\,
        \mathcal G_{\lambda_\ell}
        \right],
\end{equation}
where $\mathcal G_{\lambda_\ell}=\sigma(E_1,\ldots,E_{\lambda_\ell})$.
Comparing \eqref{eq:map1} and \eqref{eq:map2}, we can choose a weighting function $w$ induced by the denominator on the right-hand side of \eqref{eq:map1}. After checking the boundedness and monotonicity (Lemma \ref{lem:monotone}) of the weighting function we choose, we obtain
\begin{equation}\label{eq:map3}
    \E\left[
        \frac{N_\ell}{R_t\vee1}
        \right]
        \le
        \E\left[\frac{N_\ell}{R_{t}-N_\ell+1}\right]
        \le
        \frac{\alpha}{1-\alpha}
        \E\left[
        \frac{X_\ell-Q_\ell}{R_t-N_\ell+1}
        \right].
\end{equation}

Finally, it remains to calculate the bound of the weighted capital cost on the right-hand side of \eqref{eq:map3}.
We use a telescoping relation between consecutive genuine
clean intervals.  When $U_\ell<\theta_\ell$, Lemma~\ref{lem:epoch-drop} gives $X_{\ell+1}\le Q_\ell$. In other words, the ending capital of one clean interval controls the starting capital
of the next. Consider the weight $\gamma_\ell=(R_t-N_\ell+1)^{-1}$. If $N_\ell=0$, then $\gamma_\ell=(R_t+1)^{-1}$; if $N_\ell>0$, then $\gamma_\ell\le N_\ell/R_t$. Hence, Lemma~\ref{lem:charge} proves in this sense that
\begin{equation}\label{eq:map4}
    \sum_{\ell=1}^{t+1} \frac{X_\ell-Q_\ell}{R_t-N_\ell+1}=\sum_{\ell=1}^{t+1} \gamma_\ell(X_\ell-Q_\ell) \le \begin{cases}
        \frac{K^\sigma_t}{R_t+1}+\frac{V^\sigma_t}{R_t} \le 1,~~V^\sigma_t>0,\\
        \frac{R_t+1}{R_t+1}=1,~~V^\sigma_t=0.
    \end{cases}
\end{equation}
Summing \eqref{eq:map3} over $\ell=1,\dots,t+1$ and using \eqref{eq:map4}, we prove that $\mathrm{FDR}^{\mathrm{ref}}_t\le \alpha/(1-\alpha)\wedge 1$. \qedhere
\end{enumerate}

\section{Online false discovery rate control without independence}
\label{sec:FDR}

We now turn to online FDR bounds without the independence assumption (Assumption \ref{ass:ind}). 
Thus, there are three upgrades: anytime validity, allowing dependence, and allowing adapted null time spots. It turns out that these three upgrades together come at a price of a factor of the scale $-\log\alpha$.
As shown in Example \ref{ex:1}, either anytime validity or allowing dependence already pays the price of a factor of the same scale. We do not examine the separate price for allowing adapted null spots because data-dependent null spots can induce dependence between null and alternative e-values.

Allowing dependence between the null e-values and non-null e-values is particularly useful when the e-values are sequentially computed to form an e-process (see the applications in Sections \ref{sec:llm-experiment} and \ref{sec:backtest}). In such applications, the null e-values are typically computed from a betting strategy that depends on the values of the non-null e-values \citep[Chapter 7]{RW25};  note that it is unknown to the tester which e-values are non-null, so avoiding dependence is impossible in such constructions. On the other hand, the conditional validity of e-values in \eqref{eq:evalue} is unaffected by choosing betting strategies.

\subsection{Anytime-valid FDR control}

Define the \emph{uniform} false discovery rate for the ordinary and refined sequential resetting procedures respectively as
\begin{equation}
\label{eq:FDR_infty}
    \mathrm{FDR}^{\mathrm{ord}}_\infty=\E\left[\sup_{t\in\T}\mathrm{FDP}^{\mathrm{ord}}_t\right]\qquad\mbox{and}\qquad\mathrm{FDR}^{\mathrm{ref}}_\infty=\E\left[\sup_{t\in\T}\mathrm{FDP}^{\mathrm{ref}}_t\right].
\end{equation}
Note that the quantities $\mathrm{FDR}^{\mathrm{ord}}_\infty$ and $\mathrm{FDR}^{\mathrm{ref}}_\infty$ are stronger than the usual online FDR notions in \eqref{eq:FDP-FDR} due to the suprema over $\T$ inside the expectations; that is, they are the expected worst FDPs over the entire horizon. This is equivalent to the SupFDR introduced in \citet{XR22} and studied in \citet{FXR24}.

The following result shows the FDR bound under the general setup. A simpler FDR bound for a modified procedure with a deterministic number of rejections is demonstrated in Appendix~\ref{app:sec:r}.

\begin{theorem}\label{thm:FDR_barrier}
  Let $\alpha\in(0,1)$. The sequential resetting procedures with rejection threshold  $1/\alpha$  satisfy
    \begin{equation}
    \label{eq:local_FDR_bound}
        \mathrm{FDR}^{\mathrm{ord}}_\infty \le \mathrm{FDR}^{\mathrm{ref}}_\infty\le 1-(1-\alpha)\E\bigg[\inf_{k\in\N}\frac{k}{k+S_k}\bigg]\le \alpha\bigg(1+\log{\frac{1}{\alpha}}\bigg),
    \end{equation}
  where  $(S_k=\sum_{j=1}^k\xi_j)_{k\in\N}$ is the partial sum process of i.i.d.~geometric random variables $\{\xi_j\}_{j\in\N}$ with parameter $1-\alpha$, i.e., $\p(\xi_j=n)=\alpha^n(1-\alpha)$, $n\in\N_0$.
\end{theorem}

Anytime validity concerns FDR control under stopping times. 
Let \(\rho\) be an almost surely finite
\((\mathcal F_t)\)-stopping time taking values in \(\mathbb T\). Pathwise, we have $$\mathrm{FDP}^{\mathrm{ref}}_\rho\le\sup_{t\in\T}\mathrm{FDP}^{\mathrm{ref}}_t~~\mbox{and}~~\mathrm{FDP}^{\mathrm{ord}}_\rho\le\sup_{t\in\T}\mathrm{FDP}^{\mathrm{ord}}_t.$$
Hence, $\FDR_\rho^{\mathrm{ref}}\le\FDR_\infty^{\mathrm{ref}}$,
$\FDR^{\mathrm{ord}}_\rho\le\FDR^{\mathrm{ord}}_\infty$, and thus the bounds in Theorem~\ref{thm:FDR_barrier}  remain valid when the deterministic evaluation time
\(t\) is replaced by \(\rho\).
This, in particular, holds true for $\rho$ being a fixed $t\in \T$.

\begin{remark}\label{rem:Su18}
   The explicit bound $\alpha(1+\log(1/\alpha))$ in Theorem \ref{thm:FDR_barrier}  can be viewed as a sequential-testing analogue of the FDR-linking theorem \citep[Theorem 1]{S18}, which connects the FDR of the full multiple-testing problem under arbitrary dependence among the p-values to a null-only FDR quantity. Consequently, if the null p-values satisfy the PRDN condition of \cite{S18} (formulated using the PRDS condition of \cite{BY01}), while allowing arbitrary dependence between the null and non-null p-values, Theorem 3 of \cite{S18} shows that the BH procedure satisfies the general bound $\alpha(1+\log(1/\alpha))$.  For online BH, Theorem~3.4 of \cite{FXR24} establishes the same bound under PRDN for SupFDR, which also appears in our Theorem~\ref{thm:FDR_barrier}.
      Our setting has a similar structure: the general uniform FDR bound of Theorem \ref{thm:FDR_barrier} is built from a refined version of Ville’s inequality applied to a single null block, which is then extended to geometric control of consecutive false rejections. In these results, the argument starts from a null-only error bound and accounts for all possible rejection levels, and summing over those levels produces the logarithmic factor. The two results are therefore closely related but neither implies the other.
\end{remark}

For a fixed finite horizon $T\in\N$, we can get a slightly better bound when $T$ is small. 
Write $h_T=\sum_{r=1}^T (1/r)$, which is the $T$-th harmonic number. 
Under the assumptions of Theorem~\ref{thm:FDR_barrier}, we have
\begin{align}
    \E\left[\max_{t\in \T}\mathrm{FDP}^{\mathrm{ref}}_t\right]
\le
\min\left\{
\alpha h_T,\,
\alpha\left(1+\log\frac1\alpha\right)
\right\}.\label{eq:22}
\end{align}
The same bound applies to the FDP for the ordinary sequential resetting procedure and to FDR at any
stopping time bounded by $T$.   
To see this, since $r\le R_t$ for every report before time $t$,
\[
\frac{V_t^\sigma}{R_t\vee1}
=
\sum_{r=1}^{R_t}
\frac{\id_{A_r^\sigma(T)}}{R_t\vee1}
\le
\sum_{r=1}^{T}\frac{\id_{A_r^\sigma(T)}}{r},
\qquad t\in \T.
\]
By taking the maximum over $t$ and then expectations and using
Lemma~\ref{lem:ville} in the appendices, we get the upper bound $\alpha h_T$ in \eqref{eq:22}.
The other upper bound of \eqref{eq:22} follows from Theorem~\ref{thm:FDR_barrier}.
The constant $h_T$ is precisely the multiplicative correction factor in   \cite{BY01}, which corrects the Simes p-value \citep{S86} and the BH procedure under arbitrary dependence.  
Therefore, our procedure pays a multiplicative price at the minimum  of $h_T$ and $1-\log \alpha$, 
which equals the price that BH pays under Su's condition, when $T$ is the total number of hypotheses  (because the factor $h_T$ is valid in case of arbitrary dependence).

In Figure~\ref{fig:benchmark-mc}, we compare the benchmark quantity $1-(1-\alpha)\E[\inf_{k\in\N}\frac{k}{k+S_k}]$ with the fully explicit upper bound $\alpha(1+\log{\frac{1}{\alpha}})$ numerically. Figure~\ref{fig:benchmark-mc} suggests that the two quantities are asymptotically equivalent as $\alpha\downarrow 0$, a fact justified by Lemma \ref{lemma:E bound} below.

\begin{figure}[t]
\centering
\includegraphics[width=0.72\textwidth]{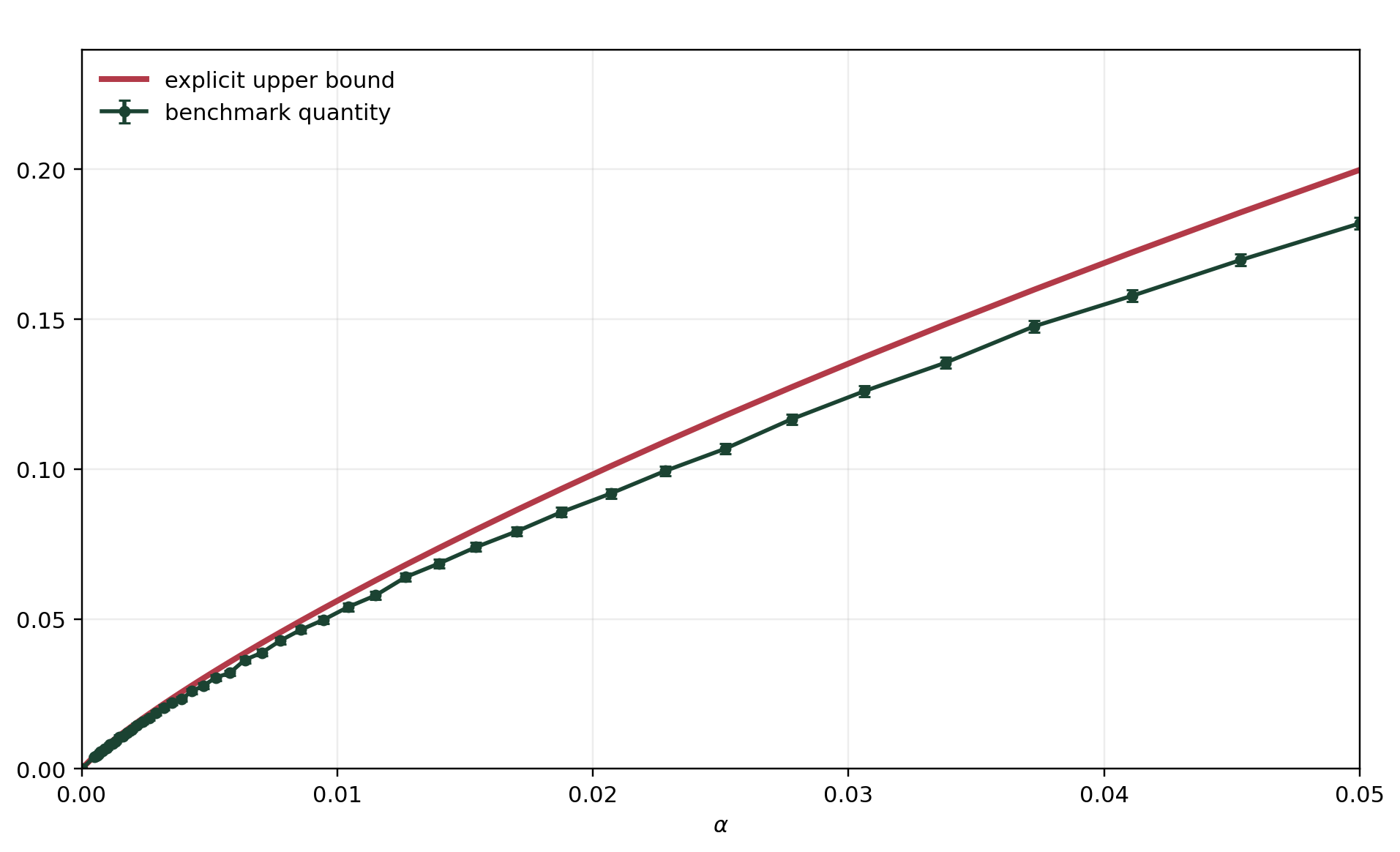}
\caption{Comparison of a Monte Carlo proxy of the benchmark quantity $1-(1-\alpha)\E[\inf_{k\in\N}\frac{k}{k+S_k}]$ (green curve) with the explicit upper bound ${\alpha}(1+\log\frac{1}{\alpha})$ (red curve) for $\alpha\in(0,0.05]$.}

\label{fig:benchmark-mc}
\end{figure}

The first upper bound on    both $\mathrm{FDR}^{\mathrm{ord}}_\infty$
 and $\mathrm{FDR}^{\mathrm{ref}}_\infty$ in  \eqref{eq:local_FDR_bound}, given by
\begin{align}
 1-(1-\alpha)\E\bigg[\inf_{k\in\N}\frac{k}{k+S_k}\bigg],\label{eq:UB2}
\end{align}
   is generally tight for both procedures, as the following example shows.

\begin{example}\label{ex:infty}
    Fix $\varepsilon\in(0,1)$ and set $q_\varepsilon
    =
    \frac{1-\varepsilon}{1/\alpha-\varepsilon}
    =
    \frac{\alpha(1-\varepsilon)}{1-\alpha\varepsilon}
    \in(0,\alpha)$. 
    For each block $r\in\N$, let $M^{[r]}_0=1$ and let $\p_{\tau_{r-1}+1}$ be the null at time $\tau_{r-1}+1$ (i.e., $\p_{\tau_{r-1}+1}\in\mathcal P_{\tau_{r-1}+1}$), under which $\p(M^{[r]}_1=1/\alpha)=1-\p(M^{[r]}_1=\varepsilon)=q_\varepsilon$.   On the event
    $\{M^{[r]}_1=1/\alpha\}$, the stopping time $\tau_r$ is triggered and we move to the next block $r+1$, starting again from the null at time $\tau_r+1$. On the event $\{M^{[r]}_1=\varepsilon\}$, we let $\p_{\tau_{r-1}+2}$ be the alternative at time $\tau_{r-1}+2$ (i.e., $\p_{\tau_{r-1}+2}\notin\mathcal P_{\tau_{r-1}+2}$) and define $M^{[r]}_2=1/\alpha$, which then triggers the stopping time $\tau_r$. The construction is repeated independently across blocks. In this case, each block has exactly one rejection, false with probability $q_\varepsilon$ and true with probability $1-q_\varepsilon$. It follows that the numbers of false rejections between two adjacent true rejections are i.i.d.\ distributed as $\mathrm{Geom}(1-q_\varepsilon)$. Hence, by the same argument as in \eqref{eq:long} in Appendix \ref{app:sec:proof_thm2},
    \[
    \mathrm{FDR}^{\mathrm{ord}}_\infty
    =
    1-(1-q_\varepsilon)\E\bigg[\inf_{k\in\N}\frac{k}{k+S_k^{(\varepsilon)}}\bigg],
    \]
    where $S_k^{(\varepsilon)}=\sum_{j=1}^k\xi_j^{(\varepsilon)}$ and $\xi_j^{(\varepsilon)}$, $j\in\N$, are i.i.d.\ geometric random variables with parameter $1-q_\varepsilon$. Since $q_\varepsilon\uparrow\alpha$ as $\varepsilon\downarrow0$, this shows that the upper bound in \eqref{eq:UB2} is tight. 
\end{example}

Theorem \ref{thm:FDR_barrier}  applies to   test supermartingales. One may wonder whether they also work for general e-processes. Our next example shows that this is not the case.

\begin{example}\label{ex:e-counterexample}
   Let $\mathcal P_1=\mathcal P_2$ be the set of positive laws with mean $\leq 1$, so $\mathfrak{P}_0$ contains laws with $\E[X_1]\leq 1$ and $\E[X_2\mid X_1]\leq 1$. Consider the e-process $(M_0,M_1,M_2)=(1,X_1/10,X_1)$ (which is an e-process since it is dominated by the test supermartingale $(1,X_1,X_1)$). If we choose $X_1=100,\,X_2=1$, then $H_1$ is false and $H_2$ is true. With a threshold $\gamma=10$, our resetting procedure rejects at both times so the FDR is $1/2$, which exceeds the upper bound in \eqref{eq:local_FDR_bound}. In general, for any $\gamma>1$ and $T\in\N$, letting $M_t=\gamma^{t-T}X_1,\,1\leq t\leq T$ where $X_1=\gamma^T,\, X_2=\cdots=X_T=1$ gives an e-process with FDR equal to  $(T-1)/T$.
\end{example}

\subsection{A confidence bound on the number of alternatives and marginal FDR}
As a byproduct of the arguments used to prove Theorem \ref{thm:FDR_barrier},  for a fixed $T$, we provide a connection between the number of rejections $R_T$ and the cardinality $N_a$ of the set of alternatives $A_T=\{u\in\T_{\ge 1}:\p_u\notin\mathcal P_u\}$, as well as a marginal FDR bound.\footnote{Here, we use the adjusted version  $\E[V_T]/(1+\E[R_T])$ for  the marginal FDR; see \cite[Section 2.3]{BH95}.} 

\begin{theorem}\label{cor:rejection-count-and-mfdr}
Let $\alpha\in(0,1)$. For the sequential resetting procedure with rejection threshold $1/\alpha$, the following statements hold.
\begin{enumerate}[(i)]
\item  
    Assume that $A_T$ is deterministic and write $N_a=|A_T|$.
    For every integer $r\ge N_a$,
    \begin{equation}\label{eq:rejection-count-dom} 
        \mathbb{P}(R_T\ge r)
        \le
        \sum_{j=0}^{N_a}
        \binom{r}{j}(1-\alpha)^j\alpha^{r-j}.
    \end{equation}
    In particular,
    \begin{equation}\label{eq:expectation-dom}
        \mathbb{E}[R_T]
        \le \frac{N_a+\alpha}{1-\alpha}.
    \end{equation}
Moreover, if $N_a<\infty$, then for $\beta\in(0,1)$,
\begin{equation}\label{eq:confidence-dom}
    \p\left(
        R_T\le
        \frac{
            N_a
            +\sqrt{2N_a\alpha\log(1/\beta)}
            +\frac{1+\alpha}{2}\log(1/\beta)
        }{1-\alpha}
    \right)
    \ge 1-\beta.
\end{equation}
For $N_a>0$, this upper bound equals
$\delta_T N_a/(1-\alpha)$, where
\[
    \delta_T
    =
    1+\sqrt{\frac{2\alpha\log(1/\beta)}{N_a}}
    +\frac{(1+\alpha)\log(1/\beta)}{2N_a}
    =
    1+O(N_a^{-1/2})\to 1
    \qquad\text{as }N_a\to\infty.
\]
 
\item For $T<\infty$,
\begin{equation}\label{eq:mfdr-one-bound}
    \frac{\E[V_T]}{1+\E[R_T]}
    \le
    \frac{\E[V_T^\sigma]}{1+\E[R_T]}
    \le\alpha.
\end{equation}
\end{enumerate}
\end{theorem}

The proof of Theorem \ref{cor:rejection-count-and-mfdr} is put in Appendix \ref{app:sec:proof_thm2}.
Theorem \ref{cor:rejection-count-and-mfdr} gives a clear relation between the number of non-null e-values and 
the rejected intervals.  
With a suitable adjustment  factor that is asymptotically $1/(1-\alpha)$, the number
of rejected intervals $R_T$ gives a lower confidence bound on
the number of non-null e-values. 
For instance, if $\alpha=0.1$ and the deterministic set of
non-null times up to $T$ has cardinality $10$, then $
    \p(R_T>13)
    < 0.05
 $ by \eqref{eq:rejection-count-dom}.
Thus, with probability at least $95\%$, the procedure rejects
at most $13$ intervals, regardless of the distributions of the e-values.

 \subsection{Proof sketch for Theorem \ref{thm:FDR_barrier}}

The full proof of Theorem \ref{thm:FDR_barrier} is in Appendix \ref{app:sec:proof_thm2}. We now explain the main idea behind the proof. It suffices to consider $\FDR^{\rm ref}_\infty$.

Unlike
Section~\ref{sec:FDR indep}, no independence assumption is imposed here.  The
proof therefore cannot decompose the path into independent clean intervals.
Instead, we use a more robust argument: we control the number of consecutive
refined false rejections between two adjacent refined true rejections. The proof proceeds in four steps.

\begin{enumerate}
\item \textbf{Reducing the FDR problem to counting false rejections between true
ones.}
Let
\[
        \kappa_j^\sigma=\min\{s\in\T:K_s^\sigma=j\},
        \qquad j\in\N_0,
\]
be the time of the \(j\)-th refined true rejection.  Between
\(\kappa_{j-1}^\sigma\) and \(\kappa_j^\sigma\), we define
\(N_j^\sigma(t)\) as the number of refined false rejections accumulated up
to time \(t\).  Thus, the whole sequence
$N_1^\sigma(t),N_2^\sigma(t),\ldots$ records the false-rejection clusters separated by refined true rejections.
The FDR can then be bounded by a function of this sequence.

\item \textbf{Controlling a single false-rejection cluster by a geometric tail.}
The key probabilistic estimate is that each \(N_j^\sigma(t)\) is stochastically dominated by
a $(1-\alpha)$-geometric random variable:
\begin{align}
    \P\left(
        N_j^\sigma(t)\ge n
        \,\middle|\,
        \mathcal F_{\kappa_{j-1}^\sigma\wedge t}
        \right)
        \le
        \alpha^n,
        \qquad n\in\N_0,\label{eq:step2}
\end{align}
which is proved by Lemma~\ref{lem:ville_n}. This is obtained by applying the following result recursively for $n$:
\begin{align}
    \P\left(A_r^\sigma(t)\,\middle|\,\mathcal F_{\tau_{r-1}}\right)
        \le \alpha,~~A_r^\sigma(t)=\{\tau_r\le t,\ \nu_r^\sigma>\tau_r\},~~r\in\N,\label{eq:Arto}
\end{align}
which can be regarded as a refined version of Ville's inequality. Here, $A_r^\sigma(t)$ stands for the event that the \(r\)-th rejection has happened by time \(t\), and that rejection is false under the refined law. The proof of \eqref{eq:Arto} (see Lemma \ref{lem:ville}) is nontrivial because the block minimum is not obtained at a stopping time.

\item \textbf{Joint stochastic dominance of false-rejection clusters by geometric benchmarks.}
A marginal geometric tail bound is not enough by itself, because the
FDR depends on the whole vector of false-rejection counts.  Lemma
\ref{lem:mono_com} upgrades \eqref{eq:step2} to a multidimensional stochastic
domination: for every bounded increasing function
\(\Phi:\N_0^m\to\R\),
\begin{align}
     \E\left[
        \Phi\bigl(N_1^\sigma(t),\ldots,N_m^\sigma(t)\bigr)
        \right]
        \le
        \E\left[
        \Phi(\xi_1,\ldots,\xi_m)
        \right],\label{eq:dom}
\end{align}
where \(\xi_1,\ldots,\xi_m\) are i.i.d.~geometric random variables satisfying
$\P(\xi_j=n)=\alpha^n(1-\alpha)$, $n\in\N_0$.
The proof is an induction on \(m\), where at each induction step, Lemma~\ref{lem:ville_n}
controls the next false-rejection cluster conditionally on the past.

\item \textbf{Evaluating the i.i.d.~geometric benchmark.}
It remains to evaluate a bound on an i.i.d.~geometric benchmark.
The benchmark FDP has a simple interpretation. If \(\xi_1>0\), then false rejections occur before the first refined true rejection, and the worst benchmark FDP is equal to \(1\). If \(\xi_1=0\), then the first rejection is true under the refined law. After \(k\) refined true rejections, the benchmark has \(k\) true rejections and $S_k=\sum_{j=1}^k \xi_j$ false rejections. Hence, the worst benchmark FDP is given by
$1-\inf_{k\in\N}(k/(k+S_k))$.
Taking into account the initial event \(\{\xi_1=0\}\), whose probability is \(1-\alpha\), gives the benchmark bound
\[ 1-(1-\alpha) \E\left[ \inf_{k\in\N}\frac{k}{k+S_k} \right]. \]
It remains to make this quantity explicit.
Lemma~\ref{lemma:E bound} does this by defining $Y=\sup_{k\in\N}S_k/k$. It follows from the maximal inequality in \cite{S15} that
\[
    \P(Y\ge x)
    \le
    \min\left\{1,\frac{\alpha}{(1-\alpha)x}\right\},~~x>0.
\]
This yields a bound on the benchmark
\[ 1-(1-\alpha) \E\left[ \inf_{k\in\N}\frac{k}{k+S_k} \right]=\alpha+(1-\alpha) \int_0^1
    \P\left(Y\ge \frac{u}{1-u}\right)\d u \le \alpha\left(1+\log\frac1\alpha\right). \]
Combining this i.i.d.~benchmark calculation with the multidimensional stochastic domination \eqref{eq:dom} proves Theorem~\ref{thm:FDR_barrier}.\qedhere
\end{enumerate}

\section{Simulation and data experiments}
\label{sec:num}

In this section, we show simulation studies and real data experiments to illustrate the application of sequential resetting procedures. 
We introduce several metrics that will be used to evaluate the performance of the sequential resetting procedures. For a finite time horizon $T\in\N$, let
\[
A_T=\{u\in\T_{\ge 1}:\p_u\notin\mathcal P_u\}
\]
be the set of alternative times. The time points covered by the refined
rejection intervals are
\[
\widehat A_T^\sigma
=
\bigcup_{r:\tau_r\le T}\bigl((\sigma_r,\tau_r]\cap\T_{\ge 1}\bigr).
\]
Define the metrics 
\begin{align}
    \operatorname{Power}_T
 =\frac{|A_T\cap\widehat A_T^\sigma|}{|A_T|\vee 1},\qquad
\operatorname{FDP}^{\mathrm{cov}}_T
 =\frac{|\widehat A_T^\sigma\setminus A_T|}
        {|\widehat A_T^\sigma|\vee1},
\qquad\text{and}\qquad
\operatorname{IoU}_T
 =\frac{|A_T\cap\widehat A_T^\sigma|}
        {|A_T\cup\widehat A_T^\sigma|\vee 1}.\label{eq:metrics}
\end{align}
Thus, the \textit{power} is the proportion of alternative times covered by the reported
intervals, while the (realized) \textit{coverage FDP} is the proportion of covered times that are
null.\footnote{The coverage FDP differs from the (realized) FDP defined in \eqref{eq:FDP-FDR}. The former is time-wise (or equivalently token-wise in the watermark detection setting) and the latter is rejection-wise. Roughly speaking, a rejected interval with more nulls than alternatives may contribute to a high coverage FDP but a low FDP (see e.g., Table \ref{tab:iid_power} below).} IoU, also called the Jaccard index, stands for \textit{intersection over union} and penalizes both missed alternative
times and rejected null times. We also recall the realized FDP from \eqref{eq:FDP-FDR}. Taking the supremum over $t\in\T$ therein leads to the uniform realized FDP.

\subsection{A simulation study under the independence setting}
\label{sec:sim_ind}

We first study the performance of sequential resetting procedures under independent simulated data. For a fixed time horizon $T=1500$, we generate independent normally distributed data $X_t\lawis\mathrm N(\mu_t,1)$ for $t\in\T_{\ge 1}$, where $\mu_t$ are deterministic nonnegative real numbers. Let $E_t=\exp(X_t-1/2)$ and construct the following process with $M_0=1$ and
\begin{equation}\label{eq:iid_mtg}
    M_t=(1-\lambda+\lambda E_t)M_{t-1}=\prod^t_{s=1}(1-\lambda+\lambda E_s),~~t\in\T_{\ge 1},
\end{equation}
where we use a constant $\lambda=1/2$ to ensure independence among all e-values $1-\lambda+\lambda E_s$.
Let $Q^0=\mathrm N(0,1)$ be the corresponding probability measure of a standard normal distribution.
It follows that $(M_t)_{t\in\T}$ is a test supermartingale relative to $(\mathcal P^0_t)_{t\in\T_{\ge 1}}$, where $\mathcal P^0_t=\{Q^0\}$ with the corresponding null hypothesis $H_t=\{\p_t\in\mathcal P^0_t\}=\{\mu_t=0\}$. Moreover, $\E[E_t]>1$ for each $t\in\T_{\ge 1}$ such that $\mu_t>0$.

Because Assumption \ref{ass:ind} holds, Theorem \ref{thm:localized-main} yields that the FDR bound with the rejection threshold $\gamma=1/\alpha$ for $\alpha\in(0,1)$ is given by $\min\{1,\alpha/(1-\alpha)\}$. Consequently, we can choose the rejection threshold $\gamma=11$ to guarantee an FDR bound of $0.1$. We conduct the following two experiments. Each experiment is repeated 1000 times.
\begin{enumerate}
    \item \textbf{Alternating null and alternative periods:} Take
$$\mu_t=\id_{\{100(3k-1)+1\le t\le 300k,~k\in\N\}},~~t\in\T_{\ge 1}.$$
Consequently, the normally distributed data have mean $0$ for $200$ points (i.e., null hypotheses hold), then mean $1$ for $100$ points (i.e., alternative hypotheses hold), and so forth. Figure \ref{fig:reg} displays the alternating pattern.

\begin{figure}
    \centering
    \begin{center}\includegraphics[width=0.9\linewidth]{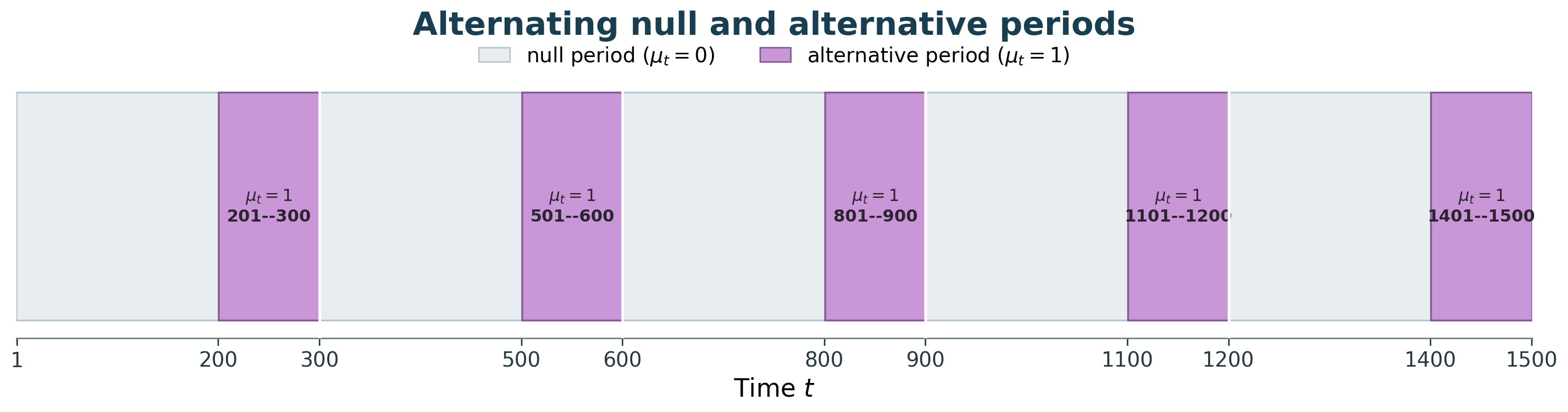}
    \end{center}
    \caption{The alternating pattern between null and alternative periods.}
    \label{fig:reg}
\end{figure}

Figure \ref{fig:iid_recurring} shows a representative log-transformed resetting process $\widetilde M$ corresponding to the test supermartingale \eqref{eq:iid_mtg}. The result shows that all the alternative periods ($\mu_t=1$, marked as green) trigger refined rejections of the resetting process (marked as red) --- although there is delay of the rejections after the beginning of the alternative periods, the first point of the rejection block still captures the starting point of the alternative periods, where the resetting process is near its block minimum in the block. Most of the data in the null periods ($\mu_t=0$) are dropped by the refinement rule.

\begin{figure}[t]
\centering
\includegraphics[width=0.96\textwidth]{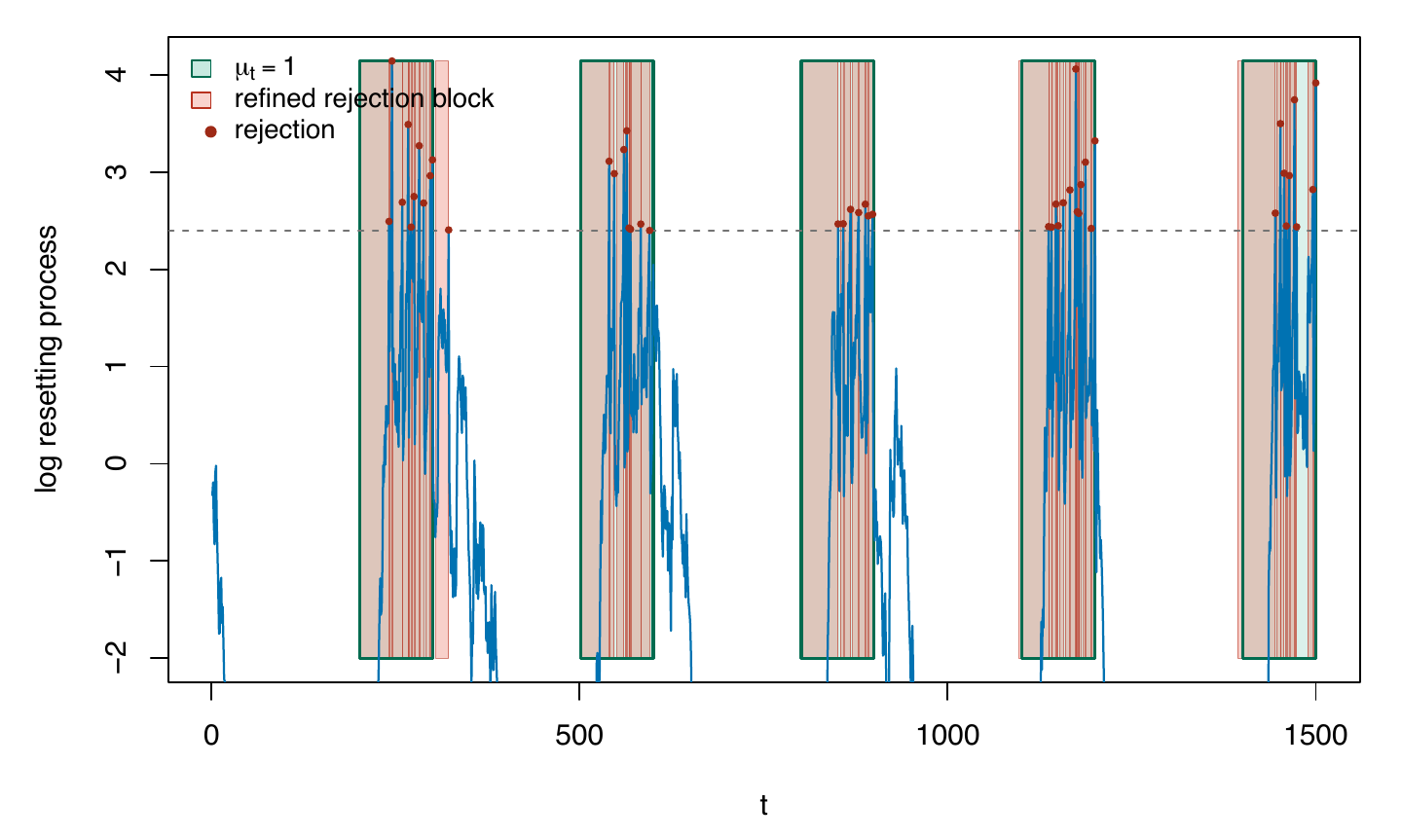}
\caption{Representative log-transformed resetting test supermartingale under independent normally distributed data with alternating null and alternative periods.}
\label{fig:iid_recurring}
\end{figure}

\item \textbf{Detecting single-point shocks:} Instead of studying the alternating periods, we generate independently normally distributed data with several significant single-point  shocks. Namely, take
$$\mu_t=50\id_{\{t=300k,~k\in\N\}},~~t\in\T_{\ge 1}.$$
The data are thus standard normal but change to $\mathrm{N}(50,1)$ every $300$ points.

Figure \ref{fig:iid_shocks} plots the corresponding log-transformed resetting process $\widetilde M$. Different from the previous case, we obtain narrow refined rejection blocks (marked in red) near each point corresponding to the shock with $\mu_t=50$ (marked in green). This shows that a sequential resetting procedure can also be used to detect instantaneous significant shocks.

\begin{figure}[t]
\centering
\includegraphics[width=0.96\textwidth]{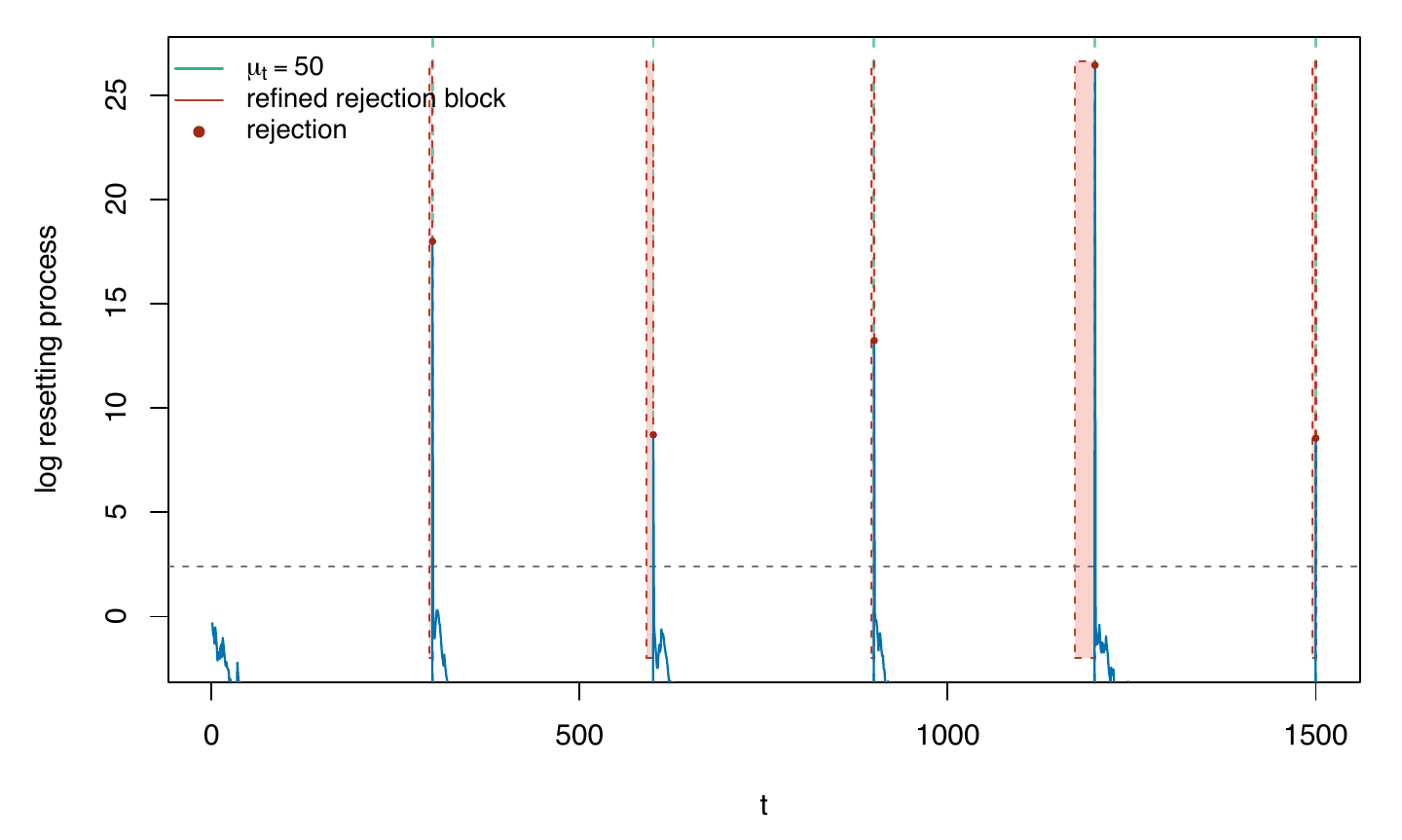}
\caption{Representative log-transformed resetting test supermartingale under independent normally distributed data with single-point shocks.}
\label{fig:iid_shocks}
\end{figure}
\end{enumerate}

Table \ref{tab:iid_power} presents the detailed average performance metrics of the two experiments above over the 1000 iterations. Both cases yield high power, showing that most of the alternative points are covered by the refined rejection blocks. The case of single-point shocks exhibits high coverage FDP and low IoU because the alternative time points are too sparse, which is acceptable given that the refined rejection blocks are narrow in Figure \ref{fig:iid_shocks}.
\begin{table}[t]
    \centering
\caption{Average performance metrics, over 1000 iterations, of the simulation studies of independent normal data with alternating periods and single-point shocks.}
\label{tab:iid_power}
    \begin{tabular}{c c c c c}
        \toprule
        Case & Power & Coverage FDP & IoU & Realized FDP\\\midrule
        Alternating periods & 0.9035 & 0.1128 & 0.8085 & 0.0083\\
        Single-point shocks & 0.9922 & 0.8377 & 0.1617 & 0.0483\\
        \bottomrule
    \end{tabular}
\end{table}

\subsection{An application to LLM watermark detection}
\label{sec:llm-experiment}

In this section, we discuss an application of our sequential resetting procedure to LLM watermark detection with e-processes \citep{SWZ26}. Complementary results are also discussed in Appendix \ref{sec:add_num}. 

\subsubsection{A mathematical formulation of LLM watermark}

A large language model (LLM) receives a prompt, represented by a sequence of tokens, and generates a continuation one token at a time.  At each token position, the prompt and the
preceding tokens together determine a probability distribution over the set of  vocabulary $\mathcal W$.\footnote{In this section, each token takes up one position $t\in\T_{\ge 1}$.} Let $W_t$ be the token generated at
position $t$.  Given the prompt and all previously generated tokens, the
model produces logits $z_t(w),\,w\in\mathcal W$.  At temperature $\vartheta>0$, its probability
of sampling token $w$ is given by
\[
 \pi_t(w)
 =\frac{\exp\left({z_t(w)}/{\vartheta}\right)}
 {\sum_{v\in\mathcal W}\exp\left({z_t(v)}/{\vartheta}\right)},
 \qquad w\in\mathcal W.
\]

An LLM watermark uses a secret pseudorandom key to change the coupling between $W_t$ and the key while preserving the conditional law $W_t\lawis\pi_t$. A detector that knows the key can test for this dependence. 
In our application, we use the Gumbel--max watermark and the tournament watermark, but our method is not tied
to these particular settings. Following \cite{SWZ26}, it can be applied to any
watermarking scheme that produces a token-level score with a known distribution
when no watermark is present. Such scores can be converted into e-values and
used directly in our resetting and refinement procedures.

\paragraph{Gumbel--max watermark.}
The Gumbel--max watermark construction was proposed by \cite{Aaronson23} and analyzed statistically by \cite{LRWLS25}. At every position, the secret key samples a random number
$U_t(w)$, uniformly distributed on $[0,1]$, for each token $w\in\mathcal W$. We assume that $\{U_t(w):t\in\T_{\ge 1},w\in\mathcal W\}$ are mutually independent and independent of the randomness used for standard sampling.
Standard token generation samples directly from $\pi_t$ and does not depend on $\{U_t(w)\}_{w\in\mathcal W}$.
On the other hand, watermarked token generation uses
\begin{equation}
 W_t=\argmax_{w\in\mathcal W}\frac{\log U_t(w)}{\pi_t(w)}.
 \label{eq:llm-gumbel}
\end{equation}
The Gumbel--max identity \citep[see, e.g.,][]{LRWLS25} says that the token
chosen by \eqref{eq:llm-gumbel} has the same distribution $\pi_t$, i.e., standard
and watermarked generation have the same next-token distribution.  Their only
difference is that the watermarked rule uses the current key to decide which
token is selected, but the standard rule does not.

By definition, at position $t$, the null hypothesis $H_t$ states that the current token was
selected without using the current key, i.e., under the standard sampling rule; the Gumbel--max watermark
sampling is the alternative hypothesis.  In our setting, under $H_t$, the statistic $Y_t=U_t(W_t)$ is exactly uniform on
$[0,1]$ conditional on the past.  It follows that $L_t=-\log(1-Y_t)$ 
is an e-value, since
$\E[L_t\mid\mathcal F_{t-1}]=1$. Under the Gumbel--max watermark with a nondegenerate token distribution $\pi_t$, \(Y_t\mid \mathcal F_{t-1}\) strictly stochastically dominates
a \(\mathrm{Unif}(0,1)\) random variable \citep{SWZ26}.
Consequently, $\E[L_t\mid\mathcal F_{t-1}]>1$ under a nondegenerate Gumbel--max alternative, whereas it equals one under the null; thus, $L_t$ is a nontrivial e-value.
To conclude, if we denote by $\p_t$ the true conditional law of $Y_t$ given $\mathcal F_{t-1}$, then $H_t=\{\p_t\in\mathcal P_t\}$, where $\mathcal P_t=\{\mathrm{Unif}(0,1)\}$.

\paragraph{Tournament watermark.}
We also consider the 30-layer Tournament watermark introduced by 
\cite{DathathriEtAl24}. 
At position $t\in\T_{\ge 1}$ and layer $\ell\in\{1,\dots,30\}$, the secret key assigns every token $w$ a
fresh independent bit
$g_{t,\ell}(w)\lawis\operatorname{Ber}(1/2)$.  Write
$P_t^{(0)}=\pi_t$.  Inductively, at layer $\ell$, two tokens are sampled independently
from $P_t^{(\ell-1)}$; the token with the larger bit is retained, with ties
broken uniformly, and $P_t^{(\ell)}$ denotes the distribution of the retained
token (conditional on the bits $(g_{t,j}(w))_{j=1,\ldots,\ell;\,w\in\mathcal W}$).  The watermarked token is sampled from $P_t^{(30)}$, whose marginal distribution remains $\pi_t$.

For the observed token, let
$S_t=\sum_{\ell=1}^{30}g_{t,\ell}(W_t)$.  At a null position, $W_t$ is
selected without using its current bit tables, and hence
$S_t\mid\mathcal F_{t-1}\lawis\operatorname{Binom}(30,1/2)$.  Let $F_0$ and
$p_0$ denote the CDF and PMF of this distribution, and let
$V_t\lawis\operatorname{Unif}(0,1)$ be independent of anything else.  We define
\begin{equation*}
 Y_t=F_0(S_t-1)+V_t p_0(S_t)
\end{equation*}
and set $L_t=-\log(1-Y_t)$. The following proposition establishes
its validity under the null and its nontriviality under the
Tournament watermark. We use $\succeq_{\mathrm{st}}$ to denote first-order stochastic dominance: $Y\succeq_{\mathrm{st}}X$ if $\mathbb P(X>x)\le\mathbb P(Y>x)$ for every $x\in\mathbb R$.

\begin{proposition}
\label{prop:tournament-pivot}
At a null position,
\[
Y_t\mid\mathcal F_{t-1}\lawis\operatorname{Unif}(0,1),
\qquad
\E[L_t\mid\mathcal F_{t-1}]=1.
\]
Moreover, under the Tournament watermark,
\[
S_t\mid\mathcal F_{t-1}
\succeq_{\mathrm{st}}\operatorname{Binom}(30,1/2),
\qquad
Y_t\mid\mathcal F_{t-1}
\succeq_{\mathrm{st}}\operatorname{Unif}(0,1).
\]
If $\pi_t$ is not concentrated on a single token, then
$\E[L_t\mid\mathcal F_{t-1}]>1$.
\end{proposition}

It follows from Proposition \ref{prop:tournament-pivot} that $L_t=-\log(1-Y_t)$ is an e-value for $H_t$, 
and is nontrivial under a nondegenerate Tournament
watermark alternative. The rest follows in the same way as Gumbel--max.

\subsubsection{The e-processes and the role of resetting}\label{sec:refreshing exp}

{For each of the two watermark regimes, we apply the following e-process to its corresponding statistic $Y_t$.}
Following \cite[Section 4.2]{SWZ26}, we use the \emph{weight-adaptive} (WA) e-process, given by
\begin{align}
     M_t^{\mathrm{WA}}=\prod_{s=1}^t
 (1-\lambda_s+\lambda_sL_s),\label{eq:WA}
\end{align}
where $\lambda_1=0$ and
\begin{align}
    \lambda_t=\argmax_{\lambda\in[0,\frac{1}{2}]}
 \sum_{s=1}^{t-1}\log(1-\lambda+\lambda L_s),\qquad t\geq 2.\label{eq:lambda def}
\end{align}
This is the empirical growth rule known as GRAPA---growth-rate adaptive to the
particular alternative \citep{WSR24}.  Following \cite{SWZ26}, the argmax is taken over $[0,1/2]$, which prevents an early unusually small score from
destroying almost all of the accumulated wealth.  Since $\lambda_t$ is chosen before $Y_t$ is observed, each e-factor has
conditional mean one whenever the current position is null.  Consequently,
$(M_t^{\mathrm{WA}})_{t\in\T}$ is an e-process under the global null.
The one-time detectors introduced by \cite{SWZ26} stop after they first cross a threshold.  Our resetting process $\widetilde M$
defined in Section \ref{sec:e} breaks this barrier and allows us to detect several separated stretches of watermarked text.  

Note that after a reset, $\lambda_t$ is still computed from all earlier scores, i.e., the sum in \eqref{eq:lambda def} starts from $1\in\T$, not from the previous reset.  In addition, Assumption \ref{ass:ind} need not hold for the WA e-process, so we will refer to the upper bound given by Theorem \ref{thm:FDR_barrier}.  In Appendix \ref{subsec:llm-comparators} below, we introduce and compare other e-processes, including the resetting version of the ones in \cite{SWZ26} and a fixed-$\lambda$ version that satisfies  Assumption \ref{ass:ind}.

Other approaches to locating watermarked content in mixed-source text include changepoint segmentation \citep{BRK25}, adaptive online locator \citep{ZLWL25}, and multiple testing at the token level \citep{BCLLLS26}.

\subsubsection{Experimental design}

We use the public pretrained \texttt{facebook/opt-1.3b} model of
\cite{ZhangOPT22} with 20 fixed prompts and temperature given by $0.75$ or $1$.\footnote{A smaller temperature makes the next-token distribution more concentrated, and thus temperature $0.75$ serves as a stress test.}  For each prompt and each temperature, we generate 25
independently seeded paths, giving 500 paths per temperature. 
{We run the same design separately for the Gumbel--max and Tournament watermarks, using 500 newly generated paths per temperature for each watermark.}
Each path contains 600 generated tokens, all following the same null--alternative sequence:
\begin{align*}
 A_T
 &= \{51,\ldots,250\}\cup\{301,\ldots,500\},\\
 A_T^c
 &= \{1,\ldots,50\}\cup\{251,\ldots,300\}
    \cup\{501,\ldots,600\}.
\end{align*}
Here $A_T$ is simply the set of alternative (watermarked) positions; see Figure \ref{fig:llm-schedule}. 

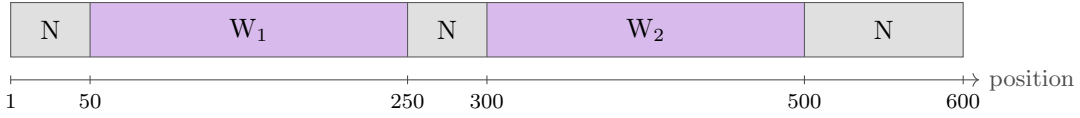
\begin{figure}[t]
\centering
\begin{tikzpicture}[x=0.21cm,y=0.72cm,font=\small]
  \fill[black!12] (0,0) rectangle (5,1);
  \fill[purple!38] (5,0) rectangle (25,1);
  \fill[black!12] (25,0) rectangle (30,1);
  \fill[purple!38] (30,0) rectangle (50,1);
  \fill[black!12] (50,0) rectangle (60,1);
  \draw[black!65] (0,0) rectangle (60,1);
  \foreach \x in {5,25,30,50}{\draw[black!65] (\x,0)--(\x,1);}
  \node at (2.5,.5) {$\mathrm N$};
  \node at (15,.5) {$\mathrm W_1$};
  \node at (27.5,.5) {$\mathrm N$};
  \node at (40,.5) {$\mathrm W_2$};
  \node at (55,.5) {$\mathrm N$};
  \draw[->,black!70] (0,-.42)--(61,-.42)
       node[right,font=\footnotesize]{position};
  \foreach \x/\lab in {0/1,5/50,25/250,30/300,50/500,60/600}{
    \draw[black!70] (\x,-.34)--(\x,-.50);
    \node[below,font=\scriptsize] at (\x,-.50) {$\lab$};
  }
\end{tikzpicture}
\caption{The null--alternative sequence for our primary experiment.  Grey segments marked $\mathrm N$
are standard, and purple segments marked $\mathrm W$ are watermarked.}
\label{fig:llm-schedule}
\end{figure}

Every resetting process uses the same threshold $\gamma=49$.  Theorem
\ref{thm:FDR_barrier} gives
\[
 \mathrm{FDR}_\infty^{\mathrm{ref}}
 \le \frac{1+\log 49}{49}\approx 0.0998<0.1,
\]
which gives the 10\% uniform FDR calibration used in the experiment.

\subsubsection{Results on the sequential resetting procedure}

Table \ref{tab:llm-primary-results} evaluates the union of all reported intervals using power, coverage FDP,
and IoU as defined in \eqref{eq:metrics}.  It also reports the uniform realized FDP, whose expectation is the uniform FDR in
\eqref{eq:FDR_infty}.

Figure
\ref{fig:llm-primary-trajectories} shows representative paths of the resetting
e-process. {Our results show that resetting substantially improves detection relative to the classic one-shot approach of stopping after
the first alarm.  At \(\vartheta=0.75\), its power is \(0.751\)--\(0.871\),
compared with \(0.055\) for the one-shot procedure; at \(\vartheta=1\),
the corresponding values are \(0.927\)--\(0.990\), compared with \(0.031\).
The ordinary resetting procedure admits a higher power, but the refined procedure has much lower coverage FDP and higher
IoU.  Thus, the refinement rule better identifies the
watermarked positions, at the cost of omitting a small proportion of alternative positions.
All empirical uniform realized FDP values are below \(1\%\), well below the theoretical \(10\%\) bound.}

\begin{table}[t]
\centering
\small
\setlength{\tabcolsep}{3.2pt}
\caption{Reports of the resetting e-processes {under the Gumbel--max watermark}. Green bold values mark the best
direction for the key metrics.}
\label{tab:llm-primary-results}
\begin{tabular}{@{}cccccc@{}}
\toprule
Procedure & Mean reports & Power & Coverage FDP & IoU
& Uniform realized FDP\\
\midrule
\multicolumn{6}{l}{$\vartheta=0.75$}\\
First alarm + block min., $\gamma=49$
  & 0.98 & 0.055 & 0.094 & 0.054
  & \green{\textbf{0.0000}}\\
Reset + whole block, $\gamma=49$
  & 24.50 & \green{\textbf{0.871}} & 0.215 & 0.699 & 0.0002\\
Reset + block min., $\gamma=49$
  & 24.50 & 0.751 & \green{\textbf{0.034}}
  & \green{\textbf{0.731}} & 0.0002\\
\midrule
\multicolumn{6}{l}{$\vartheta=1$}\\
First alarm + block min., $\gamma=49$
  & 1.00 & 0.031 & 0.124 & 0.031
  & \green{\textbf{0.0060}}\\
Reset + whole block, $\gamma=49$
  & 58.55 & \green{\textbf{0.990}} & 0.204 & 0.789 & 0.0062\\
Reset + block min., $\gamma=49$
  & 58.55 & 0.927 & \green{\textbf{0.035}}
  & \green{\textbf{0.897}} & 0.0064\\
\bottomrule
\end{tabular}
\end{table}

\begin{figure}[t]
\centering
\includegraphics[width=\textwidth]{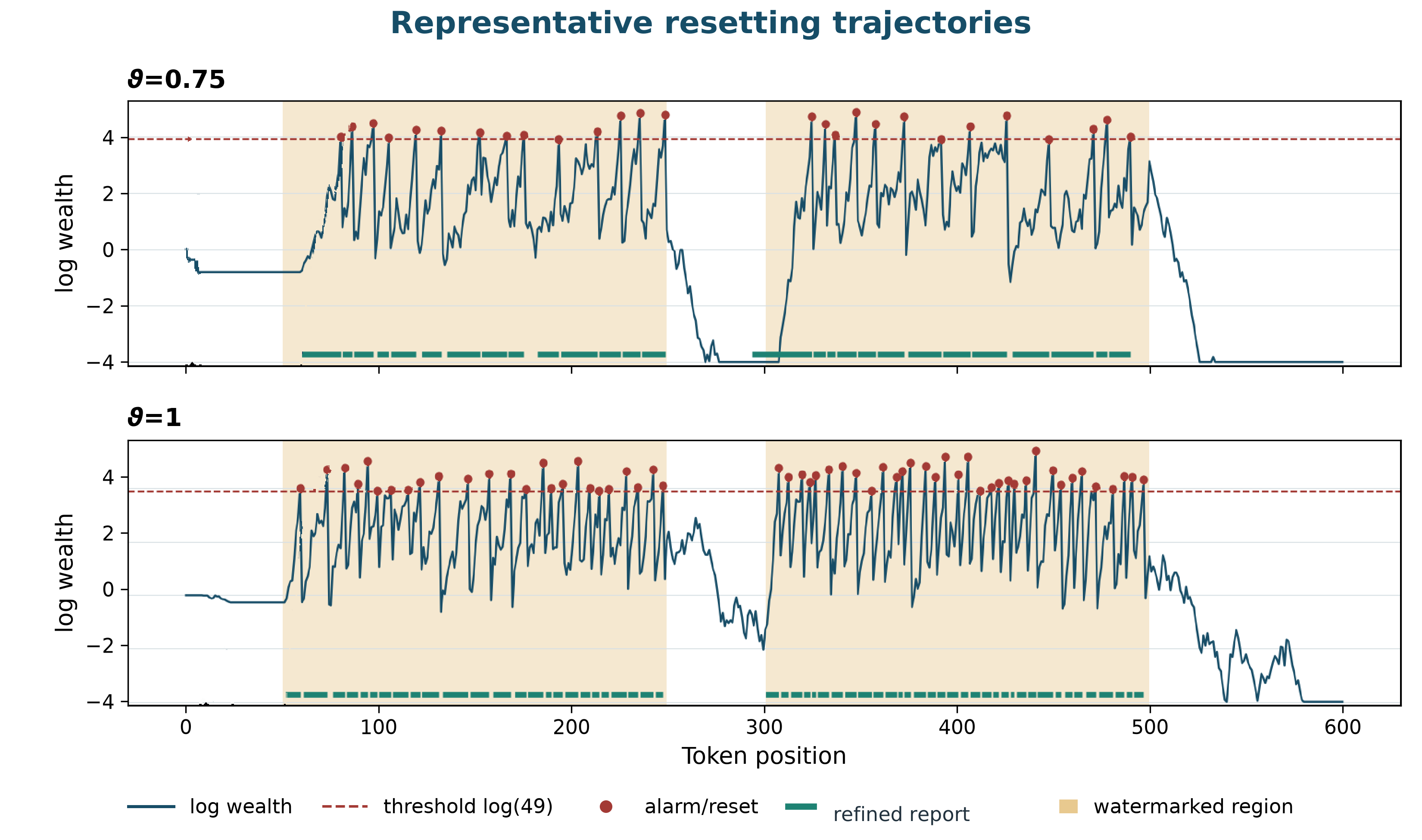}
\caption{{Representative resetting WA e-process trajectories {under the Gumbel--max watermark} at rejection threshold $\gamma=49$. Values below $-4$ are clipped from the drawing.}}
\label{fig:llm-primary-trajectories}
\end{figure}

{The Tournament experiment gives the same main conclusions.  At
$\vartheta=0.75$, the refined sequential resetting procedure has power
$0.740$, coverage FDP $0.037$, and IoU $0.718$; at $\vartheta=1$, the
corresponding values are $0.925$, $0.036$, and $0.894$.  These values are close
to the Gumbel--max results, showing that the benefits of resetting and
refinement persist under both watermark constructions.}

\begin{table}[t]
\centering
\small
\setlength{\tabcolsep}{3.2pt}
\caption{Tournament-watermark reports of the resetting e-processes. Green bold
values mark the best direction for the key metrics.}
\label{tab:llm-tournament-primary-results}
\begin{tabular}{@{}cccccc@{}}
\toprule
Procedure & Mean reports & Power & Coverage FDP & IoU
& Uniform realized FDP\\
\midrule
\multicolumn{6}{l}{$\vartheta=0.75$}\\
First alarm + block min., $\gamma=49$
  & 0.98 & 0.071 & 0.094 & 0.070
  & \green{\textbf{0.0060}}\\
Reset + whole block, $\gamma=49$
  & 19.09 & \green{\textbf{0.865}} & 0.217 & 0.693
  & 0.0063\\
Reset + block min., $\gamma=49$
  & 19.09 & 0.740 & \green{\textbf{0.037}}
  & \green{\textbf{0.718}} & 0.0071\\
\midrule
\multicolumn{6}{l}{$\vartheta=1$}\\
First alarm + block min., $\gamma=49$
  & 1.00 & 0.037 & 0.120 & 0.037
  & \green{\textbf{0.0020}}\\
Reset + whole block, $\gamma=49$
  & 49.08 & \green{\textbf{0.984}} & 0.205 & 0.785
  & \green{\textbf{0.0020}}\\
Reset + block min., $\gamma=49$
  & 49.08 & 0.925 & \green{\textbf{0.036}}
  & \green{\textbf{0.894}} & 0.0022\\
\bottomrule
\end{tabular}
\end{table}

\begin{figure}[t]
\centering
\includegraphics[width=\textwidth]{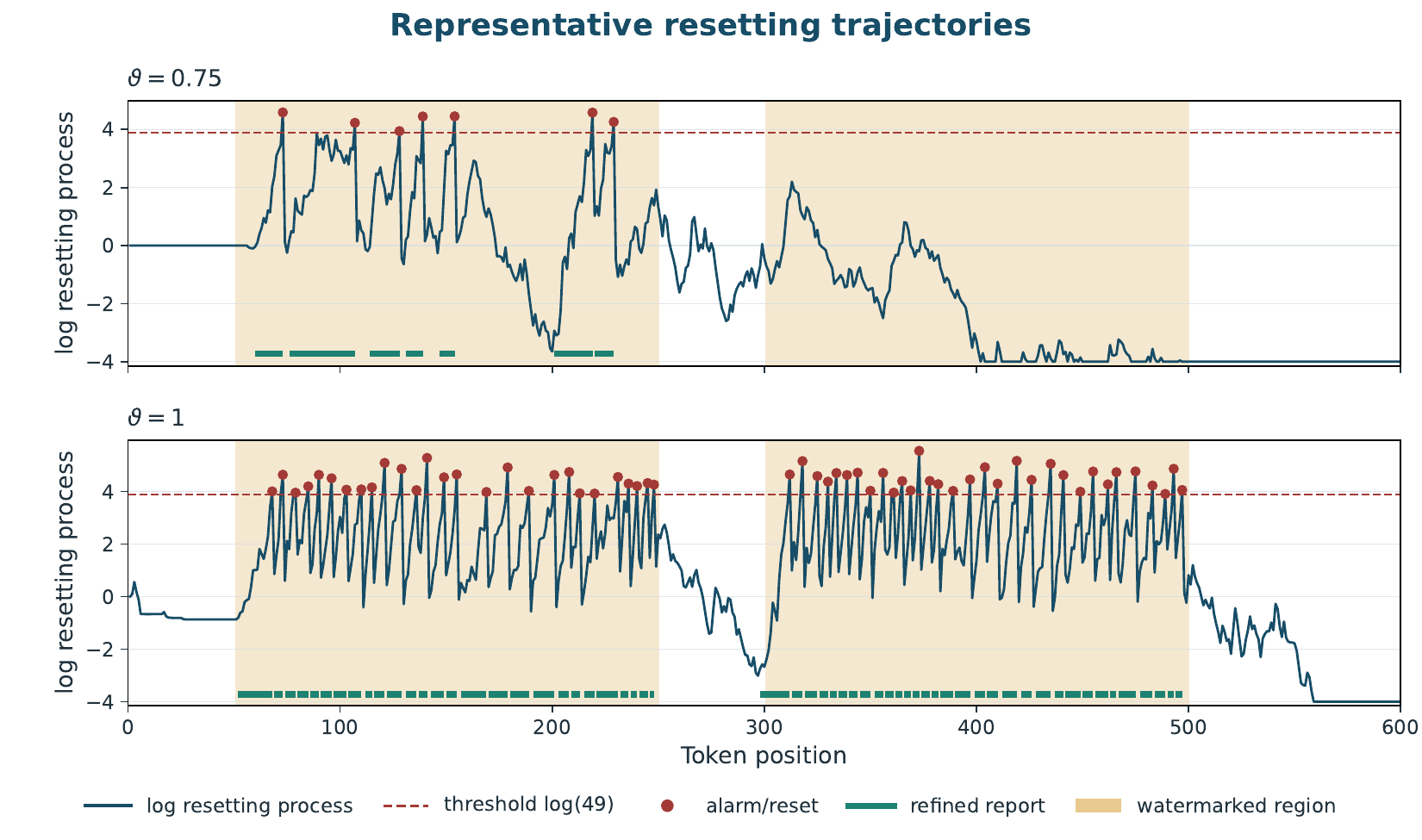}
\caption{{Representative resetting WA e-process trajectories for the Tournament experiment at rejection threshold $\gamma=49$. Values
below $-4$ are clipped from the drawing.}}
\label{fig:llm-tournament-primary-trajectories}
\end{figure}

Later in Appendix
\ref{subsec:four-interval}, we perform a four-interval stress test of the same experiment.

\subsubsection{Mixed human and watermarked text: {case studies}}
\label{sec:llm-mixed-case}

In this section, we develop a concrete example of the watermark detection application. We illustrate our sequential resetting procedure using {the four prose paragraphs in the Prologue of Bradley
Efron's 2010 book \emph{Large-Scale Inference}~\citep{Efron10}}. We partition the
excerpt at sentence boundaries into five alternating human/source-replacement
blocks, with target lengths 80, 100, 80, 100, and 80 OPT tokens. We retain the
three human blocks and replace each of the other two blocks by newly generated
watermarked OPT-1.3B sentences with temperature $\vartheta=1$. Each generated
block ends at the first sentence boundary reached after at least 90 tokens,
subject to a maximum of 140 tokens. The realized Gumbel--max mixed document has
$T=447$ tokens and block lengths 86, 98, 77, 123, and 63; the Tournament mixed
document has $T=425$ tokens and block lengths 86, 107, 77, 92, and 63. We use the resetting WA e-process (defined in \eqref{eq:WA}) with threshold $\gamma=49$.

We recall that the coverage FDP ($\mathrm{FDP}^{\mathrm{cov}}_T$)  in \eqref{eq:metrics} is token-wise, which is different from the rejection-wise realized FDP in \eqref{eq:FDP-FDR}. For example, if a test has one refined rejection block containing four positions (N, A, A, A), where ``N" represents a null position and ``A" represents an alternative position, then the coverage FDP is $1/4$ whereas the realized FDP is $0$ for this single block. Also, we note that the coverage FDR does not have a theoretical FDR control.

For the Gumbel-max watermark, Table~\ref{tab:llm-mixed-block-coverage} summarizes the blockwise coverage results, and the resetting e-process trajectory is given by Figure \ref{fig:efron} below.

\begin{table}[t]
\centering
\small
\setlength{\tabcolsep}{4pt}
\caption{Blockwise coverage on the Gumbel-max mixed document. ``Reported'' means included
in the union of the refined reports.}
\label{tab:llm-mixed-block-coverage}
\begin{tabular}{@{}ccccc@{}}
\toprule
Source block & Positions & No.~of tokens & Reported & Reported proportion \\
\midrule
Human source 1 & 1--86 & 86 & 56 & 65.1\% \\
Watermarked OPT-1.3B block 1 & 87--184 & 98 & 85 & 86.7\% \\
Human source 2 & 185--261 & 77 & 0 & 0.0\% \\
Watermarked OPT-1.3B block 2 & 262--384 & 123 & 116 & 94.3\% \\
Human source 3 & 385--447 & 63 & 0 & 0.0\% \\
\bottomrule
\end{tabular}
\end{table}

\begin{figure}[t]
\centering
\includegraphics[width=\textwidth]{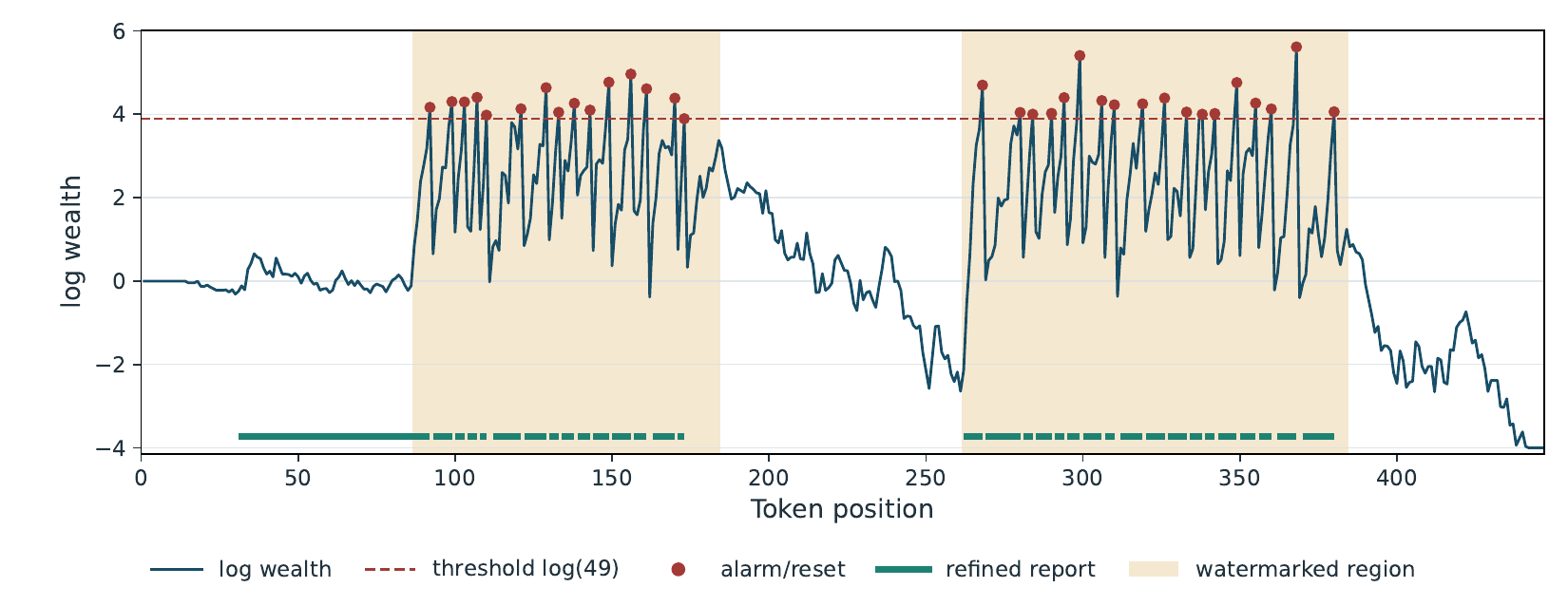}
\caption{{Resetting WA e-process trajectories for the Gumbel-max mixed document watermark detection example. Values below $-4$ are clipped from the drawing.}}
\label{fig:efron}
\end{figure}

The refined reports include 201 of the 221 watermarked tokens and 56 of the
226 human tokens. The resulting token-wise metrics are given by
\[
  \operatorname{Power}_{447}=0.910,\qquad
  \operatorname{FDP}^{\mathrm{cov}}_{447}=0.218,\qquad
  \operatorname{IoU}_{447}=0.726.
\]
The procedure contains 33 refined reports, with a zero realized FDP and
uniform realized FDP because every reported interval contains watermarked text.

{We repeat the same experiment using two complete, sentence-aligned
Tournament-watermarked continuations.
Table~\ref{tab:llm-tournament-mixed-block-coverage} and
Figure~\ref{fig:llm-tournament-efron} display the corresponding results.} See Figure~\ref{fig:llm-mixed-passage} for the full mixed document.

\begin{table}[t]
\centering
\small
\setlength{\tabcolsep}{4pt}
\caption{{Blockwise coverage on the Tournament mixed document.
``Reported'' means included in the union of the refined reports.}}
\label{tab:llm-tournament-mixed-block-coverage}
\begin{tabular}{@{}ccccc@{}}
\toprule
Source block & Positions & No.~of Tokens & Reported & Reported proportion \\
\midrule
Human source 1 & 1--86 & 86 & 1 & 1.2\% \\
Watermarked OPT-1.3B block 1 & 87--193 & 107 & 97 & 90.7\% \\
Human source 2 & 194--270 & 77 & 0 & 0.0\% \\
Watermarked OPT-1.3B block 2 & 271--362 & 92 & 87 & 94.6\% \\
Human source 3 & 363--425 & 63 & 0 & 0.0\% \\
\bottomrule
\end{tabular}
\end{table}

\begin{figure}[t]
\centering
\includegraphics[width=\textwidth]{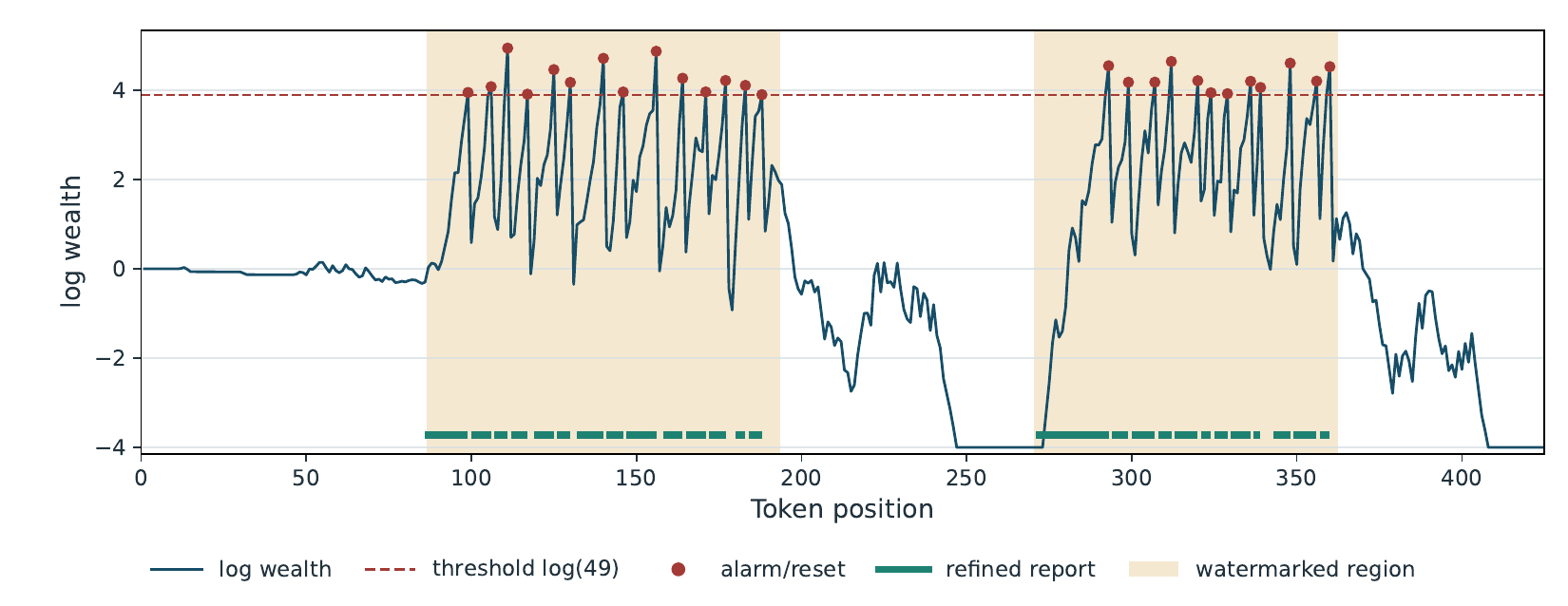}
\caption{{Resetting WA e-process trajectory for the
Tournament mixed-document watermark detection example. Values below $-4$ are
clipped from the drawing.}}
\label{fig:llm-tournament-efron}
\end{figure}

{
The refined reports include 184 of the 199 watermarked tokens and one of the
226 human tokens. Thus,
\[
  \operatorname{Power}_{425}=0.925,\qquad
  \operatorname{FDP}^{\mathrm{cov}}_{425}=0.005,\qquad
  \operatorname{IoU}_{425}=0.920.
\]
The procedure contains 26 refined reports, and its realized FDP and
uniform realized FDP are both zero.
}

\subsection{An application to financial backtesting}
\label{sec:backtest}

\subsubsection{Model-free sequential backtests of ES}

Backtesting Expected Shortfall (ES) is an important task in financial regulation. In a sequential manner, financial institutions report risk measure forecasts of Value-at-Risk (VaR) and ES. The regulator assesses the forecast quality according to the observed financial losses. Let $\M_0(\R)$ be the set of probability measures on $(\R,\mathcal B(\R))$ and $\M_1(\R)\subset\M_0(\R)$ be the set of those with finite first moment. \cite{WWZ26} introduced a model-free sequential testing procedure to backtest ES, defined at level $p\in(0,1)$ as
$$\ES_{p}(Q)=\frac{1}{1-p}\int^1_p\VaR_u(Q)\,\d u,~~Q\in\M_1(\R),$$
using the auxiliary information of VaR, defined as
$$\VaR_p(Q)=\inf\{x\in\R:Q((-\infty,x])\ge p\},~~Q\in \M_0(\R).$$
Specifically, for a fixed time horizon $T\in\N\cup\{\infty\}$,
let $X_t$ at time $t\in\T_{\ge 1}$ be the negated percentage log-return of the closing indices at $t$ and $t-1$. The forecasts $r_t$ and $z_t$ are $\mathcal F_{t-1}$-measurable and target, respectively, $\mathrm{ES}_p(X_t|\mathcal F_{t-1})$ and $\mathrm{VaR}_p(X_t|\mathcal F_{t-1})$. At time $t\in\T_{\ge 1}$, let
\[
  \mathcal P_t
  =\left\{Q\in\M_1(\R):
      \mathrm{ES}_p(Q)\leq r_t,
      \ \mathrm{VaR}_p(Q)=z_t
    \right\}.
\]
Recall that $\p_t$ is the true conditional law of $X_t$ given $\mathcal F_{t-1}$.
The null hypothesis of the backtest is thus given by $H_t=\{\p_t\in\mathcal P_t\}.$
In other words, this is a one-sided testing hypothesis where ES is not underestimated and the auxiliary forecast of VaR is correctly reported conditional on the past information.

Use the convention $0/0=0$. By Theorem 1 of \cite{WWZ26}, using the \emph{backtest e-statistic}
\[
  E_t=e^{\mathrm{ES}}_p(X_t,r_t,z_t)
  =\frac{(X_t-z_t)_+}{(1-p)(r_t-z_t)},~~r_t\ge z_t,~t\in\T_{\ge 1},
\]
for any $\mathcal F_{t-1}$-measurable $\lambda_t$ taking values in $[0,1]$, the following process with $M_0=1$ and
\begin{equation}\label{eq:backtest_mtg}
    M_t=M_t(\bm\lambda)=(1-\lambda_t+\lambda_tE_t)M_{t-1}=\prod^t_{s=1}(1-\lambda_s+\lambda_sE_s),~~t\in\T_{\ge 1},
\end{equation}
is a test supermartingale relative to $(\mathcal P_t)_{t\in\T_{\ge 1}}$. Here, the process $\bm\lambda=(\lambda_t)_{t\in\T_{\ge 1}}$ is called a \emph{betting process}, which can be calculated by various approaches to ensure the test power. Let $w\in\T_{\ge 1}$ be a rolling window size. Following \cite{WWZ26}, the \emph{growth rate for empirical e-statistics (GREE)} method estimates a growth-rate-optimal bet from past realized e-statistics: $\lambda_t^{\mathrm{GREE}}
  =0$ for $t=1,\dots,w$, and
\[
  \lambda_t^{\mathrm{GREE}}
  =\argmax_{\lambda\in[0,1/2]}
    \sum_{s=1}^{w}
    \log(1-\lambda+\lambda e^{\mathrm{ES}}_p(X_{t-s},r_{t-s},z_{t-s})),~~t\in\T_{\geq w+1},
\]
whereas the \emph{growth rate for empirical losses (GREL)} method applies today's forecasts to the historical losses: $\lambda_t^{\mathrm{GREL}}
  =0$ for $t=1,\dots,w$, and
\[
 \lambda_t^{\mathrm{GREL}}=\lambda_t^{\mathrm{GREL}}(r_t,z_t)
  =\argmax_{\lambda\in[0,1/2]}
    \sum_{s=1}^{w}
    \log(1-\lambda+\lambda e^{\mathrm{ES}}_p(X_{t-s},r_t,z_t)),~~t\in\T_{\geq w+1}.
\]
Here, we set an upper bound $1/2$ to the maximizer of the empirical criteria only to prevent ill-behaving events causing the process to take the value $0$. The bound is not restrictive under the current context because practical optimal betting processes are usually very small.
The numerical implementation uses Taylor approximation to these empirical criteria:
$$\lambda_t^{\mathrm{GREL}} \approx 0 \vee \frac{\sum_{s=1}^{w} e^{\mathrm{ES}}_p\left(X_{t-s}, r_t, z_t\right)-w}{\sum_{s=1}^{w}\left(e^{\mathrm{ES}}_p\left(X_{t-s}, r_t, z_t\right)-1\right)^2} \wedge \frac{1}{2},~~t\in\T_{\geq w+1}.$$
The approximated betting process for the GREE method can be calculated in a similar way by replacing $r_t$ and $z_t$ by $r_{t-s}$ and $z_{t-s}$.

The GREE and GREL methods are asymptotically optimal under different scenarios \citep[see Theorem 3 of][for details]{WWZ26}. To guarantee the test power when the practical data pattern is unclear, the \emph{growth rate for empirical mixture (GREM)} method is recommended, which takes the average of the test supermartingales generated by the GREE and GREL methods:
\[
  M_t(\bm\lambda^{\mathrm{GREM}})
    =\tfrac12M_t(\bm\lambda^{\mathrm{GREE}})+\tfrac12M_t(\bm\lambda^{\mathrm{GREL}}),~~t\in\T_{\geq 1},
\]
where
\[
\lambda_t^{\mathrm{GREM}}
  =\frac{M_{t-1}(\bm\lambda^{\mathrm{GREE}})\lambda_t^{\mathrm{GREE}}
         +M_{t-1}(\bm\lambda^{\mathrm{GREL}})\lambda_t^{\mathrm{GREL}}}
        {M_{t-1}(\bm\lambda^{\mathrm{GREE}})+M_{t-1}(\bm\lambda^{\mathrm{GREL}})}.
\]
We use the GREM method throughout.

\subsubsection{Revisiting the problem with sequential resetting procedures}

The method of \cite{WWZ26} is a powerful tool for financial regulators to backtest ES, especially under the sequential testing setup and when the underlying model is unknown. However, such a test is one-shot, which means that the testing procedure is regarded as complete after one rejection occurs. The issue matters
in market-risk surveillance because model adequacy can deteriorate, recover, and deteriorate again as
volatility, tail behavior, portfolio exposures, and estimation windows change. A regulator may therefore
want distinct warnings for the global financial crisis, the pandemic, and later episodes rather than one permanent verdict based on the first crisis.

Our sequential resetting procedure solves the practical problem described above by defining a resetting process $\widetilde M$ based on the test supermartingale $(M_t)_{t\in\T}$ in \eqref{eq:backtest_mtg}. Under the multiple testing framework, the general FDR bound in Theorem \ref{thm:FDR_barrier} applies. Specifically, we choose the rejection threshold $\gamma=1/\alpha\approx 19.972$ to guarantee an FDR bound of $\alpha(1+\log(1/\alpha))= 0.2$.\footnote{We set a higher FDR bound in this section as $0.2$ because it is usually beneficial to get early warnings for financial backtests even though it becomes more likely to obtain false positives, which is not a severe problem.}

For comparison, we adopt the same data setup as in \cite{WWZ26} but extend the sample to contain the daily negated percentage log-returns of the NASDAQ Composite index from January 17, 1996 through December 31, 2025. The forecasts of $\VaR_{0.975}$ and $\ES_{0.975}$ are based on an AR(1)--GARCH(1,1) model estimated with a moving window. The estimation window and the computation window for the betting process are both set to be $w=500$ days.
Therefore, the computed test supermartingale contains 6,539 out-of-sample data points from January 3, 2000 through December 31, 2025.
Innovations are assumed to be i.i.d.~with mean zero and variance one, following normal, t, or skewed-t distributions; an empirical forecast obtained by a nonparametric rolling window method is also included, together with a deliberately conservative specification that increases the skewed-t ES forecasts by 10\%.
We plot the losses and forecasts in
Figure~\ref{fig:extended-figure5}.

\begin{figure}[t]
\centering
\includegraphics[width=0.98\textwidth]{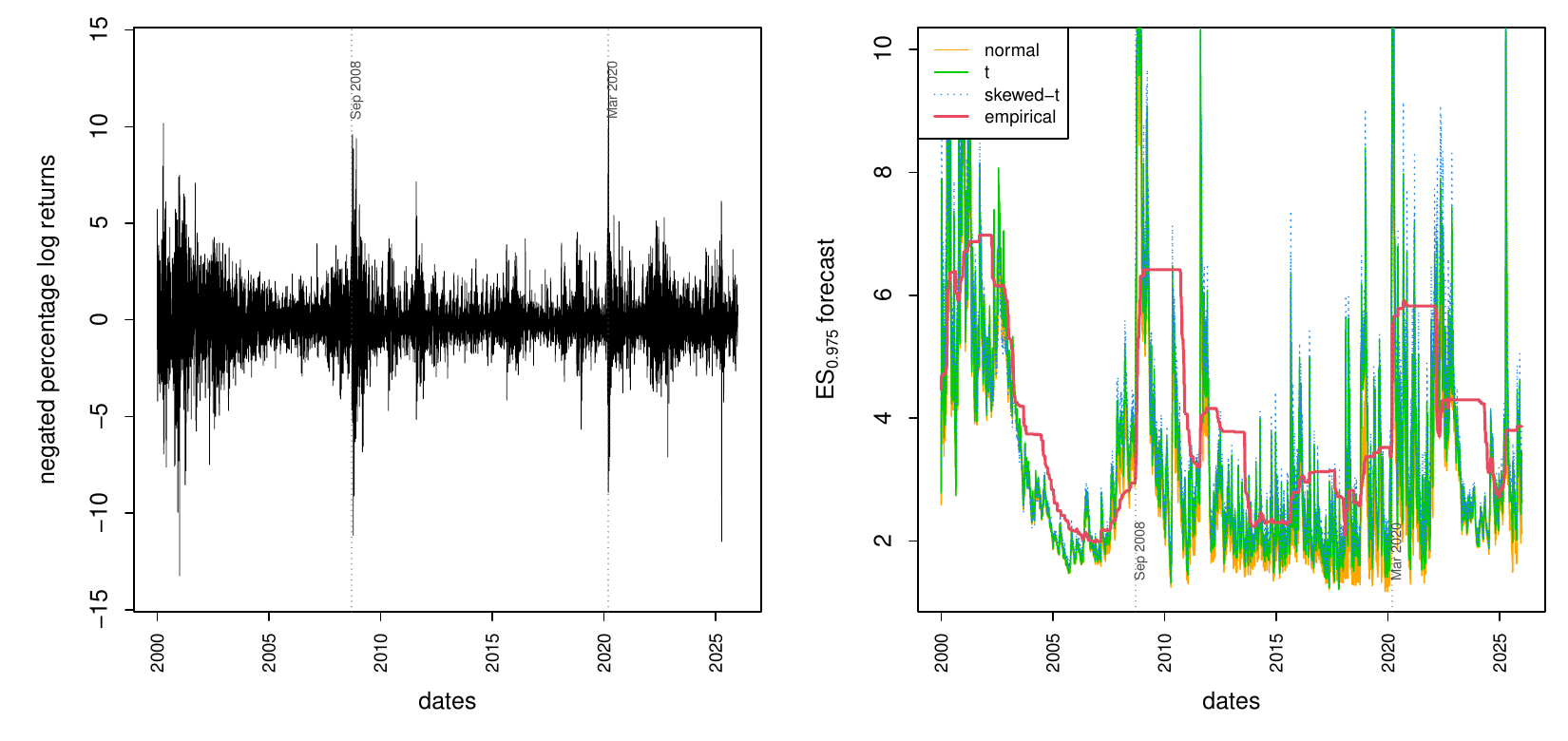}
\caption{Negated percentage log-returns of the NASDAQ Composite index (left panel) and $\ES_{0.975}$ forecasts (right panel) from January 3, 2000, to December 31, 2025.}
\label{fig:extended-figure5}
\end{figure}

\begin{figure}[!t]
\centering
\includegraphics[width=0.98\linewidth]{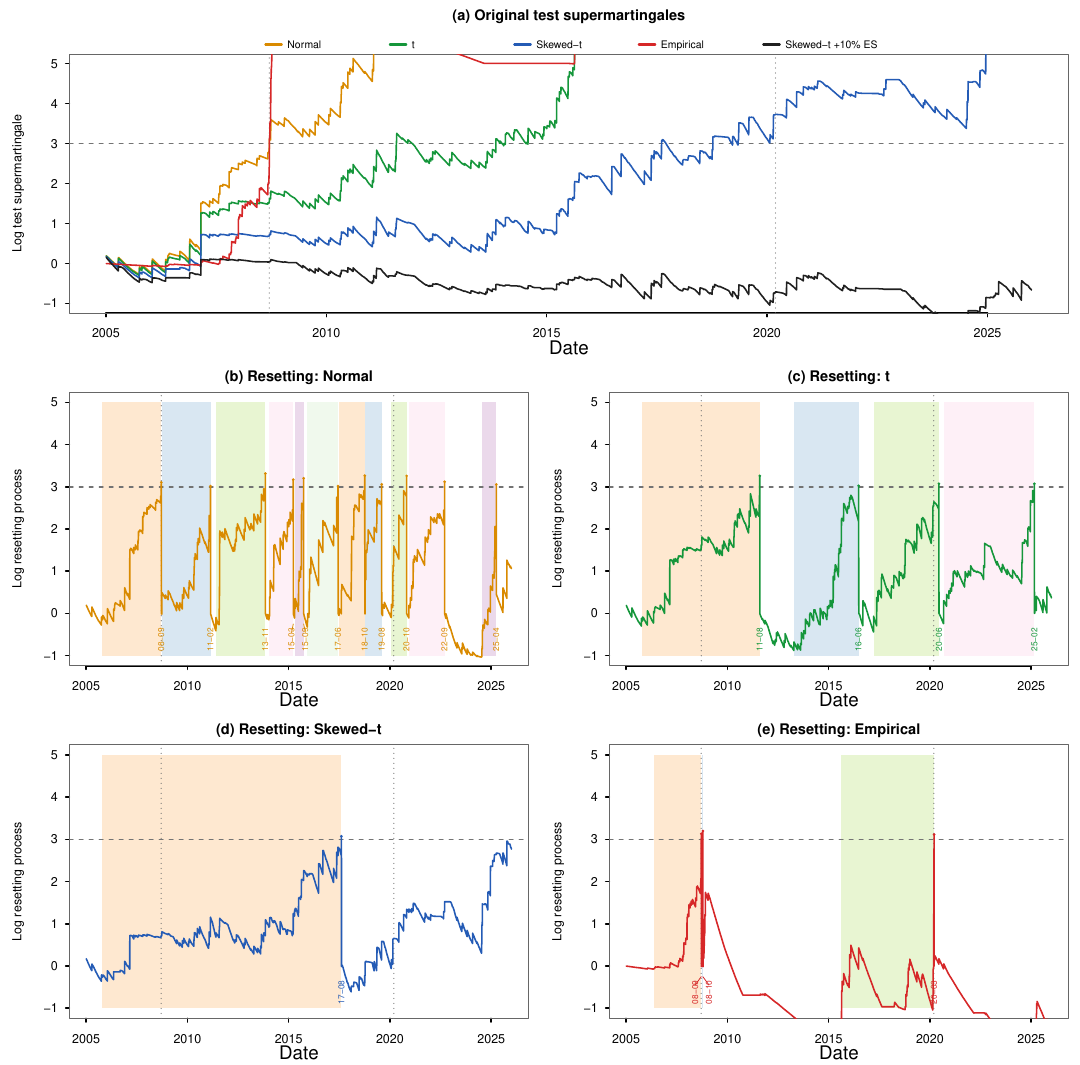}
\caption{Log-transformed original test supermartingales (a) and sequential resetting processes (b--e) for normal, t, skewed-t, and empirical forecasts, testing $\ES_{0.975}$ for the NASDAQ Composite index from January 4, 2005, to December 31, 2025. The forecast methods are shown in orange, green, blue, and red, respectively; the black line in (a) uses skewed-t ES forecasts increased by $10\%$. Colored areas mark refined rejection intervals, and labels give rejection dates (YY-MM). Horizontal dashed lines show the rejection threshold. Vertical dotted lines mark September 15, 2008 (financial crisis) and March 11, 2020 (COVID-19).}
\label{fig:e-backtest}
\end{figure}

We take the computed test supermartingales starting from January 4, 2005 through December 31, 2025 and focus on the impact of the 2008 financial crisis (marked on September 15, 2008) and the COVID-19 pandemic period (marked on March 11, 2020).\footnote{We do not record the test supermartingale between January 3, 2000 and January 3, 2005 because it is not helpful to show the effect of the 2008 financial crisis.}  Figure \ref{fig:e-backtest} shows the refined rejection blocks of backtesting $\ES_{0.975}$ forecasts under different forecasting methods (normal, t, skewed-t, and empirical). The normal forecasts generate 11 warning blocks, whereas the t and
skewed-t forecasts generate four and one, respectively. The deliberately conservative skewed-t $+10\%$ ES forecasts are not
rejected throughout. This pattern is broadly consistent with the increasing
conservativeness of the normal, t, skewed-t, and skewed-t $+10\%$ forecasts. The empirical forecast generates three warning blocks. As shown in
Figure~\ref{fig:extended-figure5}, it varies much less from day to day than
the parametric forecasts because it is based on a rolling historical window.
It may therefore adapt slowly when tail risk rises abruptly. This helps explain
its rapid warnings following the large market losses in September 2008 and
March 2020, but does not imply that the empirical specification is generally
less conservative or more frequently rejected than all parametric
specifications.

For the 2008 financial crisis, normal and empirical forecasts are both rejected on September 17, 2008, two trading days after September 15, 2008. Their refined intervals begin on October 5, 2005 and May 11, 2006, respectively, so the block-minimum times identify long evidence-bearing periods culminating in the crisis warning rather than claiming that the crisis began at those refined starts. The t and skewed-t forecasts yield longer rejection blocks covering the crisis. However, they also exhibit early warnings during the crisis and can be rejected when lower rejection thresholds are set, yielding a higher FDR bound.
For the COVID-19 pandemic period, the empirical forecasts are rejected on March 16, 2020, immediately after the March 2020 market break. The heavier-tailed parametric forecasts respond later: t forecasts are rejected on June 11 and normal on October 28, while skewed-t forecasts are not rejected in 2020. After 2021, normal forecasts are rejected again on September 13, 2022 and April 3, 2025, and t forecasts are rejected on February 27, 2025, showing that the extension also captures later evidence not attributable solely to the initial pandemic shock.

\section{Concluding remarks}

The new   testing procedure proposed in this paper builds on e-values and allows for continuous identification of alternatives through resetting after each rejection. The refinement technique further identifies locations of the alternatives more precisely, for which FDR bounds are established.
 The proposed approach is different from  existing methods for resetting procedures, such as ARL-based methods in changepoint detection \citep[Sections~1.1 and~2.5]{SRR24}, which have different error criteria than ours. Compared with the ARL-based methods, our procedures are often more conservative in terms of rejections, because they have infinite ARL under the global null: Ville's inequality gives $\P(\tau_1<\infty)\le\alpha<1$. This is necessary for the FDR control under global null. %

The bound expressed through geometric random variables in Theorem~\ref{thm:FDR_barrier} is sharp, and the explicit bound $\alpha(1+\log(1/\alpha))$ therein is asymptotically sharp as $\alpha\downarrow0$. The sharpness construction in Example~\ref{ex:infty} uses a particular data-dependent arrangement of null and alternative times. In our numerical experiments such as watermark detection, we observed that the empirical FDRs are well below the corresponding theoretical bounds. This is a typical situation for e-value-based methods, and there is room for further improvements of power if additional assumptions are available. 

We discussed  two applications in the paper, and there are potentially more; for instance, 
manufacturing process monitoring and anomaly detection  are two application areas in which our refined procedures can be useful.

\section*{Acknowledgments}

We thank Xiaoming Huo, Aaditya Ramdas and Qihe Tang for helpful discussions. 
Ruodu Wang was supported by the Natural Sciences and Engineering Research Council of Canada (CRC-2022-00141, RGPIN-2024-03728).
Zhenyuan Zhang was supported by the Jump Trading Fellowship.
GPT-5.6 Sol and GPT-6 Astra were used in the development of technical results and in the production of codes for numerical experiments.
The paper is human-written.
All replication files for the numerical results can be found at \url{https://github.com/qwangan/Sequential-Resetting-Procedures}.

\setcounter{lemma}{0}
\renewcommand{\thelemma}{\thesection.\arabic{lemma}}
\setcounter{proposition}{0}
\renewcommand{\theproposition}{\thesection.\arabic{proposition}}
\setcounter{theorem}{0}
\renewcommand{\thetheorem}{\thesection.\arabic{theorem}}
\setcounter{definition}{0}
\renewcommand{\thedefinition}{\thesection.\arabic{definition}}
\setcounter{corollary}{0}
\renewcommand{\thecorollary}{\thesection.\arabic{corollary}}
\setcounter{example}{0}
\renewcommand{\theexample}{\thesection.\arabic{example}}
\setcounter{equation}{0}
\renewcommand{\theequation}{\thesection.\arabic{equation}}
\renewcommand{\thefigure}{\thesection.\arabic{figure}}
\renewcommand{\thetable}{\thesection.\arabic{table}}

\appendix

\renewcommand{\theHequation}{\theHsection.\arabic{equation}}
\renewcommand{\theHlemma}{\theHsection.\arabic{lemma}}
\renewcommand{\theHproposition}{\theHsection.\arabic{proposition}}
\renewcommand{\theHtheorem}{\theHsection.\arabic{theorem}}
\renewcommand{\theHdefinition}{\theHsection.\arabic{definition}}
\renewcommand{\theHcorollary}{\theHsection.\arabic{corollary}}
\renewcommand{\theHexample}{\theHsection.\arabic{example}}
\renewcommand{\theHfigure}{\theHsection.\arabic{figure}}
\renewcommand{\theHtable}{\theHsection.\arabic{table}}

\section{Additional illustrations and remarks}

\label{app:sec:r}

First, we present in Figure~\ref{fig:ex1} an illustration of the FDP paths  in 
Example~\ref{ex:non-iid-alpha-counterexample}, which gives the computation 
\[
\begin{aligned}
\mathrm{FDR}^{\mathrm{ord}}_4
&= q\left(
    \frac34\times \frac1{25}
    +\frac23\times\frac4{25}
    +\frac45
    \right)
    +(1-q)\frac12\times\frac1{25}= \frac{11q}{12}+\frac1{50}
   \xrightarrow[q\uparrow1/5]{}
   \frac{61}{300}>\frac15.
\end{aligned}
\]

\begin{figure}[t]
  \includegraphics[width=\linewidth]{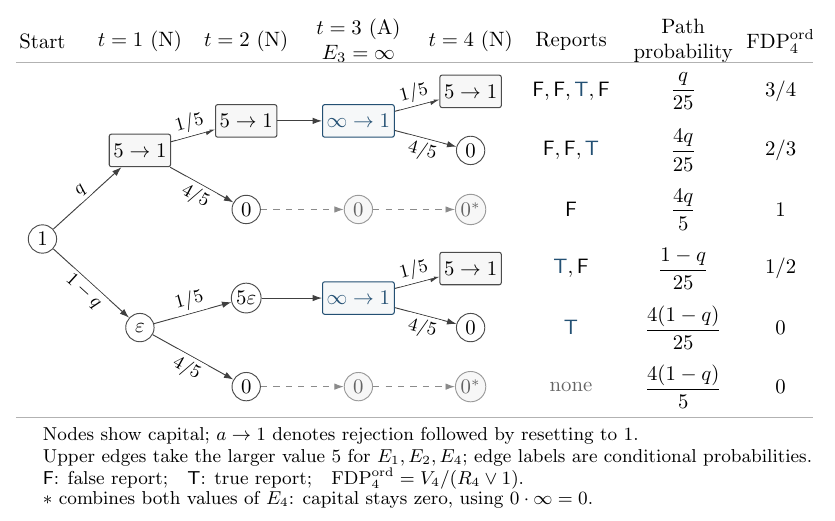}
  \caption{Computation of the ordinary FDR in Example~\ref{ex:non-iid-alpha-counterexample}.
    N and A denote null and alternative times. } 
\label{fig:ex1}
 
\end{figure}

Second, we discuss 
 FDR bound with a deterministic number of rejections. In this 
simpler setting, we  fix a number $r^*\in\N$ and stop the testing procedure once the total number of rejections reaches $r^*$, i.e., at time $\tau_{r^*}$. On the complementary event where $\tau_{r^*}=\infty$, we declare no rejections. The number of false positives at time $\tau_{r^*}$ is thus given by $$V^\sigma_{\tau_{r^*}}=\id_{\{\tau_{r^*}<\infty\}}\sum^{r^*}_{r=1}\min\left\{1-\id_{\{\p_t\notin\mathcal P_t\}}:t\in(\sigma_r,\tau_r]\right\}=\id_{\{\tau_{r^*}<\infty\}}\sum^{r^*}_{r=1}\id_{\{\nu^\sigma_r>\tau_r\}},$$
    and the na\"{i}ve truncated FDR of the testing procedure  with target $r^*$ is defined as
    $$\mathrm{FDR}^{\mathrm{ref}}_{\mathrm{tr}}(r^*)=\E\left[\frac{V_{\tau_{r^*}}^\sigma}{r^*}\right].$$
Because $r^*\in\N$ is deterministic, we have
\begin{align*}
    \mathrm{FDR}^{\mathrm{ref}}_{\mathrm{tr}}(r^*)=\frac{1}{r^*}\,\E[V_{\tau_{r^*}}^\sigma]\leq \frac{1}{r^*}\,\E\bigg[\sum^{r^*}_{r=1}\id_{\{\tau_r<\infty,\,\nu^\sigma_r>\tau_r\}}\bigg]
        \le \alpha,
\end{align*}
where the last inequality follows from Lemma \ref{lem:ville} in Appendix \ref{app:sec:proof_thm2}. Although yielding a smaller FDR upper bound, the current approach is overly conservative because it sets the error to zero whenever fewer than $r^*$ rejections occur, which understates the genuine FDR and may lead to lower power.

\section{Proofs of the main results}
\label{app:sec:proofs}

\subsection{Proof of Theorem \ref{thm:localized-main}}
\label{app:sec:proof_thm1}

\subsubsection{Monotonicity and weighted capital costs under mutual independence}
\label{sec:tech_ind}

For a fixed $\alpha\in(0,1)$, we introduce the following definitions that will be used throughout this section.
For $m\in\mathbb N_0$, $c\in[0,1/\alpha]$, and
$\mathbf z=(z_1,\ldots,z_m)\in[0,\infty]^m$,
define recursively
$Y_0(c,\mathbf z)=C_0(c,\mathbf z)=c$,
$A_0(c,\mathbf z)=0$,
\begin{equation}
\label{eq:reset-maps}
\begin{aligned}
&Y_k(c,\mathbf z)=C_{k-1}(c,\mathbf z)z_k,~C_k(c,\mathbf z)
=\begin{cases}
1,
&
\mbox{if}~Y_k(c,\mathbf z)\ge\frac{1}{\alpha},\\
Y_k(c,\mathbf z),
&
\mbox{if}~Y_k(c,\mathbf z)<\frac{1}{\alpha},
\end{cases}\\
\mbox{and}~~& A_k(c,\mathbf z)
=A_{k-1}(c,\mathbf z)
+\id_{\{Y_k(c,\mathbf z)\ge1/\alpha\}},~~k=1,\ldots,m.
\end{aligned}
\end{equation}
Thus, $Y_m(c,\mathbf z)$ and $C_m(c,\mathbf z)$ are the terminal capital before and after the reset
decision at time $m$, respectively, and $A_m(c,\mathbf z)$ is the total number of resets up to time $m$.

The following lemma shows that a deterministic process, which resets to $1$ whenever it exceeds a threshold, only yields a larger number of resets for a higher starting capital.

\begin{lemma}\label{lem:monotone}
Fix $m\in\mathbb N_0$, $\alpha\in(0,1)$, and deterministic multipliers
$\mathbf z=(z_1,\dots,z_m)\in[0,\infty]^m$. For each $c\in[0,1/\alpha]$, denote $Y_k(c)=Y_k(c,\mathbf z)$, $C_k(c)=C_k(c,\mathbf z)$, and $A_k(c)=A_k(c,\mathbf z)$ for $k=1,\dots,m$ as in \eqref{eq:reset-maps}.
Then the mapping $x\mapsto A_m(x)$ is increasing on $[0,1/\alpha]$.
Namely,
\[
0\le c_1\le c_2\le \frac{1}{\alpha}\implies A_m(c_1)\le A_m(c_2).
\]
\end{lemma}
\begin{proof}
Run the two processes with the same multipliers, starting from
\(0\le c_1\le c_2\le1/\alpha\).
We prove by induction on \(k=0,\dots,m\) that
\begin{equation}
        A_k(c_2)\ge A_k(c_1),
        \qquad
        C_k(c_2)\ge
        \alpha^{A_k(c_2)-A_k(c_1)}C_k(c_1).
        \label{eq:monotone-invariant}
\end{equation}
The assertion \eqref{eq:monotone-invariant} holds for \(k=0\).  For $k\in\{1,\dots,m\}$, suppose \eqref{eq:monotone-invariant} holds at time \(k-1\), and
put \(d=A_{k-1}(c_2)-A_{k-1}(c_1)\ge 0\). It follows that \begin{equation}\label{eq:lem1_induction}
    Y_k(c_2)=C_{k-1}(c_2)z_k\ge \alpha^d C_{k-1}(c_1)z_k=\alpha^d Y_k(c_1).
\end{equation}
If $Y_k(c_2),Y_k(c_1)\in [0,1/\alpha)$ or $Y_k(c_2),Y_k(c_1)\in [1/\alpha,\infty]$, then $A_k(c_2)-A_k(c_1)=d$ is kept unchanged and thus
$C_k(c_2)\ge\alpha^{d}C_k(c_1)$ by \eqref{eq:reset-maps} and \eqref{eq:lem1_induction}. If $Y_k(c_2)\ge 1/\alpha$ while $Y_k(c_1)<1/\alpha$, then
$A_k(c_2)-A_k(c_1)=d+1$ and by \eqref{eq:reset-maps},
$$C_k(c_2)=1\geq \alpha^d> \alpha^{d+1}Y_k(c_1)=\alpha^{d+1}C_k(c_1).$$
If $Y_k(c_1)\ge 1/\alpha$ while $Y_k(c_2)<1/\alpha$, we have $d\ge 1$. Otherwise, if $d=0$, we have by \eqref{eq:lem1_induction} that $$\frac{1}{\alpha}>Y_k(c_2)\ge Y_k(c_1)\ge \frac{1}{\alpha},$$
which leads to a contradiction. Therefore, $A_k(c_2)-A_k(c_1)=d-1\ge 0$. By \eqref{eq:reset-maps} and \eqref{eq:lem1_induction},
$$C_k(c_2)=Y_k(c_2)\ge\alpha^d Y_k(c_1)\ge \alpha^{d-1}=\alpha^{d-1}C_k(c_1).$$
By the arguments above, \eqref{eq:monotone-invariant} holds for all $k=0,\dots,m$ and the proof is thus complete.
\end{proof}

The following lemma indicates that the conditional expectation of the weighted capital cost $c-Q$ is bounded below by the conditional expectation of the weighted benchmark $(1/\alpha-1)N$. It serves as a deterministic building block of Lemma \ref{lem:clean-start} below, which will change the initial capital $c$ to a random variable.

\begin{lemma}
\label{lem:clean-budget}
Fix $m\in\N_0$ and \(c\in[0,1]\). Let $\bm\zeta=(\zeta_1,\ldots,\zeta_m)$ be a nonnegative random vector with mutually independent components satisfying
$\E[\zeta_i]\le1$ for $i=1,\ldots,m$. Write $Y_i=Y_i(c,\bm\zeta)$, $C_i=C_i(c,\bm\zeta)$, and $A_i=A_i(c,\bm\zeta)$
for $i=0,\dots,m$ as in \eqref{eq:reset-maps}.
Then, for every bounded decreasing function
\(w:[0,1/\alpha]\to[0,\infty)\),
\begin{equation}
        \E\left[
        w(C_m)\bigl(C_m+(\alpha^{-1}-1)A_m\bigr)
        \right]
        \le
        c\,\E[w(C_m)].
        \label{eq:clean-budget}
\end{equation}
\end{lemma}

\begin{proof}
The case \(c=0\) or $m=0$ is trivial, so it suffices to assume \(c>0\) and $m\in\N$.
For \(i=0,\ldots,m\), define
$
B_i=C_i+(1/\alpha-1)A_i.
$

We first show that
$\E[B_m]\le c.$
Let \(\mathcal G_i=\sigma(\zeta_1,\ldots,\zeta_i)\). Conditional on
\(\mathcal G_{i-1}\), if \(Y_i<1/\alpha\), then
\[
B_i-B_{i-1}
=
Y_i-C_{i-1}
=
C_{i-1}(\zeta_i-1).
\]
If \(Y_i\ge1/\alpha\), then
\[
B_i-B_{i-1}
=
\frac{1}{\alpha}-C_{i-1}
\le
Y_i-C_{i-1}
=
C_{i-1}(\zeta_i-1).
\]
Therefore,
\[
\E[B_i-B_{i-1}\mid \mathcal G_{i-1}]
\le
C_{i-1}\bigl(\E[\zeta_i]-1\bigr)
\le0,
\]
where we use the independence of \(\zeta_i\) from \(\mathcal G_{i-1}\) and
\(\E[\zeta_i]\le1\). Hence \((B_i)_{i=0}^m\) is a supermartingale, and
$\E[B_m]\le B_0=c.$

We next prove the lower-tail inequality
\begin{equation}
        \E\left[(B_m-c)\id_{\{C_m\le y\}}\right]\le0,~~\mbox{for all}~ 0\le y\le\frac{1}{\alpha}.
        \label{eq:lower-tail}
\end{equation}
For \(y\ge c\), on \(\{C_m>y\}\) we have \(B_m\ge C_m>y\ge c\). Hence
\begin{equation}\label{eq:lem2_case1}
        \E[(B_m-c)\id_{\{C_m\le y\}}]
        =
        \E[B_m-c]-\E[(B_m-c)\id_{\{C_m>y\}}]\le0.
\end{equation}
For \(0<y<c\), we prove the following stronger inequality by
induction for all \(i=0,\dots,m\): 
\begin{equation}
        \E\left[
        \left(B_i-c+\left(\frac cy-1\right)C_i\right)
        \id_{\{C_i\le y\}}
        \right]\le0,
        ~~\mbox{for all}~ 0<y<c.
        \label{eq:strong-lower-tail}
\end{equation}
For \(i=0\), we have \(C_0=c\). The event \(\{C_0\le y\}\) is empty because \(y<c\).
For $i\in\{1,\dots,m\}$, assume the claim \eqref{eq:strong-lower-tail} holds for length \(i-1\), uniformly over all thresholds
\(y\in(0,c)\). Namely,
\begin{equation}\label{eq:lem2_m-1}
    \E\left[
        \left(B_{i-1}-c+\left(\frac cy-1\right)C_{i-1}\right)
        \id_{\{C_{i-1}\le y\}}
        \right]\le 0,
        ~~\mbox{for all}~ 0<y<c.
\end{equation}
Conditional on $\zeta_i=z$ for any fixed $z\ge 0$, on the event $\{C_{i}\le y\}$, because $y<c\le 1$, we have $C_i<1$ and thus $C_{i}=Y_i<1/\alpha$ by \eqref{eq:reset-maps}. It follows that
$$
        C_i=zC_{i-1},
        \qquad\mbox{and}\qquad
        B_i=B_{i-1}+(z-1)C_{i-1}.
$$
Write the new threshold $y'=y/z$, where \(y/0=\infty\). We have
\begin{equation}
        \E\left[
        \left(B_i-c+\left(\frac cy-1\right)C_i\right)
        \id_{\{C_i\le y\}}
        \right]=\E\left[
        \left(
        B_{i-1}-c+\left(\frac{c}{y'}-1\right)C_{i-1}
        \right)
        \id_{\{C_{i-1}\le y'\}}
        \right].
        \label{eq:last-step-recursion}
\end{equation}
We check the possible values of \(z\) and \(y'\).

\begin{itemize}
\item If \(z>0\) and \(y'<c\), then
\eqref{eq:last-step-recursion} is nonpositive by \eqref{eq:lem2_m-1}.

\item If \(z>0\) and \(c\le y'\le1/\alpha\), then
${c}/{y'}-1\le 0.$
The previous case in \eqref{eq:lem2_case1} yields that
$\E[(B_{i-1}-c)\id_{\{C_{i-1}\le y'\}}]\le 0.$
Together with \(C_{i-1}\ge 0\), this makes
\eqref{eq:last-step-recursion} nonpositive.

\item If \(z>0\) and \(y'>1/\alpha\), then ${c}/{y'}-1\le 0$ and
\(\id_{\{C_{i-1}\le y'\}}=1\) because
\(0\le C_{i-1}<1/\alpha\).
Hence, \eqref{eq:last-step-recursion} is at most
$\E[B_{i-1}-c]\le 0.$

\item If \(z=0\), then \eqref{eq:last-step-recursion} becomes
$
\E[B_{i-1}-c-C_{i-1}]\le 0.
$
\end{itemize}
Averaging over \(z\) proves \eqref{eq:strong-lower-tail}. The endpoint \(y=0\) follows by
dominated convergence of \eqref{eq:lower-tail} as \(y\downarrow0\). Indeed, because
$
0\le C_m\le 1/\alpha$ and
$0\le A_m\le m,$
\(B_m\) is bounded by a constant depending only on \(\alpha\) and \(m\). Therefore, \eqref{eq:lower-tail} holds true.

Now define the finite signed measure
\[
\mu(A)=\E\left[(B_m-c)\id_{\{C_m\in A\}}\right],
\qquad A\in\mathcal B([0,1/\alpha]).
\]
By \eqref{eq:lower-tail},
$\mu([0,y])\le 0$
for all $y\in[0,1/\alpha]$.
Moreover, \(\mu([0,y))\le 0\) by continuity from below.
We claim that
$
\int w\,\mathrm d\mu\le0
$
for any bounded nonnegative decreasing Borel function \(w\). Let
\(\mu=\mu^+-\mu^-\) be the Jordan decomposition of $\mu$.
Tonelli's theorem applied separately to \(\mu^+\) and \(\mu^-\) gives
\[
        \int w\,\mathrm{d}\mu
        =\int\int_0^\infty \id_{\{w(x)>u\}}\,\mathrm du\,\d\mu
        =
        \int_0^\infty \mu(\{x:w(x)>u\})\,\mathrm{d}u .
\]
For each \(u\ge 0\), because \(w\) is decreasing, the superlevel set
\(\{x:w(x)>u\}\) is either empty, the whole interval
\([0,1/\alpha]\), or an initial interval of the form \([0,y]\) or
\([0,y)\). Each such set has nonpositive \(\mu\)-mass. Therefore
$
\int w\,\mathrm d\mu\le 0.
$

Finally, by the definition of \(\mu\),
\[
0\ge \int w\,\mathrm d\mu
=
\E[w(C_m)(B_m-c)].
\]
Substituting \(B_m=C_m+(1/\alpha-1)A_m\) gives
\eqref{eq:clean-budget}.
\end{proof}

Fix $t\in\T$. For \(s=0,\dots,t\), let
\begin{equation}
\label{eq:ab-def}
\begin{aligned}
        a(s)=\min\{u\in\{s+1,\ldots,t\}:u\notin\H_0\}~~
        \mbox{and}~~
        b(s)=(a(s)-1)\wedge t,
\end{aligned}
\end{equation}
so \((s,b(s)]\) is the all-null interval immediately following \(s\) within the terminal horizon \(t\). The e-values $E_1,\dots,E_t$ are the same as those in \eqref{eq:M_tilde}. Define the $\sigma$-field generated by the e-values as
\begin{equation}\label{eq:sigma_increments}
    \mathcal G_u=\sigma(E_1,\ldots,E_u),~~u=1,\dots,t.
\end{equation}
Let $\mathcal G_0$ be the trivial $\sigma$-field.
The following result extends Lemma \ref{lem:clean-budget} into the probabilistic version.

\begin{lemma}
\label{lem:clean-start}
Fix $t\in\T$.
Suppose that \(E_1,\ldots,E_t\) satisfy Assumption \ref{ass:ind}. For $s=0,\dots,t$, write \(\mathbf E_s=(E_{s+1},\dots,E_{b(s)})\) and $\tilde b(s)=b(s)-s$.
Let \(S\) be a \((\mathcal G_u)\)-stopping time taking values in $\{0,\dots,t\}$.
Let \(X\) be \(\mathcal G_S\)-measurable taking values in \([0,1]\).
Let
$Q=C_{\tilde b(S)}(X,\mathbf E_S)$, $N=A_{\tilde b(S)}(X,\mathbf E_S)$ as in \eqref{eq:reset-maps}, and
$w:[0,1/\alpha]\times\Omega\to[0,\infty)$
be a bounded and
\(\mathcal B([0,1/\alpha])\otimes\mathcal G_S\)-measurable function such that
$q\mapsto w(q,\omega)$
is decreasing for every \(\omega\in\Omega\). Write
$w(Q)(\omega)=w(Q(\omega),\omega).$
We have
\begin{equation}
\E\left[
w(Q)\bigl(Q+(\alpha^{-1}-1)N\bigr)
\,\middle|\,
\mathcal G_S
\right]
\le
X\,\E\left[w(Q)\,\middle|\,\mathcal G_S\right]~~\mbox{a.s.}
\label{eq:clean-start}
\end{equation}
\end{lemma}
\begin{proof}
Put $B=Q+(1/\alpha-1)N$.
For $s\in\{0,\dots,t\}$, because \(X\) is \(\mathcal G_S\)-measurable, there exists a
\(\mathcal G_s\)-measurable random variable \(X_s\) such that $X_s=X$ on $\{S=s\}$;
similarly, there exists a
\(\mathcal B([0,1/\alpha])\otimes\mathcal G_s\)-measurable function
\(w_s:[0,1/\alpha]\times\Omega\to[0,\infty)\) such that $w_s(q,\omega)=w(q,\omega)$, decreasing in $q$, for all $q\in[0,1/\alpha]$ and $\omega\in\{S=s\}$.

For $\omega\in\Omega$ and \(\mathbf z\in[0,\infty)^{\tilde b(s)}\), define 
\[
\begin{aligned}
        h_s(\omega,\mathbf z)
        =
        w_s\left(C_{\tilde b(s)}(X_s(\omega),\mathbf z),\omega\right)
        \left(
        C_{\tilde b(s)}(X_s(\omega),\mathbf z)
        +(\alpha^{-1}-1) A_{\tilde b(s)}(X_s(\omega),\mathbf z)
        -X_s(\omega)
        \right).
\end{aligned}
\]
It is clear that \(h_s\) is bounded and
\(\mathcal G_s\otimes\mathcal B([0,\infty)^{\tilde b(s)})\)-measurable.
Let $\Pi_s=
        \otimes_{j=1}^{\tilde b(s)}\mathcal L(E_{s+j})$, with the convention that \(\Pi_s\) is the point mass at \(\varnothing\) when \(\tilde b(s)=0\).
By \eqref{eq:ab-def}, $u\in\H_0$ for all \(u\in\{s+1,\ldots,b(s)\}\). Note that
\(\mathbb E[E_u]\le1\) and  \(E_u<\infty\) almost surely for all \(u\in\{s+1,\ldots,b(s)\}\).
 Moreover,
Assumption~\ref{ass:ind} implies that the components of
\(\mathbf E_{s}\) are mutually independent and
\(\mathbf E_{s}\) is independent of \(\mathcal G_s\).  In particular,
\(\mathcal L(\mathbf E_{s})=\Pi_s\).
For any bounded nonnegative $\mathcal G_S$-measurable random variable $H$, we have \(H\id_{\{S=s\}}\in\mathcal G_s\). Hence,
\[
\begin{aligned}
\E\left[H w(Q)(B-X)\right] &=
\sum_{s=0}^t
\E\left[
H\id_{\{S=s\}}
h_s(\cdot,\mathbf E_{s})
\right]\\
&=\sum_{s=0}^t
\E\left[
H\id_{\{S=s\}}
\E[h_s(\cdot,\mathbf E_{s})\,|\,\mathcal G_s]
\right]
=
\sum_{s=0}^t
\E\left[
H\id_{\{S=s\}}
\int h_s(\cdot,\mathbf z)\,\Pi_s(\mathrm d\mathbf z)
\right].
\end{aligned}
\]
For any fixed $\omega\in\Omega$ and $s\in\{0,\dots,t\}$, $X_s(\omega)\in[0,1]$, and $q\mapsto w_s(q,\omega)$ is bounded and decreasing on $q\in[0,1/\alpha]$. Hence, Lemma~\ref{lem:clean-budget}
gives
$\int h_s(\omega,\mathbf z)\,\Pi_s(\mathrm d\mathbf z)\le 0$ and thus $\int h_s(\cdot,\mathbf z)\,\Pi_s(\mathrm d\mathbf z)\le 0$.
Furthermore, because \(H\id_{\{S=s\}}\ge 0\), we have
$\E\left[H w(Q)(B-X)\right]\le 0$, and thus
$$\E\left[w(Q)(Q+(\alpha^{-1}-1)N-X)\,\middle|\,\mathcal G_S\right]\le 0~~\mbox{a.s.}$$
Therefore, \eqref{eq:clean-start} holds by the \(\mathcal G_S\)-measurability of \(X\).
\end{proof}

\subsubsection{Clean intervals and their properties}
\label{sec:clean-stretches}
In this section, we construct clean intervals used in the proof of Theorem \ref{thm:localized-main}. Fix any time horizon $t\in\T$. Recall that \(Z_s\) in \eqref{eq:M_tilde} denotes the capital carried
forward after the reset at time \(s\in\{0,\dots,t\}\), $R_s$ is the number of rejections up to $s$, and $\tau_{R_s}$ is
the most recent rejection time by time \(s\). By definition, no rejection occurs on
\((\tau_{R_s},s]\). Define
$L_s=\min_{\tau_{R_s}\le u\le s}Z_u$
as the running minimum of the currently active process in $[\tau_{R_s},s]$.

Use $a(\cdot)$ and $b(\cdot)$ given by \eqref{eq:ab-def}. Set $G_1=\Omega$ and $\lambda_1=0$. Recursively for $\ell=1,\dots,t+1$, define
\begin{equation}\label{eq:stretch_variables}
    \beta_\ell= b(\lambda_\ell)\id_{G_\ell} + t\id_{G_\ell^c},~~X_\ell=Z_{\lambda_\ell}\id_{G_\ell},~~ Q_\ell=Z_{\beta_\ell}\id_{G_\ell},\quad\mbox{and}\quad N_\ell=(R_{\beta_\ell}-R_{\lambda_\ell})\id_{G_\ell}.
\end{equation}
Define the random interval
\[ I_\ell(\omega) = \begin{cases} \{\lambda_\ell(\omega)+1,\ldots,\beta_\ell(\omega)\}, & \omega\in G_\ell,\\ \emptyset, & \omega\notin G_\ell. \end{cases} \]
On the event \(G_\ell\), define
\begin{equation}\label{eq:reject_minimum}
\begin{aligned}
    &\theta_\ell = \inf\left\{i\in\{a(\lambda_\ell)\wedge(t+1),\ldots,t\}: \widetilde M_i\ge\frac{1}{\alpha}\right\}\\
\mbox{and}~~&U_\ell = \inf\left\{ i\in\{a(\lambda_\ell)\wedge(t+1),\ldots,t\}: \widetilde M_q<\frac{1}{\alpha}\ \text{for all }q=a(\lambda_\ell)\wedge(t+1),\ldots,i,~ Z_i\le L_{\beta_\ell} \right\}.
\end{aligned}
\end{equation}
By definition, $\theta_\ell\neq U_\ell$ on $G_\ell\cap \{\theta_\ell\wedge U_\ell<\infty\}$.
On \(\omega\notin G_\ell\), set $\theta_\ell=U_\ell=\infty.$
Now define
\begin{equation}\label{eq:stretch_begin}
    G_{\ell+1} = G_\ell \cap \{\theta_\ell\wedge U_\ell<\infty\}\qquad\mbox{and}\qquad\lambda_{\ell+1}=\theta_\ell\wedge U_\ell\wedge t. 
\end{equation}

After this recursion has been defined, we explain the notations. For $\ell\in\{1,\ldots,t+1\}$, because \(\beta_\ell=b(\lambda_\ell)\) on \(G_\ell\), the interval \(I_\ell=(\lambda_\ell,\beta_\ell]\) contains only null times. We call the all-null interval $I_\ell$ the $\ell$-th \emph{clean interval}.

A clean interval may or may not be empty: On $G^c_{\ell+1}$, $I_{\ell+1}=\emptyset$ is called \emph{dummy}, in which case $\theta_\ell=U_\ell=\infty$, $\lambda_{\ell+1}=t$, and $X_{\ell+1}=Q_{\ell+1}=N_{\ell+1}=0$. On $G_{\ell+1}$, the $(\ell+1)$-st clean interval $I_{\ell+1}$ is called \emph{genuine}, which starts at $\lambda_{\ell+1}+1$. Here, $\lambda_{\ell+1}$ is either a true rejection time $\theta_{\ell}$ or the time $U_{\ell}$ corresponding to a running minimum $L_{\beta_\ell}$ in the block \([\tau_{R_{\beta_\ell}},\beta_\ell]\), between which we choose the earlier time:\footnote{We take $\ell$ up to $t+1$ only for the purpose of bookkeeping. In the real problem, genuine clean intervals do not have one-to-one correspondence with the time indices $1,\dots,t+1$.}
\begin{itemize}
    \item When $\theta_\ell<U_\ell$, a rejection occurs before the test supermartingale returns to the running minimum $L_{\beta_\ell}$. There exists one and only one rejection in the interval $(\beta_\ell,\lambda_{\ell+1}]$ at $\lambda_{\ell+1}=\theta_\ell$, which is true under the refined law. On the other hand, all rejections in the clean intervals $(\lambda_\ell,\beta_\ell]$ and $(\lambda_{\ell+1},\beta_{\ell+1}]$ are false under the refined law. See Figure \ref{fig:localized-decomposition}(a) for an illustration.
    \item When $\theta_\ell>U_\ell$, the test supermartingale drops below the running minimum $L_{\beta_\ell}$ before it exceeds the rejection threshold. In this case, $\lambda_{\ell+1}=U_\ell$. All data in the interval $(\beta_\ell,\lambda_{\ell+1}]$ should be dropped according to the refinement criteria. Therefore, there are no true rejections between the two consecutive genuine clean intervals $(\lambda_\ell,\beta_\ell]$ and $(\lambda_{\ell+1},\beta_{\ell+1}]$. See Figure \ref{fig:localized-decomposition}(b) for this scenario.
\end{itemize}

\begin{figure}[t]
\centering
\begin{tikzpicture}[x=0.74cm,y=1.76cm,>=Latex,font=\scriptsize]
\useasboundingbox (-1.75,-4.00) rectangle (13.55,3.22);
\fill[white] (-1.75,-4.00) rectangle (13.55,3.22);
\tikzset{
        every node/.append style={text=black},
        caplabel/.style={anchor=east,fill=white,inner sep=1.2pt,text=black},
        timelabel/.style={align=center,fill=white,inner sep=1.1pt,text=black},
        callout/.style={align=center,fill=white,inner sep=1.4pt,text=black},
        guide/.style={black!50,dotted},
        stretch/.style={fill=black!6},
        stretchB/.style={fill=black!12}
}

\begin{scope}
        \node[anchor=west,font=\small\bfseries] at (-1.5,2.62)
                {(a) $\theta_\ell<U_\ell$: rejection before returning to the running minimum};
        \draw[->,line width=0.85pt] (0,0) -- (12.75,0) node[right]{time};
        \draw[->,line width=0.85pt] (0,0) -- (0,2.38);
        \node[rotate=90] at (-1.25,1.20) {capital levels};
        \draw[black!55,dashed] (0.15,2.12) -- (12.18,2.12) node[right,text=black]{\(1/\alpha\)};
        \draw[black!50,dotted] (0.15,0.52) -- (12.18,0.52) node[right,text=black]{running min. $L_{\beta_\ell}$};
        \draw[guide] (0.10,1.18) -- (-0.02,1.18) node[caplabel]{\(X_\ell\)};
        \draw[guide] (0.10,0.86) -- (-0.02,0.86) node[caplabel]{\(Q_\ell\)};
        \fill[stretch] (0.82,0) rectangle (4.25,2.12);
        \fill[stretchB] (8.72,0) rectangle (11.70,2.12);
        \node[fill=white,inner sep=1pt] at (2.54,1.84) {clean interval $I_\ell$};
        \node[callout, text width=2.5cm, fill=white,inner sep=1pt] at (10.22,1.84) {next clean interval $I_{\ell+1}$};
        \draw[black,line width=1.12pt]
                plot[smooth] coordinates {
                (0.82,1.18) (1.42,0.86) (2.15,0.52) (3.00,0.68)
                (4.25,0.86) (5.35,1.14) (6.72,1.32) (8.10,1.56)
                (8.72,2.12)};
        \draw[black,line width=1.12pt] (8.72,2.12) -- (8.72,1.18);
        \draw[black,line width=1.12pt]
                plot[smooth] coordinates {
                (8.72,1.18) (9.28,0.98) (10.10,0.76)
                (10.92,0.88) (11.70,0.94)};
        \foreach \x/\y/\lab in {
                0.82/1.18/{\lambda_\ell},
                4.25/0.86/{\beta_\ell},
                5.35/1.14/{a(\lambda_\ell)},
                8.72/2.12/{\theta_\ell=\lambda_{\ell+1}}} {
                \draw[guide] (\x,0) -- (\x,\y);
                \fill (\x,\y) circle (1.25pt);
                \node[timelabel,below=6pt] at (\x,0) {\(\lab\)};
        }
        \node[callout] at (7.30,1.82) {refined\\true reset};
        \draw[->,black!65] (7.68,1.92) -- (8.62,2.08);
\end{scope}

\begin{scope}[yshift=-5.95cm]
        \node[anchor=west,font=\small\bfseries] at (-1.5,2.62)
                {(b) $\theta_\ell>U_\ell$: dropping below the running minimum before rejection};
        \draw[->,line width=0.85pt] (0,0) -- (12.75,0) node[right]{time};
        \draw[->,line width=0.85pt] (0,0) -- (0,2.38);
        \node[rotate=90] at (-1.25,1.20) {capital levels};
        \draw[black!55,dashed] (0.15,2.12) -- (12.18,2.12) node[right,text=black]{\(1/\alpha\)};
        \draw[black!50,dotted] (0.15,0.40) -- (12.18,0.40) node[right,text=black]{running min. $L_{\beta_\ell}$};
        \draw[guide] (0.10,1.18) -- (-0.02,1.18) node[caplabel]{\(X_\ell\)};
        \draw[guide] (0.10,0.90) -- (-0.02,0.90) node[caplabel]{\(Q_\ell\)};
        \draw[guide] (0.10,0.66) -- (-0.02,0.66) node[caplabel]{\(Q_{\ell+1}\)};
        \draw[guide] (0.10,0.40) -- (-0.02,0.40) node[caplabel]{\(X_{\ell+1}\)};
        \fill[stretch] (0.82,0) rectangle (4.25,2.12);
        \fill[stretchB] (8.00,0) rectangle (11.20,2.12);
        \node[fill=white,inner sep=1pt] at (2.54,1.84) {clean interval $I_\ell$};
        \node[callout, text width=2.5cm, fill=white,inner sep=1pt] at (9.60,1.84) {next clean interval $I_{\ell+1}$};
        \draw[black,line width=1.12pt]
                plot[smooth] coordinates {
                (0.82,1.18) (1.42,0.90) (2.15,0.40) (3.00,0.62)
                (4.25,0.90) (5.35,1.14) (6.25,0.95) (7.25,0.68)
                (8.00,0.40) (8.55,0.32) (9.20,0.50)
                (9.85,0.36) (10.55,0.60) (11.20,0.66)};
        \foreach \x/\y/\lab in {
                0.82/1.18/{\lambda_\ell},
                4.25/0.90/{\beta_\ell},
                5.35/1.14/{a(\lambda_\ell)}} {
                \draw[guide] (\x,0) -- (\x,\y);
                \fill (\x,\y) circle (1.25pt);
                \node[timelabel,below=6pt] at (\x,0) {\(\lab\)};
        }
        \draw[guide] (8.00,0) -- (8.00,0.40);
        \fill (8.00,0.40) circle (1.25pt);
        \node[timelabel,below=6pt] at (8.00,0) {\(U_\ell=\lambda_{\ell+1}\)};
        \node[callout] at (9.46,1.20) {first drop below\\running minimum};
        \draw[->,black!65] (9.00,1.00) -- (8.08,0.44);
\end{scope}
\end{tikzpicture}
\caption{Graphical illustration of clean intervals.}
\label{fig:localized-decomposition}
\end{figure}
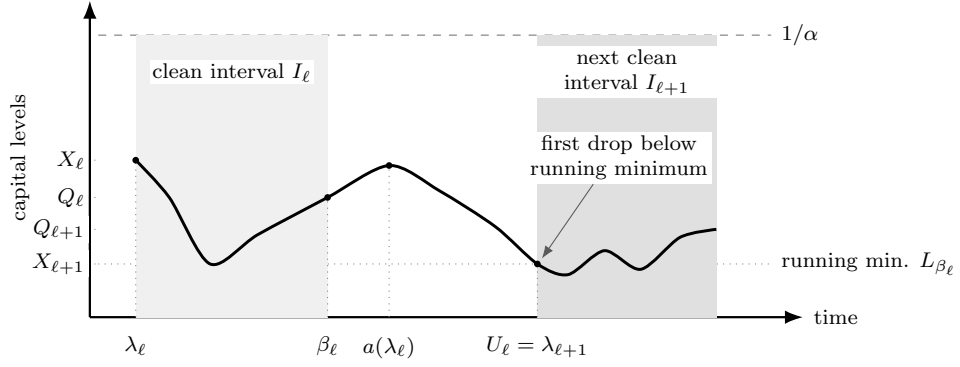

 For the $\ell$-th clean interval, $X_{\ell}$ and $Q_{\ell}$ represent its initial and terminal capitals, respectively, and $N_\ell$ counts the total number of rejections within the interval.

Recall the $\sigma$-fields $(\mathcal G_u)$ in \eqref{eq:sigma_increments}. The following lemma verifies the stopping time and measurability properties of the quantities defined above.

\begin{lemma}
\label{lem:decomposition}
The recursive construction has the following properties.
\begin{enumerate}
\item[(i)] For every \(\ell\in\{1,\ldots,t+1\}\), \(\lambda_\ell\) is a stopping time and
\(G_\ell\in\mathcal G_{\lambda_\ell}\).
\item[(ii)] For every \(\ell\in\{1,\ldots,t+1\}\), \(\beta_\ell=b(\lambda_\ell)\) is a stopping time, and $u\in\H_0$ for all
\(u\in(\lambda_\ell,\beta_\ell]\) on \(G_\ell\).
\item[(iii)] For every \(\ell\in\{1,\ldots,t+1\}\), \(X_\ell\) is \(\mathcal G_{\lambda_\ell}\)-measurable, while
\(Q_\ell\) and \(N_\ell\) are \(\mathcal G_{\beta_\ell}\)-measurable.
\item[(iv)] For every \(\ell\in\{1,\ldots,t\}\), if $I_\ell$ and $I_{\ell+1}$ are both genuine, then
\[
        \lambda_\ell
        <\beta_\ell+1= b(\lambda_\ell)+1
        =a(\lambda_\ell)
        \le \lambda_{\ell+1}.
\]
If $I_{\bar \ell}$ is dummy for some $\bar \ell\in\{1,\ldots,t+1\}$, then $I_\ell$ is dummy for every $\ell\in\{\bar\ell+1,\dots,t+1\}$.
\end{enumerate}
\end{lemma}

\begin{proof}
Statement (iv) follows directly from \eqref{eq:ab-def} and \eqref{eq:stretch_begin}.
We prove statements (i)--(iii) by induction.
For $\ell=1$, $\lambda_1=0$ and $G_1=\Omega$. Hence,
$$\beta_1=b(0)=\left(\min\{u\in\{1,\ldots,t\}:u\notin\H_0\}-1\right)\wedge t
$$
is a stopping time. Moreover, $X_1=Z_0=1$, and $Q_1=Z_{\beta_1}$ and $N_1=R_{\beta_1}$ are $\mathcal G_{\beta_1}$-measurable. Therefore, statements (i)--(iii) are true for $\ell=1$.

Suppose that statements (i)--(iii) are true for $\ell\in\{1,\dots,t\}$. We prove they also hold true for $\ell+1$. Because $\lambda_\ell$ is a stopping time, $\theta_\ell$ and $U_\ell$ are both stopping times by \eqref{eq:reject_minimum}.
For statement (i), we have $\{\lambda_{\ell+1}\le t\}=\Omega\in\mathcal G_t$. For all $n<t$, we have $G_\ell\in\mathcal G_{s}$ on $\{\lambda_\ell=s\}$ for any $s=0,\dots,n-1$. Hence,
\[
        \{\lambda_{\ell+1}\le n\}
        =
        \bigcup_{s=0}^{n-1}
        \left(
        \{\lambda_\ell=s\}\cap G_\ell
        \cap\{\theta_\ell\wedge U_\ell\le n\}
        \right)\in\mathcal G_{n}.
\]
Thus, $\lambda_{\ell+1}$ is a stopping time.
Moreover, $G_{\ell+1}=G_\ell\cap\{\theta_\ell\wedge U_\ell<\infty\}\in\mathcal G_{\lambda_{\ell+1}}$ because $G_{\ell}\in\mathcal G_{\lambda_\ell}$, $\lambda_{\ell+1}\ge \lambda_\ell$, and $\theta_\ell\wedge U_\ell\le t$ on $\{\theta_\ell\wedge U_\ell<\infty\}$.

The assertion that $u\in\H_0$ for all $u\in(\lambda_{\ell+1},\beta_{\ell+1}]$ on $G_{\ell+1}$ is straightforward by definition. The rest of statement (ii) follows from the fact that $\lambda_{\ell+1}$ is a stopping time and
\[
        \{\beta_{\ell+1}\le n\}
        =
        \bigcup_{s:\,b(s)\le n}
        \{\lambda_{\ell+1}=s\}\in\mathcal G_n~~\mbox{for all }n=0,\dots,t,
\]
because $s\le b(s)\le n$ for all $s=0,\dots,t$ such that $b(s)\le n$.

Statement (iii) follows from the definition \eqref{eq:stretch_variables}, $G_{\ell+1}\in\mathcal G_{\lambda_{\ell+1}}$, $\beta_{\ell+1}\ge \lambda_{\ell+1}$, and the adaptiveness of $(Z_u)_{u=0}^t$ and $(R_u)_{u=0}^t$. The proof is complete.
\end{proof}

For the fixed terminal horizon \(t\in\T\), define the break-index set
\begin{equation}\label{eq:break_set}
    \mathcal K = \{0,t+1\} \cup \{\ell\in\{1,\ldots,t\}:\theta_\ell<U_\ell\}.
\end{equation}
The next lemma shows that rejections in clean intervals count exactly the refined false rejections, whereas all refined true rejections are characterized by the set $\mathcal K$. %

\begin{lemma}\label{lem:classification}
It holds that
$V^\sigma_t=\sum_{\ell=1}^{t+1}N_\ell$ and $K^\sigma_t=|\mathcal K|-2$ for $t\in\T$.
\end{lemma}
\begin{proof}
For $r\in\N$ and $t\in\T$, we say that a rejection time \(\tau_r\) is \textit{charged} if there exists \(\ell\in\{1,\ldots,t+1\}\) such that $\omega\in G_\ell$ and $\lambda_\ell<\tau_r\le \beta_\ell$.
By construction, for any $\ell=1,\dots,t+1$,
\[
        N_\ell
        =
        \sum_{r\in\N}
        \id_{G_\ell}
        \id_{\{\lambda_\ell<\tau_r\le\beta_\ell\}}.
\]
By Lemma~\ref{lem:decomposition}(iv), the genuine clean intervals are ordered and disjoint. Hence,
$\sum_{\ell=1}^{t+1}N_\ell$ is the number of charged rejections up to time $t$.

(i) We first show that every charged rejection $\tau_{r_1}\in(\lambda_{\ell_1},\beta_{\ell_1}]$, for some $r_1\in\N$ and $\ell_1\in\{1,\ldots,t+1\}$, is false under the refined law. We prove the following statement for all $\ell=1,\dots,t+1$ by induction:
\begin{equation}\label{eq:lem6_invariant}
\lambda_\ell=\mu_{\lambda_\ell}~~\mbox{on }G_\ell,~~\mbox{where }
\mu_s=\max\left\{u\in\{\tau_{R_{s}},\ldots,s\}:Z_u=L_s\right\},~~s=0,\dots,t.
\end{equation}
Indeed, \eqref{eq:lem6_invariant} holds for $\ell=1$ because $\lambda_1=0$. Suppose that \eqref{eq:lem6_invariant} holds for $\ell\in\{1,\ldots,t\}$. We have $\theta_\ell\wedge U_\ell<\infty$ on $G_{\ell+1}$. If $\theta_\ell<U_\ell$, then $\theta_\ell$ is a rejection time and thus $\mu_{\theta_\ell}=\theta_\ell=\lambda_{\ell+1}$. If $\theta_\ell>U_\ell$, then $Z_{U_\ell}=L_{U_\ell}$ and hence $\mu_{U_\ell}=U_\ell=\lambda_{\ell+1}$.
Therefore, \eqref{eq:lem6_invariant} is true for all $\ell=1,\dots,t+1$.

We claim that the block-minimum time \(\sigma_{r_1}\) of the rejected block in \eqref{eq:sigma} satisfies $ \sigma_{r_1}\ge \lambda_{\ell_1}$.
Indeed, if the previous reset time \(\tau_{r_1-1}\) satisfies
\(\tau_{r_1-1}\ge\lambda_{\ell_1}\), then trivially
\(\sigma_{r_1}\ge\tau_{r_1-1}\ge\lambda_{\ell_1}\).  If
\(\tau_{r_1-1}<\lambda_{\ell_1}\), then \eqref{eq:lem6_invariant} shows that $\lambda_{\ell_1}$ is the last time where the running minimum $L_{\lambda_{\ell_1}}$ occurs in the time interval $(\tau_{r_1-1},\lambda_{\ell_1}]$. It is clear that $\sigma_{r_1}\ge \lambda_{\ell_1}$ by definition.
Consequently,
$(\sigma_{r_1},\tau_{r_1}]\subseteq(\lambda_{\ell_1},\beta_{\ell_1}].$ By Lemma \ref{lem:decomposition}(ii), the rejection at \(\tau_{r_1}\) is false under the refined law.

(ii) Next, we show that every uncharged rejection is true under the refined law. If $\tau_{r_2}$ is uncharged for some $r_2\in\N$, then there exists $\ell_2\in\{1,\ldots,t\}$, such that $\tau_{r_2}\in(\beta_{\ell_2}, \lambda_{\ell_2+1}]$. By \eqref{eq:reject_minimum} and \eqref{eq:stretch_begin}, we have $\tau_{r_2}=\theta_{\ell_2}=\lambda_{\ell_2+1}<U_{\ell_2}$. Hence, $Z_s>L_{\beta_{\ell_2}}$ for all $s\in (\beta_{\ell_2}, \tau_{r_2})$. This implies that $\sigma_{r_2}<a(\lambda_{\ell_2})\le \tau_{r_2}$. Therefore, the rejection at $\tau_{r_2}$ is true. Arguments (i) and (ii) yield the assertion $V^\sigma_t=\sum_{\ell=1}^{t+1}N_\ell$ for all $t\in\T$.

(iii) It remains to characterize refined true rejections. For $\ell\in\{1,\ldots,t\}$, if $\theta_\ell<U_\ell$, then $\omega\in G_{\ell+1}$ and $\beta_\ell<\lambda_{\ell+1}=\theta_\ell$ by \eqref{eq:stretch_begin} and Lemma \ref{lem:decomposition}(iv). Hence, the rejection at $\theta_{\ell}$ is uncharged and is thus true under the refined law by argument (ii). Conversely, if a rejection $\tau_{r_3}$, for some $r_3\in\N$, is true under the refined law, it is uncharged by argument (i). Hence, we have by argument (ii) that $\tau_{r_3}=\theta_{\ell_3}<U_{\ell_3}$ for some $\ell_3\in\{1,\ldots,t\}$ and thus $\ell_3\in\mathcal K$. The above argument implies that $K^\sigma_t=|\mathcal K|-2$.
\end{proof}

The following lemma shows that the initial value of every genuine clean interval is between $0$ and $1$. Moreover, if there is no true rejection between two consecutive genuine clean intervals, then the initial capital of the later clean interval does not rise above the terminal capital of the previous clean interval.

\begin{lemma}
\label{lem:epoch-drop}
The following statements are true.
\begin{enumerate}[(i)]
\item For every \(\ell\in\{1,\ldots,t+1\}\), $0\le X_\ell\le1$. 
\item For every \(\ell\in\{1,\ldots,t\}\), $X_{\ell+1}\le Q_\ell$ on $G_\ell\cap G_{\ell+1}\cap\{U_\ell<\theta_\ell\}$.
\end{enumerate}
\end{lemma}
\begin{proof}
(i) For any $\ell\in\{1,\ldots,t+1\}$, on $G^c_{\ell}$, \(X_\ell=0\) by definition. On $G_{\ell}$, $X_\ell=Z_{\lambda_\ell}\ge 0$. We prove the upper bound \(X_\ell\le1\) on $G_\ell$ for any $\ell\in\{1,\ldots,t+1\}$.  First, we have $X_1=Z_0=1$. Next, for all $\ell\in\{1,\ldots,t\}$, on $G_{\ell+1}$, if $\theta_\ell<U_\ell$, then $X_{\ell+1}=Z_{\lambda_{\ell+1}}=Z_{\theta_\ell}=1$. If $\theta_\ell>U_\ell$, then $X_{\ell+1}=Z_{U_{\ell}}\le L_{\beta_\ell}\le 1$. Hence, $0\le X_\ell\le 1$.

(ii) Fix $\ell\in\{1,\ldots,t\}$. Suppose that the \((\ell+1)\)-st
genuine clean interval is generated by the event \(U_\ell<\theta_\ell\). It follows that
$$X_{\ell+1}=Z_{\lambda_{\ell+1}}=Z_{U_{\ell}}\le L_{\beta_\ell}\le Z_{\beta_\ell}=Q_{\ell}.$$
This completes the proof.
\end{proof}

The next lemma uses
Lemma~\ref{lem:epoch-drop} to turn the weighted capital drop
\(X_{\ell+1}\le Q_\ell\) into the bound used in the final FDR estimate.

\begin{lemma}
\label{lem:charge}
For $t\in\T$, we have $$\sum_{\ell=1}^{t+1}
        \frac{X_\ell-Q_\ell}{R_t-N_\ell+1}
        \le 1.$$
\end{lemma}

\begin{proof}
By Lemma~\ref{lem:classification}, the number of false rejections up to time $t\in\T$ is $V^\sigma_t=\sum^{t+1}_{\ell=1}N_\ell$. Hence, the number of true rejections up to $t$ is $K^\sigma_t=R_t-V^\sigma_t=R_t-\sum^{t+1}_{\ell=1}N_\ell$.

Because $\mathcal K$ in \eqref{eq:break_set} is a finite set, let $0=\pi_0<\pi_1<\cdots<\pi_J=t+1$
be the ordered elements of \(\mathcal K\), where $J=|\mathcal K|-1$. By Lemma \ref{lem:classification}, $J=K^\sigma_t+1$. For \(j\in\{1,\ldots,J\}\), put
$\mathcal I_j=\{\pi_{j-1}+1,\ldots,\pi_j\}$.
For every \(j\in\{1,\ldots,J\}\) and $\ell\in \I_j\setminus\{\pi_j\}$, we have $U_\ell<\theta_\ell$ on $G_{\ell+1}$, and $X_{\ell+1}=0\le Q_\ell$ on $G^c_{\ell+1}$.
Therefore, Lemma~\ref{lem:epoch-drop}(ii) gives
$X_{\ell+1}\le Q_\ell$.

Fix $j\in\{1,\dots,J\}$. Put $\gamma_\ell=\gamma_\ell(j)=(R_t-N_\ell+1)^{-1}$ for $\ell\in\I_j$. We have
\begin{equation} \begin{aligned} \sum_{\ell=\pi_{j-1}+1}^{\pi_j}\gamma_\ell(X_\ell-Q_\ell)&= \gamma_{\pi_{j-1}+1}X_{\pi_{j-1}+1}+\sum_{\ell=\pi_{j-1}+2}^{\pi_j}\gamma_\ell X_\ell -\sum_{\ell=\pi_{j-1}+1}^{\pi_j-1}\gamma_\ell Q_\ell-\gamma_{\pi_j}Q_{\pi_j} \\
&\le \gamma_{\pi_{j-1}+1}X_{\pi_{j-1}+1}+\sum_{\ell=\pi_{j-1}+2}^{\pi_j}(\gamma_\ell-\gamma_{\ell-1})X_\ell -\gamma_{\pi_j}Q_{\pi_j} \\
&\le \gamma_{\pi_{j-1}+1}+\sum_{\ell=\pi_{j-1}+2}^{\pi_j}(\gamma_\ell-\gamma_{\ell-1})_+, \end{aligned} \label{eq:epoch-charge} 
\end{equation}
where the first inequality uses the fact that $X_{\ell+1}\le Q_\ell$ for all $\ell\in\I_j\setminus\{\pi_j\}$, and the second inequality uses the fact that $0\le X_\ell\le 1$ for all $\ell\in\I_j$ by Lemma \ref{lem:epoch-drop}(i) and $Q_{\pi_j}\ge 0$. Next, we claim that
\begin{equation}
\gamma_{\pi_{j-1}+1}+\sum_{\ell=\pi_{j-1}+2}^{\pi_j}(\gamma_\ell-\gamma_{\ell-1})_+\le \id_{\left\{\sum_{\ell\in\I_j}N_\ell=0\right\}} \frac{1}{R_t+1} + \sum_{\ell\in\mathcal I_j:\,N_\ell>0} \frac{1}{R_t-N_\ell+1}. \label{eq:epoch-variation-bound}
\end{equation}
Indeed, if $\sum_{\ell\in\I_j}N_\ell=0$, then $\gamma_\ell=(R_t+1)^{-1}$ for all $\ell\in\I_j$ and thus the left-hand side of \eqref{eq:epoch-variation-bound} equals $(R_t+1)^{-1}$; whereas if
$\ell^*_j=\min\{\ell\in\I_j:N_\ell>0\}<\infty$, then
$$\begin{aligned}
    \gamma_{\pi_{j-1}+1}+\sum_{\ell=\pi_{j-1}+2}^{\pi_j}(\gamma_\ell-\gamma_{\ell-1})_+&=\gamma_{\ell^*_j}+\sum_{\ell^*_j<\ell\le\pi_j:N_\ell>0}(\gamma_\ell-\gamma_{\ell-1})_+\\
    &\le \gamma_{\ell^*_j}+\sum_{\ell^*_j<\ell\le\pi_j:N_\ell>0}\gamma_\ell=\sum_{\ell^*_j\le\ell\le\pi_j:N_\ell>0}\gamma_\ell=\sum_{\ell\in\mathcal I_j:\,N_\ell>0} \frac{1}{R_t-N_\ell+1}.
\end{aligned}$$
Hence, \eqref{eq:epoch-variation-bound} holds.

If $V^\sigma_t=0$, then $\sum^{t+1}_{\ell=1}N_\ell=0$ and $R_t=K^\sigma_t=J-1$. Summing \eqref{eq:epoch-charge} and \eqref{eq:epoch-variation-bound} over \(j\in\{1,\ldots,J\}\) gives
\[ \sum_{\ell=1}^{t+1} \frac{X_\ell-Q_\ell}{R_t-N_\ell+1} \le \frac{R_t+1}{R_t+1} = 1. \]
If \(V^\sigma_t>0\), then at least one index block $\I_j$ contains a false rejection (i.e.~$\ell^*_j<\infty$). Thus, because \(J=K^\sigma_t+1\), at most \(K^\sigma_t\) index blocks have no false rejections (i.e.~$\sum_{\ell\in\I_j}N_\ell=0$).
Moreover, for \(N_\ell\ge1\), because $N_\ell(R_t-N_\ell+1)-R_t=(R_t-N_\ell)(N_\ell-1)\ge 0$, we have $(R_t-N_\ell+1)^{-1}\le N_\ell/R_t$. Therefore, summing \eqref{eq:epoch-variation-bound} over \(j\in\{1,\ldots,J\}\),
\[ \sum_{\ell=1}^{t+1} \frac{X_\ell-Q_\ell}{R_t-N_\ell+1} \le \frac{K^\sigma_t}{R_t+1}+\frac{V^\sigma_t}{R_t} \le1 . \]
This completes the proof.
\end{proof}

\subsubsection{Conclusion of the proof}

\begin{proof}[Proof of Theorem~\ref{thm:localized-main}]
Fix \(t\in\T\). Let
\[
        \mathbf E^A=(E_u)_{u=1,\dots,t,\,u\notin \H_0} ~~\mbox{and}~~\mathcal A=\sigma(\mathbf E^A).
\]
Hence, $\mathcal A$ is the \(\sigma\)-field generated by the alternative e-values up to time
\(t\). For $\mathcal L(\mathbf E^A)$-almost every realization $\mathbf e=(e_u)_{u=1,\dots,t,\,u\notin \H_0}$, all alternative e-values are fixed. By Assumption \ref{ass:ind}, the null e-values $(E_u)_{u=1,\dots,t,\,u\in \H_0}$ are mutually independent and are independent of $\mathbf E^A$. Hence, the conditional law of $(E_1,\dots,E_t)$ given $\mathbf E^A=\mathbf e$ is the product law
\[
\mathbb Q_{\mathbf e} = \bigotimes_{u=1}^t \mu_u^{\mathbf e},\quad\mbox{where}~~\mu_u^{\mathbf e} = \begin{cases} \mathcal L(E_u), & u\in\H_0,\\ \delta_{e_u}, & u\notin\H_0. \end{cases}
\]
It follows that $\E^{\mathbb Q_{\mathbf e}}[E_u]\le 1$ for all $u\in\H_0\cap\{1,\ldots,t\}$. We first prove the bound under the conditional law $\mathbb Q_{\mathbf e}$.

Fix \(\ell\in\{1,\ldots,t+1\}\). For \(s=0,\dots,t\), let \(\Pi_s=\otimes^t_{u=s+1}\mu^{\mathbf e}_u\) be the product law of
\((E_{s+1},\ldots,E_t)\) under $\mathbb Q_{\mathbf e}$. Define the random kernel
\begin{equation*}
        g_\ell(c,\omega)
        =
        \sum_{s=0}^{t}\id_{\{\lambda_\ell=s\}}(\omega)
        \int
        \frac{1}{R_s(\omega)+A_{t-b(s)}(c,\mathbf u)+1}
        \,\Pi_{b(s)}(\mathrm{d}\mathbf u),~~c\in[0,1/\alpha],~\omega\in\Omega,
\end{equation*}
where $A_{t-b(s)}(c,\mathbf u)$ is as in \eqref{eq:reset-maps} with $s=0,\dots,t$, $c\in[0,1/\alpha]$, and $\mathbf u=(u_{b(s)+1},\dots,u_t)\in[0,\infty]^{t-b(s)}$. By Lemma \ref{lem:monotone}, the mapping $c\mapsto A_{t-b(s)}(c,\mathbf u)$ is increasing on $[0,1/\alpha]$. By Lemma \ref{lem:decomposition}(i), \(\lambda_\ell\) is a stopping time and $R_s$ is $\mathcal G_s$-measurable for $s=0,\dots,t$.  Therefore,
\((c,\omega)\mapsto g_\ell(c,\omega)\) is bounded, jointly \(\mathcal B([0,1/\alpha])\otimes\mathcal G_{\lambda_\ell}\)-measurable, and decreasing in $c$.
By Lemma~\ref{lem:decomposition}(iii), \(X_\ell\) is  \(\mathcal G_{\lambda_\ell}\)-measurable, and $Q_\ell,N_\ell$ are \(\mathcal G_{\beta_\ell}\)-measurable.  Because
Lemma~\ref{lem:epoch-drop}(i) gives \(0\le X_\ell\le 1\), applying
Lemma~\ref{lem:clean-start} with \(S=\lambda_\ell\),
\(X=X_\ell\), \(Q=Q_\ell\), \(N=N_\ell\), and
\(w(c,\omega)=g_\ell(c,\omega)\) gives
\begin{equation*}
        \E^{\mathbb Q_{\mathbf e}}\left[
        g_\ell(Q_\ell)\bigl(Q_\ell+(\alpha^{-1}-1)N_\ell\bigr)
        \,\middle|\,
        \mathcal G_{\lambda_\ell}
        \right]
        \le
        X_\ell
        \,\E^{\mathbb Q_{\mathbf e}}\left[
        g_\ell(Q_\ell)
        \,\middle|\,
        \mathcal G_{\lambda_\ell}
        \right]~~\mbox{$\mathbb Q_{\mathbf e}$-a.s.},
\end{equation*}
which is equivalent to
\begin{equation}
        \E^{\mathbb Q_{\mathbf e}}\left[
        N_\ell g_\ell(Q_\ell)
        \,\middle|\,
        \mathcal G_{\lambda_\ell}
        \right]
        \le
        \frac{\alpha}{1-\alpha}
        \E^{\mathbb Q_{\mathbf e}}\left[
        (X_\ell-Q_\ell)g_\ell(Q_\ell)
        \,\middle|\,
        \mathcal G_{\lambda_\ell}
        \right]~~\mbox{$\mathbb Q_{\mathbf e}$-a.s.}
        \label{eq:stretch-fdr-bound}
\end{equation}

On the actual path,
\begin{equation*}
        R_t
        =
        R_{\lambda_\ell}
        +N_\ell
        +A_{t-\beta_\ell}
        (Q_\ell,E_{\beta_\ell+1},\ldots,E_t).
\end{equation*}
Hence,
\begin{equation}
        \frac{N_\ell}{R_t\vee1}
        \le
        \frac{N_\ell}{R_{\lambda_\ell}
        +A_{t-\beta_\ell}(Q_\ell,E_{\beta_\ell+1},\ldots,E_t)+1}=:Y_\ell,
        \label{eq:fraction-to-suffix}
\end{equation}
because \(R_t\vee1\ge R_{\lambda_\ell}
+A_{t-\beta_\ell}(Q_\ell,E_{\beta_\ell+1},\ldots,E_t)+1\) whenever
\(N_\ell\ge1\), and both sides vanish when \(N_\ell=0\).

We claim that
\begin{equation} \label{eq:suffix-average-general} \E^{\mathbb Q_{\mathbf e}}\left[Y_\ell\,\middle|\,\mathcal G_{\beta_\ell}\right] = N_\ell g_\ell(Q_\ell) \qquad\mbox{$\mathbb Q_{\mathbf e}$-a.s.}
\end{equation}
Indeed, let \(H\) be an arbitrary bounded \(\mathcal G_{\beta_\ell}\)-measurable random variable. On the event \(\{\lambda_\ell=s\}\), we have $\beta_\ell=b(s)$ and $R_{\lambda_\ell}=R_s$.
Hence, using the independence of \((E_{b(s)+1},\ldots,E_t)\) from \(\mathcal G_{b(s)}\), we obtain
\[ \begin{aligned}
\E^{\mathbb Q_{\mathbf e}}[H Y_\ell] &= \sum_{s=0}^{t} \E^{\mathbb Q_{\mathbf e}}\left[ H\id_{\{\lambda_\ell=s\}} \frac{N_\ell}{ R_s+ A_{t-b(s)} \bigl(Q_\ell,(E_{b(s)+1},\ldots,E_t)\bigr) +1} \right] \\ &= \sum_{s=0}^{t} \E^{\mathbb Q_{\mathbf e}}\left[ H\id_{\{\lambda_\ell=s\}} N_\ell \int \frac{1}{ R_s+ A_{t-b(s)}(Q_\ell,\mathbf u) +1} \,\Pi_{b(s)}(\mathrm d\mathbf u) \right]= \E^{\mathbb Q_{\mathbf e}}\left[ H N_\ell g_\ell(Q_\ell) \right].
\end{aligned} \]
Therefore, \eqref{eq:suffix-average-general} follows.
Similarly,
\begin{equation}
        \E^{\mathbb Q_{\mathbf e}}\left[
        \frac{X_\ell-Q_\ell}{R_{\lambda_\ell}
        +A_{t-\beta_\ell}(Q_\ell,E_{\beta_\ell+1},\ldots,E_t)+1}
        \,\middle|\,
        \mathcal G_{\beta_\ell}
        \right]
        =
        (X_\ell-Q_\ell)g_\ell(Q_\ell)~~\mbox{$\mathbb Q_{\mathbf e}$-a.s.}
        \label{eq:suffix-average-charge}
\end{equation}
Combining \eqref{eq:fraction-to-suffix}, \eqref{eq:suffix-average-general},
\eqref{eq:suffix-average-charge}, and \eqref{eq:stretch-fdr-bound} yields
\begin{equation}
        \E^{\mathbb Q_{\mathbf e}}\left[
        \frac{N_\ell}{R_t\vee1}
        \right]
        \le
        \frac{\alpha}{1-\alpha}
        \E^{\mathbb Q_{\mathbf e}}\left[
        \frac{X_\ell-Q_\ell}{R_t-N_\ell+1}
        \right].
        \label{eq:one-stretch-fdr}
\end{equation}
Summing \eqref{eq:one-stretch-fdr} over all $\ell=1,\dots,t+1$ and using
Lemmas~\ref{lem:classification} and \ref{lem:charge},
\[
\begin{aligned}
        \E^{\mathbb Q_{\mathbf e}}\left[\frac{V^\sigma_t}{R_t\vee1}\right]
        =
        \sum_{\ell=1}^{t+1}
        \E^{\mathbb Q_{\mathbf e}}\left[\frac{N_\ell}{R_t\vee1}\right]                    \le
        \frac{\alpha}{1-\alpha}
        \E^{\mathbb Q_{\mathbf e}}\left[
        \sum_{\ell=1}^{t+1}
        \frac{X_\ell-Q_\ell}{R_t-N_\ell+1}
        \right]
        \le
        \frac{\alpha}{1-\alpha}.
\end{aligned}
\]
Integrating over $\mathbf e$ gives $\FDR^{\mathrm{ref}}_t\le \alpha/(1-\alpha)$.
Together with the trivial bound \(\FDR^{\mathrm{ref}}_t\le1\), this proves the result.
\end{proof}

\subsection{Proofs of Theorems \ref{thm:FDR_barrier} and \ref{cor:rejection-count-and-mfdr}}
\label{app:sec:proof_thm2}

In this section, we prove Theorem \ref{thm:FDR_barrier} and Theorem \ref{cor:rejection-count-and-mfdr}. To control the FDR, it remains to control the number of consecutive false discoveries in between two adjacent true rejections. The key observation is that, once we condition on the past up to a true rejection, every additional false rejection before the next true one must come from a reported refined interval that remains entirely under the null and nevertheless crosses the threshold $1/\alpha$.

More precisely, we denote the stopping times where refined true rejections occur by \begin{equation*}%
    \kappa^\sigma_j=\min\{s\in\T:K^\sigma_s=j\}\quad~~\mbox{for}~j\in\N_0.
\end{equation*}
Between each pair of adjacent true rejections $\kappa^\sigma_{j-1}$ and $\kappa^\sigma_j$, we count the number of false positives by
\begin{align}
    N^\sigma_j(t)=R_{\kappa^\sigma_j\wedge t}-R_{\kappa^\sigma_{j-1}\wedge t}
-\id_{\{\kappa^\sigma_{j}\le t\}}=\begin{cases}
        0&\text{ if }t<\kappa^\sigma_{j-1};\\
        R_t-R_{\kappa^\sigma_{j-1}}&\text{ if }\kappa^\sigma_{j-1}\leq t<\kappa^\sigma_{j};\\
        R_{\kappa^\sigma_j}-R_{\kappa^\sigma_{j-1}}-1&\text{ if }t\geq \kappa^\sigma_{j},
    \end{cases}\quad j\in\N,~t\in\T.
    \label{eq:num_false}
    \end{align}
It is clear that each $N^\sigma_j(t)$ is $\mathcal{F}_{\kappa^\sigma_j\wedge t}$-measurable and 
$V^\sigma_t=\sum_{j\in\N}N^\sigma_j(t)$.
For any $r\in\N$ and $t\in\T$, define the event that the $r$-th rejection is false by
\begin{align}
       A^\sigma_{r}(t)=\{\tau_{r}\le t,~\nu^\sigma_{r}>\tau_{r}\}\in\mathcal F_{\tau_{r}}.\label{eq:Ar}
\end{align}

In the following result, we prove a refined version of Ville's inequality. Consequently, each false rejection has a conditional probability at most $\alpha$.

\begin{lemma}\label{lem:ville}
For every $r\in\N$, $t\in\T$, and $\alpha\in(0,1)$, the sequential resetting procedure with $\gamma=1/\alpha$ satisfies $\p\!\left(A_r^\sigma(t)\,\middle|\,\mathcal F_{\tau_{r-1}}\right)\le \alpha$ a.s.
\end{lemma}
\begin{proof}
Define the sequence
\((\theta_{r,k})_{k\in\N_0}\) of increasing stopping times by $\theta_{r,0}=\tau_{r-1}$
and
\begin{equation}
\label{eq:theta_rk}
\theta_{r,k+1}=\min\left\{
u\in\T:\theta_{r,k}<u<\tau_r,~M^{[r]}_{u-\tau_{r-1}}
=\min_{0\le s\le u-\tau_{r-1}}M^{[r]}_s\right\},~~k\in\N_0.
\end{equation}
By definition, \((\theta_{r,k})_{k\in\N_0}\) describes the sequence of running minima in the block $r$.
It follows from \eqref{eq:sigma} that
\begin{align*}
    \sigma_r
    =
    \max\{\theta_{r,k}:\theta_{r,k}\le \tau_r\}\id_{\{\tau_r<\infty\}}
    +
    \tau_{r-1}\id_{\{\tau_r=\infty\}},
    \qquad r\in\N.
\end{align*}
For $k\in\N_0$, define
\[
G_{r,k}(t)=A_r^\sigma(t)\cap \bigcup_{\ell\ge k}\{\sigma_r=\theta_{r,\ell}\}.
\]
Then $G_{r,0}(t)=A_r^\sigma(t)$. Moreover, because $\theta_{r,k}\ge k$ a.s.~by \eqref{eq:theta_rk}, we have $\theta_{r,t+1}>t$ a.s., and thus $G_{r,t+1}(t)=\emptyset$. We claim that
\begin{equation}
\label{eq:thm2_claim}
\p\!\left(G_{r,k}(t)\,\middle|\,\mathcal F_{\theta_{r,k}}\right)
\le
\alpha\,M^{[r]}_{\theta_{r,k}-\tau_{r-1}}
\qquad \mbox{on } \{\theta_{r,k}\le t\}.
\end{equation}
We may assume $t<\infty$ because in the case of $\T=\N_0$, the sequence of events $(A^\sigma_r(t))_{t\in\T}$ is increasing and converges to $\{\tau_r<\infty,\ \nu_r^\sigma>\tau_r\}
=
\bigcup_{t\in\N_0} A_r^\sigma(t)$.  
We apply backward induction on $k=t+1,t,\dots,0$. For $k=t+1$, the claim is trivial because $G_{r,t+1}(t)=\emptyset$. Fix $k\in\{0,\dots,t\}$ and suppose that \eqref{eq:thm2_claim} holds with $k+1$ in place of $k$.
If $\theta_{r,k}>\tau_r\wedge t$, then $G_{r,k}(t)=\emptyset$. Therefore, it suffices to work on the event $\{\theta_{r,k}\le \tau_r\wedge t\}$.
Define
\begin{align*}
    \nu_{r,k}=\min\{u\in\T_{\geq \theta_{r,k}+1}:\p_u\notin\mathcal P_u\}\qquad\mbox{and}\qquad
\eta_{r,k}(t)=\tau_r\wedge \theta_{r,k+1}\wedge(\nu_{r,k}-1)\wedge t.
\end{align*}
Note that $\nu_{r,k}-1$ is a stopping time, because $\{\p_u\notin\mathcal P_u\}\in\mathcal F_{u-1}$. It then follows that $\eta_{r,k}(t)$ is also a stopping time.  
Because $\eta_{r,k}(t)<\nu_{r,k}$ a.s., conditioned on $\mathcal F_{\theta_{r,k}}$, the process
\[
L^{[r,k]}_u=
\frac{M^{[r]}_{\theta_{r,k}+u-\tau_{r-1}}}
     {M^{[r]}_{\theta_{r,k}-\tau_{r-1}}},
\qquad
u=0,1,\dots,(\eta_{r,k}(t)-\theta_{r,k})\id_{\{\theta_{r,k}<\infty\}},
\]
is a supermartingale with initial value $1$, where we set $L^{[r,k]}_u=1$ whenever $M^{[r]}_{\theta_{r,k}-\tau_{r-1}}=0$. Hence,
\begin{equation}
\label{eq:thm2_optional}
\E\!\left[
M^{[r]}_{\eta_{r,k}(t)-\tau_{r-1}}
\,\middle|\,
\mathcal F_{\theta_{r,k}}
\right]
\le
M^{[r]}_{\theta_{r,k}-\tau_{r-1}}
\qquad \mbox{on } \{\theta_{r,k}\le t\}.
\end{equation}

Next, we claim that
\begin{equation}
\label{eq:thm2_key}
\p\!\left(G_{r,k}(t)\,\middle|\,\mathcal F_{\eta_{r,k}(t)}\right)
\le
\alpha\,M^{[r]}_{\eta_{r,k}(t)-\tau_{r-1}}
\qquad \mbox{on } \{\theta_{r,k}\le t\}.
\end{equation}
We split into four cases.
\begin{itemize}
    \item If $\eta_{r,k}(t)=\tau_r$, then $M^{[r]}_{\eta_{r,k}(t)-\tau_{r-1}}
    =
    M^{[r]}_{\tau_r-\tau_{r-1}}
    \ge {1}/{\alpha}$,
    and thus \eqref{eq:thm2_key} holds trivially.

    \item If $\eta_{r,k}(t)=t$ and $\eta_{r,k}(t)\neq \tau_r$, then $\tau_r>t$, and hence $$\p(G_{r,k}(t)\,|\,\mathcal F_{\eta_{r,k}(t)})\leq \p(A^\sigma_r(t)\,|\,\mathcal F_{\eta_{r,k}(t)})=0.$$ 
    Therefore, \eqref{eq:thm2_key} holds.

    \item If $\eta_{r,k}(t)=\theta_{r,k+1}\le \tau_r\wedge(\nu_{r,k}-1)\wedge t$, then a later selected running-minimum time has already occurred before rejection, and thus $G_{r,k}(t)=G_{r,k+1}(t)$. 
    Using the induction hypothesis (\eqref{eq:thm2_claim} applied with $k+1$), we obtain
    \[
    \p\!\left(G_{r,k}(t)\,\middle|\,\mathcal F_{\eta_{r,k}(t)}\right)
    =
    \p\!\left(G_{r,k+1}(t)\,\middle|\,\mathcal F_{\theta_{r,k+1}}\right)
    \le
    \alpha\,M^{[r]}_{\theta_{r,k+1}-\tau_{r-1}}
    =
    \alpha\,M^{[r]}_{\eta_{r,k}(t)-\tau_{r-1}}.
    \]

    \item Otherwise, we must have
   $\eta_{r,k}(t)=\nu_{r,k}-1<\theta_{r,k+1}\wedge\tau_r\wedge t.$ 
    On this event,
    \begin{equation}
    \label{eq:thm2_inclusion}
    G_{r,k}(t)\subseteq \{\theta_{r,k+1}<\tau_r\}\cap G_{r,k+1}(t).
    \end{equation}
    Indeed, if $\theta_{r,k+1}\ge \tau_r$, then no later selected running-minimum time occurs before the rejection at $\tau_r$, and thus $\sigma_r=\theta_{r,k}$. Since $\nu_{r,k}\le \tau_r$ on the present event, this implies $\nu_r^\sigma\le \tau_r$, contradicting the event $A_r^\sigma(t)\supseteq G_{r,k}(t)$.

    By \eqref{eq:thm2_inclusion} and the induction hypothesis \eqref{eq:thm2_claim}, we have
    \[
    \begin{aligned}
    \p\!\left(G_{r,k}(t)\,\middle|\,\mathcal F_{\nu_{r,k}-1}\right)
    &\le
    \E\!\left[
    \id_{\{\theta_{r,k+1}<\tau_r\}}
    \p\!\left(G_{r,k+1}(t)\,\middle|\,\mathcal F_{\theta_{r,k+1}}\right)
    \,\middle|\,\mathcal F_{\nu_{r,k}-1}
    \right] \\
    &\le
    \alpha\,
    \E\!\left[
    \id_{\{\theta_{r,k+1}<\tau_r\}}
    M^{[r]}_{\theta_{r,k+1}-\tau_{r-1}}
    \,\middle|\,\mathcal F_{\nu_{r,k}-1}
    \right].
    \end{aligned}
    \]
    By the running-minimum property \eqref{eq:theta_rk},
    \[
    M^{[r]}_{\theta_{r,k+1}-\tau_{r-1}}
    \le
    M^{[r]}_{\nu_{r,k}-1-\tau_{r-1}}
    \qquad \mbox{on } \{\nu_{r,k}-1<\theta_{r,k+1}<\infty\}.
    \]
    Hence
    \[
    \p\!\left(G_{r,k}(t)\,\middle|\,\mathcal F_{\nu_{r,k}-1}\right)
    \le
    \alpha\,M^{[r]}_{\nu_{r,k}-1-\tau_{r-1}}
    =
    \alpha\,M^{[r]}_{\eta_{r,k}(t)-\tau_{r-1}}.
    \]
    This gives \eqref{eq:thm2_key} in the last case.
\end{itemize}
Combining the four cases yields \eqref{eq:thm2_key}. 

Taking conditional expectation with respect to $\mathcal F_{\theta_{r,k}}$ and using \eqref{eq:thm2_optional}, we get
\[
\p\!\left(G_{r,k}(t)\,\middle|\,\mathcal F_{\theta_{r,k}}\right)
=
\E\!\left[
\p\!\left(G_{r,k}(t)\,\middle|\,\mathcal F_{\eta_{r,k}(t)}\right)
\,\middle|\,
\mathcal F_{\theta_{r,k}}
\right]
\le
\alpha\,
\E\!\left[
M^{[r]}_{\eta_{r,k}(t)-\tau_{r-1}}
\,\middle|\,
\mathcal F_{\theta_{r,k}}
\right]
\le
\alpha\,M^{[r]}_{\theta_{r,k}-\tau_{r-1}}
\]
on $\{\theta_{r,k}\le t\}$. This proves \eqref{eq:thm2_claim}. Taking $k=0$ and using $\theta_{r,0}=\tau_{r-1}$ and $M^{[r]}_0=1$, we obtain
\[
\p\!\left(A_r^\sigma(t)\,\middle|\,\mathcal F_{\tau_{r-1}}\right)
=
\p\!\left(G_{r,0}(t)\,\middle|\,\mathcal F_{\tau_{r-1}}\right)
\le \alpha.
\]
The proof is complete.
\end{proof}

The following result is useful for deriving the FDR bound, which indicates that $N^\sigma_j(t)$ for $j\in\N$ is stochastically dominated by a geometric distribution.

\begin{lemma}\label{lem:ville_n}
    For any $j\in\N$, $n\in\N_0$, $t\in\T$, and $\alpha\in(0,1)$, the sequential resetting procedure with $\gamma=1/\alpha$ satisfies
    $$\p\left(N^\sigma_j(t)\ge n~|~\mathcal F_{\kappa^\sigma_{j-1}\wedge t}\right)\le \alpha^n~~\mbox{a.s.}$$
\end{lemma}
\begin{proof}
    If $\kappa^\sigma_{j-1}>t$, then $N^\sigma_j(t)=0$ by definition \eqref{eq:num_false}. The assertion holds automatically. The case of $n=0$ is trivial. Hence, we only work on the event $\{\kappa^\sigma_{j-1}\le t\}$ and the case of $n\in\N$.
    
Write $\rho_{j-1}=R_{\kappa^\sigma_{j-1}}$ for simplicity and use $A^\sigma_r(t)$ as in \eqref{eq:Ar} for $r\in\N$. By \eqref{eq:num_false}, we have the representation
    \begin{align}
        N^\sigma_j(t)=\id_{\{\kappa^\sigma_{j-1}\leq t\}}\sum_{\ell\geq 1}\id_{\left\{\bigcap_{m=1}^\ell A^\sigma_{\rho_{j-1}+m}(t)\right\}}.\label{eq:rep}
    \end{align}
    On the event $\{\rho_{j-1}=r\}$, we have $\tau_{r+n-1}\ge\kappa^\sigma_{j-1}$. Hence, we apply Lemma \ref{lem:ville} recursively and get, for each $n\geq 1$, 
    $$\begin{aligned}
        \p\left(\bigcap^n_{m=1}A^\sigma_{r+m}(t)~\middle\vert~\mathcal F_{\kappa^\sigma_{j-1}}\right)&=\E\left[\id_{\{\bigcap^{n-1}_{m=1}A^\sigma_{r+m}(t)\}}\p\left(A^\sigma_{r+n}(t)~|~\mathcal F_{\tau_{r+n-1}}\right)~\middle\vert~\mathcal F_{\kappa^\sigma_{j-1}}\right]\\
        &\le \alpha \p\left(\bigcap^{n-1}_{m=1}A^\sigma_{r+m}(t)~\middle\vert~\mathcal F_{\kappa^\sigma_{j-1}}\right)~~~~\mbox{a.s.}
    \end{aligned}$$
    By induction, for each $n\geq 1$, 
    $$\p\left(\bigcap^n_{m=1}A^\sigma_{r+m}(t)~\middle\vert~\mathcal F_{\kappa^\sigma_{j-1}}\right)\le \alpha^n~~~~\mbox{a.s.}$$
    It follows from \eqref{eq:rep} that
    $$\p(N^\sigma_j(t)\ge n~|~\mathcal F_{\kappa^\sigma_{j-1}})=\sum_{r\in\N_0}\id_{\{\rho_{j-1}=r\}}\p\left(\bigcap^n_{m=1}A^\sigma_{r+m}(t)~\middle\vert~\mathcal F_{\kappa^\sigma_{j-1}}\right)\le\alpha^n~~~~\mbox{a.s.}$$
    The proof is complete.
\end{proof}

Use $\{\xi_j\}_{j\in\N}$ as in Theorem \ref{thm:FDR_barrier}. We have the following dominance result in a multi-dimensional setting.

\begin{lemma}\label{lem:mono_com}
    For any $t\in\T$, $m\in\N$, and bounded increasing function $\Phi:\N_0^m\to\R$, we have
    $$\E[\Phi(N^\sigma_1,\dots,N^\sigma_m)]\le \E[\Phi(\xi_1,\dots,\xi_m)],$$
    where we write $(N^\sigma_1,\dots,N^\sigma_m)=(N^\sigma_1(t),\dots,N^\sigma_m(t))$.
\end{lemma}
\begin{proof}
    We prove the assertion by induction. 
    The case of $m=1$ is straightforward by Lemma \ref{lem:ville_n}. For $m>1$, suppose that the assertion holds for $m-1$. Define
    \begin{equation}\label{eq:lem2_phi}
        g(n_1,\dots,n_{m-1})=\E[\Phi(n_1,\dots,n_{m-1},\xi_m)]~~~~\mbox{for}~n_1,\dots,n_{m-1}\in\N_0.
    \end{equation}
    It is clear that $g$ is also bounded and increasing. Define
    $$G(k)=\Phi(N^\sigma_1,\dots,N^\sigma_{m-1},k),~~k\in\N_0.$$
    We note that $G(k)$ is a bounded random variable measurable with respect to $\mathcal F_{\kappa^\sigma_{m-1}\wedge t}$, and
    $$G(k)=G(0)+\sum_{s\in\N}(G(s)-G(s-1))\id_{\{s\le k\}},~~k\in\N_0.$$
    Hence, we have by Lemma \ref{lem:ville_n} and \eqref{eq:lem2_phi} that
    $$\begin{aligned}
        \E\left[\Phi(N^\sigma_1,\dots,N^\sigma_m)~|~\mathcal F_{\kappa^\sigma_{m-1}\wedge t}\right]&=\E\left[G(N^\sigma_m)~|~\mathcal F_{\kappa^\sigma_{m-1}\wedge t}\right]\\
        &=G(0)+\sum_{s\in\N}(G(s)-G(s-1))\p\left(N^\sigma_m\ge s~|~\mathcal F_{\kappa^\sigma_{m-1}\wedge t}\right)\\
        &\le G(0)+\sum_{s\in\N}(G(s)-G(s-1))\alpha^s\\
        &=G(0)+\sum_{s\in\N}(G(s)-G(s-1))\p(\xi_m\ge s)=g(N^\sigma_1,\dots,N^\sigma_{m-1})~~\mbox{a.s.}
    \end{aligned}$$
    It follows that
    $$\E[\Phi(N^\sigma_1,\dots,N^\sigma_m)]\le \E[g(N^\sigma_1,\dots,N^\sigma_{m-1})]\le \E[g(\xi_1,\dots,\xi_{m-1})]=\E[\Phi(\xi_1,\dots,\xi_m)].$$
    The proof is complete.
\end{proof}

Use $(S_k)_{k\in\N}$ as in Theorem \ref{thm:FDR_barrier}. Write
\begin{equation*}
    X^{\mathrm{ind}}_{\infty}=\inf_{k\in\N}\frac{k}{k+S_k}.
\end{equation*}
The following result gives the bounds for the i.i.d.~benchmark case.

\begin{lemma}\label{lemma:E bound}
    For any $\alpha\in(0,1)$,
    $$
1-\frac{\alpha}{1-\alpha}\log\frac1\alpha
\le
\E[X^{\mathrm{ind}}_{\infty}]
\le
1-\frac{\alpha}{(1-\alpha)^2}\log\frac1\alpha+\frac{\alpha}{1-\alpha}. 
$$
\end{lemma}

\begin{proof}
For the upper bound, we let $K$ be a $\sigma(\{\xi_j\})$-measurable random variable, defined as the smallest index such that $\xi_k\geq 1$. Clearly, $\p(K=k)=\alpha(1-\alpha)^{k-1}$. Thus
\begin{align*}
    \E[X^{\mathrm{ind}}_{\infty}]&\leq \E\left[\frac{K}{K+S_K}\right]\\
    &\leq \sum_{k\in\N}\alpha(1-\alpha)^{k-1}\frac{k}{k+1}=1-\sum_{u\geq 2}\frac{\alpha(1-\alpha)^{u-2}}{u}=1-\frac{\alpha}{(1-\alpha)^2}\log\frac1\alpha+\frac{\alpha}{1-\alpha},
\end{align*}
   where we have used the identity $\sum_{u\ge 1}q^u/u=-\log(1-q)$ with $q=1-\alpha\in(0,1)$. 

For the lower bound, write $Y=\sup_{n\in\N}{S_n}/{n}$. 
We have
\[
    1-X^{\mathrm{ind}}_{\infty}=1-\inf_{n\in\N}\frac{n}{n+S_n}
    =
    \frac{Y}{1+Y}.
\]
Moreover, $\E[\xi_1]=\alpha/(1-\alpha)$. By the maximal inequality in \citet[equation (5)]{S15}, for every $x>0$,
\[
    \P(Y\ge x)
    \le
    \min\left\{1,\frac{\alpha}{(1-\alpha)x}\right\}.
\]
Therefore,
\begin{align*}
    \E\left[1-X^{\mathrm{ind}}_{\infty}\right]
    =
    \int_0^1
    \P\left(1-X^{\mathrm{ind}}_{\infty}\ge u\right)\d u
    &=
    \int_0^1
    \P\left(Y\ge \frac{u}{1-u}\right)\d u\\
    &\le
    \alpha+
    \int_\alpha^1
    \frac{\alpha}{1-\alpha}\frac{1-u}{u}\,\d u=
    \frac{\alpha}{1-\alpha}\log\frac1\alpha.
\end{align*}
This gives the desired lower bound.
\end{proof}

\begin{proof}[Proof of Theorem \ref{thm:FDR_barrier}]
For any $m\in\N$, define
$$
\begin{aligned}
    &\Psi_m(n_1,\dots,n_{m+1})=1-\id_{\{n_1=0\}}\left(\min_{k\in\{1,\dots,m\}}\frac{k}{k+\sum^{k+1}_{j=2}n_j}\right)\\
    \mbox{and}~~&\Psi(n_1,n_2,\dots)=1-\id_{\{n_1=0\}}\left(\inf_{k\in\N}\frac{k}{k+\sum^{k+1}_{j=2}n_j}\right).
\end{aligned}$$
It follows that $\Psi_m$ and $\Psi$ are bounded and increasing in each argument, and for any sequence $(n_k)_{k\geq 1}$, $\Psi_m(n_1,\dots,n_{m+1})\uparrow \Psi(n_1,n_2,\dots)$ as $m\to\infty$. Now we aim to prove that 
\begin{equation}\label{eq:thm1_dom}
    \frac{V^\sigma_t}{R_t\vee 1}\le \Psi(N^\sigma_1,N^\sigma_2,\dots)~~\mbox{for all $t\in\T$ a.s.}
\end{equation}
Indeed, on $\{K^\sigma_t=0\}$, we have
$$\frac{V^\sigma_t}{R_t\vee 1}=\id_{\{N^\sigma_1\ge 1\}}\le \Psi(N^\sigma_1,N^\sigma_2,\dots)~~~~\mbox{a.s.}$$
On $\{K^\sigma_t\ge 1,~N^\sigma_1\ge 1\}$, we have
$$\frac{V^\sigma_t}{R_t\vee 1}\le 1=\Psi(N^\sigma_1,N^\sigma_2,\dots)~~~~\mbox{a.s.}$$
On $\{K^\sigma_t\ge 1,~N^\sigma_1=0\}$, we have
$$\frac{V^\sigma_t}{R_t\vee 1}=1-\frac{K^\sigma_t}{K^\sigma_t+\sum^{K^\sigma_t+1}_{j=2}N^\sigma_j(t)}\le 1-\inf_{k\in\N}\left(\frac{k}{k+\sum^{k+1}_{j=2}N^\sigma_j(t)}\right)=\Psi(N^\sigma_1,N^\sigma_2,\dots)~~\mbox{a.s.}$$
Thus, \eqref{eq:thm1_dom} holds.
Because $t\mapsto N^\sigma_j(t)$ is increasing for $j\in\N$ and $\Psi$ is increasing, we apply the monotone convergence theorem and obtain
$$\mathrm{FDR}^{\mathrm{ref}}_\infty=\E\left[\sup_{t\in\T}\frac{V^\sigma_t}{R_t\vee 1}\right]\le \E\left[\sup_{t\in\T}\Psi(N^\sigma_1(t),N^\sigma_2(t),\dots)\right]=\sup_{t\in\T}\E\left[\Psi(N^\sigma_1(t),N^\sigma_2(t),\dots)\right].$$
Hence, by Lemma \ref{lem:mono_com} and \eqref{eq:FDR_infty},
\begin{align}
    \begin{split}
        \mathrm{FDR}^{\mathrm{ref}}_{t}\leq \mathrm{FDR}^{\mathrm{ref}}_{\infty}%
    &\le \sup_{t\in\T}\E\left[\Psi(N^\sigma_1(t),N^\sigma_2(t),\dots)\right]\\
    &=\sup_{t\in\T}\lim_{m\to\infty}\E[\Psi_m(N^\sigma_1(t),\dots,N^\sigma_{m+1}(t))]\\
    &\le \lim_{m\to\infty}\E[\Psi_m(\xi_1,\dots,\xi_{m+1})]\\
    &=\E[\Psi(\xi_1,\xi_2,\dots)]\\
    &=1-\E\left[\id_{\{\xi_1=0\}}\inf_{k\in\N}\frac{k}{k+\sum^{k+1}_{j=2}\xi_j}\right]=1-(1-\alpha)\E\left[\inf_{k\in\N}\frac{k}{k+S_k}\right].
    \end{split}\label{eq:long}
\end{align}
This proves the first and second inequalities of \eqref{eq:local_FDR_bound}.  
Applying Lemma \ref{lemma:E bound} yields the third inequality of \eqref{eq:local_FDR_bound}.
\end{proof}

We need the following elementary lemma for the proof of Theorem \ref{cor:rejection-count-and-mfdr}.

\begin{lemma}\label{lem:geometric-mgf-bound}
For $\alpha\in(0,1)$, define
\[
    f_\alpha(t)
    =
    \log\frac{1-\alpha}{1-\alpha\mathrm{e}^{(1-\alpha)t}}
    -\alpha t.
\]
Then
\[
    f_\alpha(t)
    \le \frac{\alpha t^2}{2-(1+\alpha)t}\qquad\text{for}
    \qquad 0\le t<\frac{2}{1+\alpha}.
\]
\end{lemma}

\begin{proof}
The stated range is valid since
$-\log\alpha>2(1-\alpha)/(1+\alpha)$.
The function $t\mapsto
\mathrm{e}^{-(1-\alpha)t/2}
-\alpha\mathrm{e}^{(1-\alpha)t/2}$
is convex on the stated range, so its tangent at $t=0$ gives
\[
    \mathrm{e}^{-(1-\alpha)t/2}
    -\alpha\mathrm{e}^{(1-\alpha)t/2}
    \ge
    (1-\alpha)\left(1-\frac{1+\alpha}{2}t\right).
\]
Consequently,
\[
    f_\alpha''(t)
    \le \frac{\alpha}{(1-(1+\alpha)t/2)^2}
    \le \frac{\mathrm{d}^2}{\mathrm{d}t^2}
         \left(\frac{\alpha t^2}{2-(1+\alpha)t}\right).
\]
Both $f_\alpha$ and the proposed upper bound have value and
derivative zero at $t=0$. Integrating twice proves the claim.
\end{proof}

\begin{proof}[Proof of Theorem \ref{cor:rejection-count-and-mfdr}]
For (i), each refined true report contains an alternative position, and the
refined report intervals are disjoint. Hence, $K_T^\sigma\le N_a$.
By \eqref{eq:num_false}, $N_j^\sigma(T)=0$ for
$j>K_T^\sigma+1$, so
\[
    R_T=K_T^\sigma+V_T^\sigma
    \le N_a+\sum_{j=1}^{N_a+1}N_j^\sigma(T).
\]
\sloppy For $n\in\N_0$, applying Lemma~\ref{lem:mono_com} to the bounded
coordinatewise increasing function $\Phi(x_1,\ldots,x_{N_a+1})
    =\id_{\{\sum_{j=1}^{N_a+1}x_j\ge n\}}$ yields 
\[
    \p(R_T\ge N_a+n)
    \le
    \p\!\left(\sum_{j=1}^{N_a+1}N_j^\sigma(T)\ge n\right)
    \le \p(S_{N_a+1}\ge n)=\sum^{N_a}_{j=0}\left(\begin{array}{c}
        N_a+n \\ j
    \end{array}\right) (1-\alpha)^j\alpha^{N_a+n-j}.
\]
Arbitrariness of $n$ proves \eqref{eq:rejection-count-dom}. The expectation bound \eqref{eq:expectation-dom} follows
from $\E[\xi_j]=\alpha/(1-\alpha)$.

To prove \eqref{eq:confidence-dom}, note that \eqref{eq:rejection-count-dom} and \cite[Lemma~2.6]{Janson18}, applied to
$S_{N_a+1}+N_a+1$, give, for every real $r\ge N_a$ and
$0<t<2/(1+\alpha)$,
\begin{align*}
    \p(R_T\ge r)
    &\le
    \mathrm{e}^{-(1-\alpha)tr}
    \left(
        \frac{(1-\alpha)\mathrm{e}^{(1-\alpha)t}}
             {1-\alpha\mathrm{e}^{(1-\alpha)t}}
    \right)^{N_a}\\
    &=
    \exp\left(
        -t\bigl((1-\alpha)r-N_a\bigr)+N_a f_\alpha(t)
    \right)\\
    &\le
    \exp\left(
        -t\bigl((1-\alpha)r-N_a\bigr)
        +\frac{N_a\alpha t^2}{2-(1+\alpha)t}
    \right),
\end{align*}
where the last inequality follows from
Lemma~\ref{lem:geometric-mgf-bound}.
For $x>0$, take
\[
    t=
    \frac{\sqrt{2x}}
         {\sqrt{N_a\alpha}+(1+\alpha)\sqrt{2x}/2}
\]
when $N_a>0$, and let $t\uparrow2/(1+\alpha)$ when $N_a=0$.
It follows that
\[
    \p\left(
        R_T\ge
        \frac{
            N_a+\sqrt{2N_a\alpha x}+(1+\alpha)x/2
        }{1-\alpha}
    \right)
    \le \mathrm{e}^{-x}.
\]
Taking $x=\log(1/\beta)$ proves \eqref{eq:confidence-dom}.

For (ii), recall that
$A_r^\sigma(T)=\{\tau_r\le T,\,\nu_r^\sigma>\tau_r\}$.
This event implies $\tau_{r-1}\le T$, and
$\{\tau_{r-1}\le T\}\in\mathcal F_{\tau_{r-1}}$.
Lemma~\ref{lem:ville} therefore gives
\[
    \p(A_r^\sigma(T))
    =\E\!\left[
        \id_{\{\tau_{r-1}\le T\}}
        \p\!\left(A_r^\sigma(T)\mid\mathcal F_{\tau_{r-1}}\right)
      \right]
    \le \alpha\,\p(\tau_{r-1}\le T).
\]
Summing over $r\in\N$ and using $\tau_0=0$ gives
\[
    \E[V_T^\sigma]
    =\sum_{r\in\N}\p(A_r^\sigma(T))
    \le\alpha\sum_{r\in\N}\p(\tau_{r-1}\le T)
    =\alpha\bigl(1+\E[R_T]\bigr).
\]
Finally, $V_T\le V_T^\sigma$ pathwise, which proves
\eqref{eq:mfdr-one-bound}.
\end{proof}

\section{Further experiments and details}
\label{sec:add_num}

\subsection{LLM watermark detection}\label{sec:LLM complement}

\subsubsection{Proof of Proposition \ref{prop:tournament-pivot}}

\begin{proof}[Proof of Proposition~\ref{prop:tournament-pivot}]
Fix $t$, and for notational convenience, all expectations below are conditional on $\mathcal F_{t-1}$. Under $H_t$, $S_t\mid\mathcal F_{t-1}\lawis\operatorname{Binom}(30,1/2)$. The randomized probability integral transform \citep[Proposition~2.1]{Ruschendorf09} therefore gives $Y_t\mid\mathcal F_{t-1}\lawis\operatorname{Unif}(0,1)$ and $\E[L_t]=1$.

We next prove the stochastic dominance. Under the Tournament watermark, write
\[
s_\ell(w)=\sum_{j=1}^{\ell}g_{t,j}(w),
\qquad s_0(w)=0.
\]
Given $\mathcal F_{t-1}$ and all bits $(g_{t,j}(v))_{1\le j<\ell,\,v\in\mathcal W}$, the transition formula for two-token Tournament sampling \citep[Supplementary Appendix~E, Corollary~14 and Theorem~15]{DathathriEtAl24} gives
\begin{align}
    P_t^{(\ell)}(w)
=P_t^{(\ell-1)}(w)
\left(
1+g_{t,\ell}(w)
-\sum_{v\in\mathcal W}P_t^{(\ell-1)}(v)g_{t,\ell}(v)
\right),
\qquad w\in\mathcal W.\label{eq:transition}
\end{align}
Conditioning on $g_{t,\ell}(w)=b$, where $b\in\{0,1\}$, yields
\[
\begin{aligned}
&\E\!\left[
\sum_{v\in\mathcal W}P_t^{(\ell-1)}(v)g_{t,\ell}(v)
\,\middle|\,
\mathcal F_{t-1},
(g_{t,j}(v))_{1\le j<\ell,\ v\in\mathcal W},
g_{t,\ell}(w)=b
\right]\\
&\quad=
P_t^{(\ell-1)}(w)b
+\frac12\sum_{\substack{v\in\mathcal W\\v\ne w}}P_t^{(\ell-1)}(v)
=
P_t^{(\ell-1)}(w)b+\frac12\left(1-P_t^{(\ell-1)}(w)\right).
\end{aligned}
\]
Substituting this into \eqref{eq:transition} gives
\[
\begin{aligned}
&\E\!\left[
P_t^{(\ell)}(w)
\,\middle|\,
\mathcal F_{t-1},
(g_{t,j}(v))_{1\le j<\ell,\ v\in\mathcal W},
g_{t,\ell}(w)=b
\right]=
\begin{cases}
\frac{P_t^{(\ell-1)}(w)}2(1+P_t^{(\ell-1)}(w)),&\text{ if }b=0;\\
\frac{P_t^{(\ell-1)}(w)}2(3-P_t^{(\ell-1)}(w)),&\text{ if }b=1.
\end{cases}
\end{aligned}
\]
Therefore, for every increasing function $\varphi$, 
\begin{equation}
\label{eq:tournament-layer-identity}
\begin{aligned}
&\E\!\left[
P_t^{(\ell)}(w)
\varphi\!\left(s_{\ell-1}(w)+g_{t,\ell}(w)\right)
\,\middle|\,
\mathcal F_{t-1},
(g_{t,j}(v))_{1\le j<\ell,\ v\in\mathcal W}
\right]\\
&\quad =
\frac{P_t^{(\ell-1)}(w)}{4}
\Big(
\left(1+P_t^{(\ell-1)}(w)\right)\varphi(s_{\ell-1}(w))
+\left(3-P_t^{(\ell-1)}(w)\right)\varphi(s_{\ell-1}(w)+1)
\Big)\\
&\quad \ge
\frac{P_t^{(\ell-1)}(w)}{2}
\Big(\varphi(s_{\ell-1}(w))+\varphi(s_{\ell-1}(w)+1)\Big).
\end{aligned}
\end{equation}
Let $B_\ell\lawis\operatorname{Binom}(\ell,1/2)$ be independent of $\mathcal F_{t-1}$. Summing \eqref{eq:tournament-layer-identity} over $w\in\mathcal W$, and inducting on $\ell$ using the increasing map $s\mapsto[\varphi(s)+\varphi(s+1)]/2$ yields
\[
\begin{aligned}
\E\!\left[\sum_{w\in\mathcal W}P_t^{(\ell)}(w)\varphi(s_\ell(w))\right]
&\ge
\E\!\left[
\sum_{w\in\mathcal W}P_t^{(\ell-1)}(w)
\frac{\varphi(s_{\ell-1}(w))+\varphi(s_{\ell-1}(w)+1)}2
\right]\\
&\ge
\frac12\E[\varphi(B_{\ell-1})]
+\frac12\E[\varphi(B_{\ell-1}+1)]
=\E[\varphi(B_\ell)].
\end{aligned}
\]
At $\ell=30$, the left-hand side is $\E[\varphi(S_t)]$, which proves the claimed conditional stochastic dominance of $S_t$.

It remains to show the nontriviality. Taking $\varphi(s)=s$ in the equality in \eqref{eq:tournament-layer-identity}, and using $s_\ell(w)=s_{\ell-1}(w)+g_{t,\ell}(w)$, gives
\[
\begin{aligned}
&\E\!\left[
P_t^{(\ell)}(w)s_\ell(w)
\,\middle|\,
\mathcal F_{t-1},
(g_{t,j}(v))_{1\le j<\ell,\ v\in\mathcal W}
\right]\\
&\quad=
\frac{P_t^{(\ell-1)}(w)}4
\Big(
(1+P_t^{(\ell-1)}(w))s_{\ell-1}(w)
+(3-P_t^{(\ell-1)}(w))(s_{\ell-1}(w)+1)
\Big)\\
&\quad=
P_t^{(\ell-1)}(w)s_{\ell-1}(w)
+\frac14P_t^{(\ell-1)}(w)\left(3-P_t^{(\ell-1)}(w)\right).
\end{aligned}
\]
Summing over $w\in\mathcal W$ and using the tower property and $\sum_{w\in\mathcal W}P_t^{(\ell-1)}(w)=1$ give
\[
\begin{aligned}
\E\!\left[
\sum_{w\in\mathcal W}P_t^{(\ell)}(w)s_\ell(w)
\right]
&=
\E\!\left[
\E\!\left[
\sum_{w\in\mathcal W}P_t^{(\ell)}(w)s_\ell(w)
\,\middle|\,
\mathcal F_{t-1},
(g_{t,j}(v))_{1\le j<\ell,\ v\in\mathcal W}
\right]
\right]\\
&=
\E\!\left[
\sum_{w\in\mathcal W}P_t^{(\ell-1)}(w)s_{\ell-1}(w)
\right]
+\frac14\E\!\left[
\sum_{w\in\mathcal W}P_t^{(\ell-1)}(w)
\left(3-P_t^{(\ell-1)}(w)\right)
\right]\\
&=
\E\!\left[
\sum_{w\in\mathcal W}P_t^{(\ell-1)}(w)s_{\ell-1}(w)
\right]+\frac12
+\frac14\E\!\left[
1-\sum_{w\in\mathcal W}\left(P_t^{(\ell-1)}(w)\right)^2
\right].
\end{aligned}
\]
Since $\sum_{w\in\mathcal W}(P_t^{(\ell-1)}(w))^2\le1$, the last term is nonnegative. Summing this recursion over $\ell=1,\ldots,30$, using $s_0(w)=0$ and $P_t^{(0)}=\pi_t$, and retaining the contribution from $\ell=1$ gives
\[
\E[S_t]
\ge 15+\frac14\left(1-\sum_{w\in\mathcal W}\pi_t(w)^2\right)>15
\]
whenever $\pi_t$ is not concentrated on a single token.
Next, write
\[
h(k,v)=F_0(k-1)+vp_0(k),\qquad k=0,\ldots,30.
\]
For every $v\in(0,1)$, $h(k,v)$ is strictly increasing in $k$. By the conditional stochastic dominance just proved, we can couple $S^\ast$ and $B^\ast$ so that conditional on $\mathcal F_{t-1}$, $S^\ast\laweq S_t$, $B^\ast\lawis\operatorname{Binom}(30,1/2)$, and $S^\ast\geq B^\ast$ a.s. 
Let $V\lawis\operatorname{Unif}(0,1)$ be independent of $\sigma(\mathcal F_{t-1},S^\ast,B^\ast)$. The same randomized probability integral transform gives
\(
h(B^\ast,V)\mid\mathcal F_{t-1}
\lawis\operatorname{Unif}(0,1).
\)
Since $\mathcal L(h(S^\ast,V)\mid\mathcal F_{t-1})=\mathcal L(Y_t\mid\mathcal F_{t-1})$ and $\mathbb P(h(S^\ast,V)\ge h(B^\ast,V)\mid\mathcal F_{t-1})=1$, this proves the claimed conditional stochastic dominance of $Y_t$.

When $\pi_t$ is nondegenerate, $\E[S_t]>15=\E[B^\ast]$ implies $\mathbb P(S^\ast>B^\ast\mid\mathcal F_{t-1})>0$, and hence $\mathbb P(h(S^\ast,V)>h(B^\ast,V)\mid\mathcal F_{t-1})>0$. Strict monotonicity of $y\mapsto-\log(1-y)$ then gives
\[
\E[L_t]
>
\E[-\log(1-h(B^\ast,V))]
=1.
\]
This proves the nontriviality.
\end{proof}

\subsubsection{A four-interval stress test}
\label{subsec:four-interval}

In this section, we repeat the LLM watermark experiments from Section \ref{sec:llm-experiment} with a changed null--alternative sequence containing four disjoint alternative (watermarked) intervals. This experiment serves as a stress test for repeated refinement. We retain the same 20 prompts but generate 10 continuations per prompt and temperature, giving 200 paths per temperature. Otherwise, the setup remains the same.

The four watermarked intervals each contain 50 positions and are separated by
25 standard positions:
\[
 A_T=\{51,\ldots,100\}\cup\{126,\ldots,175\}
 \cup\{201,\ldots,250\}\cup\{276,\ldots,325\}.
\]
All remaining positions in $\{1,\dots,600\}$ are null; see Figure \ref{fig:llm-four-schedule}.

We first consider the Gumbel--max watermark. Table \ref{tab:llm-four-results} reports our main results. 
The resetting procedure can still effectively detect most watermarked positions.  At
\(\vartheta=1\), power reaches \(0.899\) and IoU reaches \(0.797\); performance
is weaker but remains substantial at \(\vartheta=0.75\).  The empirical uniform realized
FDP is below \(0.6\%\) in both cases, well below the theoretical \(10\%\) bound. Figure \ref{fig:llm-primary-trajectories2} displays a typical trajectory.

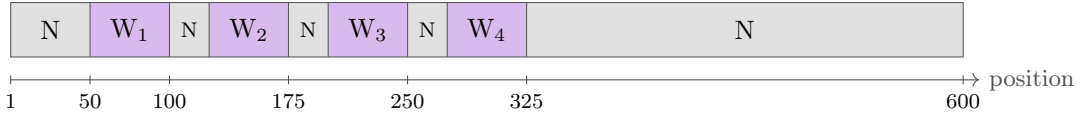
\begin{figure}[t]
\centering
\begin{tikzpicture}[x=0.21cm,y=0.72cm,font=\small]
  \fill[black!12] (0,0) rectangle (5,1);
  \fill[purple!38] (5,0) rectangle (10,1);
  \fill[black!12] (10,0) rectangle (12.5,1);
  \fill[purple!38] (12.5,0) rectangle (17.5,1);
  \fill[black!12] (17.5,0) rectangle (20,1);
  \fill[purple!38] (20,0) rectangle (25,1);
  \fill[black!12] (25,0) rectangle (27.5,1);
  \fill[purple!38] (27.5,0) rectangle (32.5,1);
  \fill[black!12] (32.5,0) rectangle (60,1);
  \draw[black!65] (0,0) rectangle (60,1);
  \foreach \x in {5,10,12.5,17.5,20,25,27.5,32.5}{
    \draw[black!65] (\x,0)--(\x,1);
  }
  \node at (2.5,.5) {$\mathrm N$};
  \node at (7.5,.5) {$\mathrm W_1$};
  \node[font=\scriptsize] at (11.25,.5) {$\mathrm N$};
  \node at (15,.5) {$\mathrm W_2$};
  \node[font=\scriptsize] at (18.75,.5) {$\mathrm N$};
  \node at (22.5,.5) {$\mathrm W_3$};
  \node[font=\scriptsize] at (26.25,.5) {$\mathrm N$};
  \node at (30,.5) {$\mathrm W_4$};
  \node at (46.25,.5) {$\mathrm N$};
  \draw[->,black!70] (0,-.42)--(61,-.42)
       node[right,font=\footnotesize]{position};
  \foreach \x/\lab in {0/1,5/50,10/100,17.5/175,25/250,32.5/325,60/600}{
    \draw[black!70] (\x,-.34)--(\x,-.50);
    \node[below,font=\scriptsize] at (\x,-.50) {$\lab$};
  }
\end{tikzpicture}
\caption{The four-interval null--alternative sequence.  Grey segments marked $\mathrm N$ are
standard, and purple segments marked $\mathrm W$ are watermarked.}
\label{fig:llm-four-schedule}
\end{figure}

\begin{table}[t]
\centering
\small
\setlength{\tabcolsep}{5pt}
\caption{{Results for the resetting WA e-process in the
four-interval experiment {under the Gumbel--max watermark}, with rejection threshold 
$\gamma=49$.}}
\label{tab:llm-four-results}
\begin{tabular}{@{}ccccccc@{}}
\toprule
$\vartheta$ & Mean $\lambda$ & Mean reports & Power & Coverage FDP & IoU
& Uniform realized FDP\\
\midrule
0.75 & 0.366 & 10.98 & 0.744 & 0.161 & 0.643 & 0.0059\\
1 & 0.452 & 27.48
  & \green{\textbf{0.899}}
  & \green{\textbf{0.122}}
  & \green{\textbf{0.797}}
  & \green{\textbf{0.0012}}\\
\bottomrule
\end{tabular}
\end{table}

\begin{figure}[t]
\centering
\includegraphics[width=\textwidth]{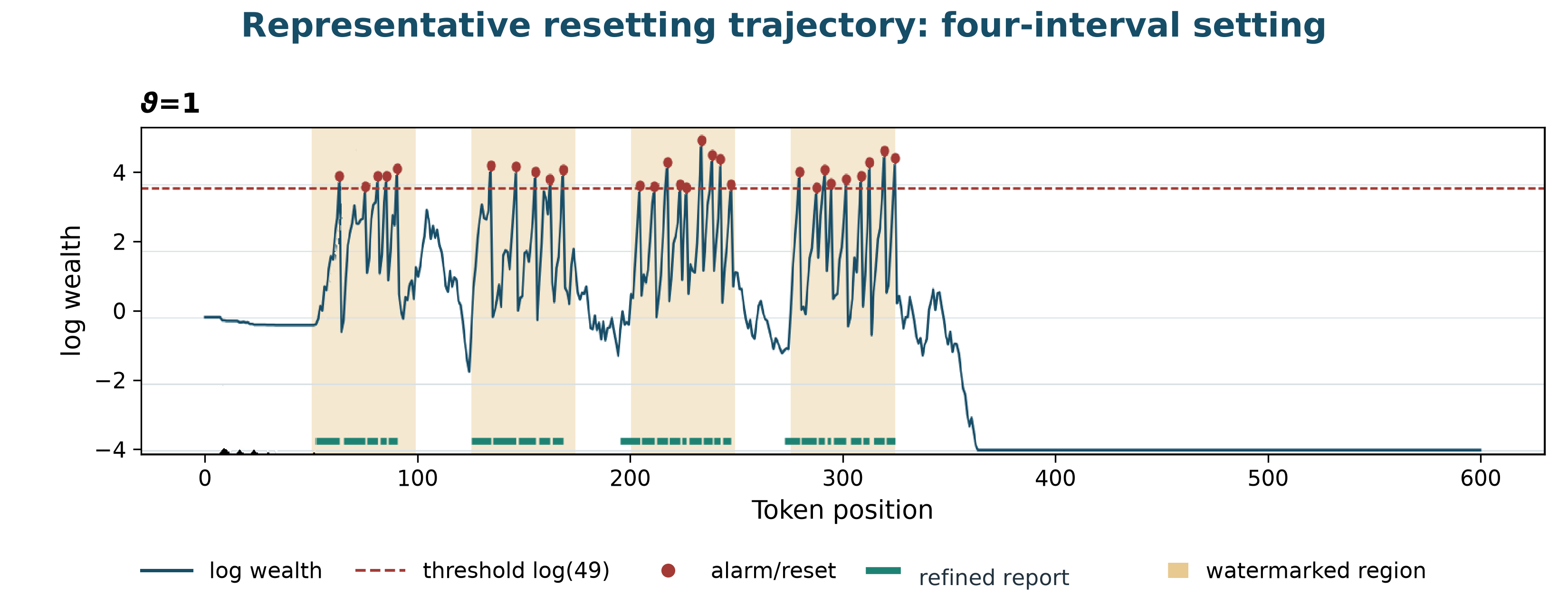}
\caption{{Representative resetting WA e-process trajectories for the four-interval null--alternative sequence {under the Gumbel--max watermark} at rejection threshold $\gamma=49$. Values below $-4$ are clipped from the drawing.}}
\label{fig:llm-primary-trajectories2}
\end{figure}

{Next, we repeat the four-interval stress design under the Tournament
watermark.  At $\vartheta=0.75$, power is $0.722$ and IoU is $0.612$; at
$\vartheta=1$, they increase to $0.900$ and $0.784$, respectively.  The
corresponding uniform realized FDP values are $0.0106$ and $0.0013$.}

\begin{table}[t]
\centering
\small
\setlength{\tabcolsep}{5pt}
\caption{{Results for the resetting WA e-process in the four-interval
Tournament experiment, with rejection threshold $\gamma=49$.}}
\label{tab:llm-tournament-four-results}
\begin{tabular}{@{}ccccccc@{}}
\toprule
$\vartheta$ & Mean $\lambda$ & Mean reports & Power & Coverage FDP & IoU
& Uniform realized FDP\\
\midrule
0.75 & 0.347 & 8.68 & 0.722 & 0.173 & 0.612 & 0.0106\\
1 & 0.450 & 23.05
  & \green{\textbf{0.900}}
  & \green{\textbf{0.139}}
  & \green{\textbf{0.784}}
  & \green{\textbf{0.0013}}\\
\bottomrule
\end{tabular}
\end{table}

\begin{figure}[t]
\centering
\includegraphics[width=0.94\textwidth]{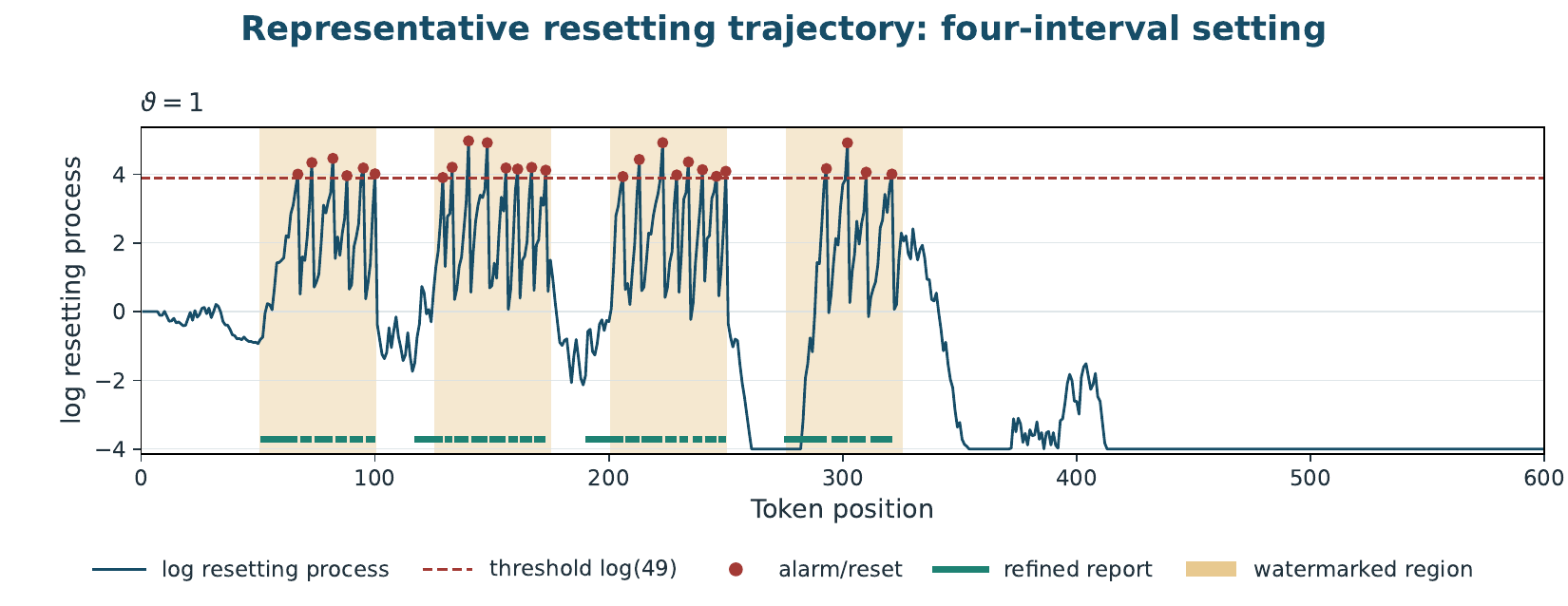}
\caption{{Representative resetting WA e-process trajectory for the
four-interval Tournament experiment at rejection threshold $\gamma=49$.
Values below $-4$ are clipped from the drawing.}}
\label{fig:llm-tournament-four-trajectory}
\end{figure}

\subsubsection{Comparison of different e-processes}
\label{subsec:llm-comparators}

In addition to the weight-adaptive e-process in Section
\ref{sec:refreshing exp}, we further compare three families of e-processes from \cite[Section 4]{SWZ26}.
{We apply the same collection of e-processes separately to the Gumbel--max and Tournament pivots.}
\begin{itemize}

\item \emph{Nonadaptive e-processes.}
Compared to \eqref{eq:WA}, a nonadaptive e-process fixes the betting fraction $\lambda$ before observing the path:
\[
 M_t^{(\lambda)}=\prod_{s=1}^t(1-\lambda+\lambda L_s),
 \qquad \lambda\in(0,1].
\]
Under this fresh-key construction, conditional on the token path and the keys at alternative positions, the null scores are i.i.d.~$\mathrm{Unif}(0,1)$ and independent of the alternative scores. Therefore, for fixed $\lambda$, the multipliers satisfy Assumption~\ref{ass:ind}. In our experiments, we consider $\lambda\in\{0.10,0.25,0.50,0.75\}$.

\item \emph{Online Grenander.}
The online Grenander (OG) e-process \citep{HWR26} estimates a decreasing density for $1-Y_t$ from the past observations. 
Let $\mathcal D$ be the set of decreasing density functions on $[0,1]$. Following \cite[Section 4.3]{SWZ26}, we define $f_1^{\mathrm{OG}}=1$, and for $t>1$,
\begin{align*}
f_t^{\mathrm{OG}}
&\in\argmax_{f\in\mathcal D}
\left(
\sum_{s<t}\log f(1-Y_s)
+\lambda\log f(1) 
\right),\\
M_t^{\mathrm{OG}}
&=\prod_{s=1}^t f_s^{\mathrm{OG}}(1-Y_s),
\end{align*}
where we set $\lambda=1/2$. 
The term \(\lambda\log f(1)\) prevents the fitted density from taking the value zero.

\item \emph{Average e-process.}
The average e-process of \cite[Section 4.4]{SWZ26} assigns equal weight to the WA
and OG e-processes:
\[
 M_t^{\mathrm{avg}}
 =\frac{1}{2}M_t^{\mathrm{WA}}+\frac{1}{2}M_t^{\mathrm{OG}}.
\]

\end{itemize}

As in Section \ref{sec:refreshing exp}, we consider the resetting versions of the above e-processes.  Again, a reset restarts the e-process at one but keeps the previous
observations for the optimization procedure. As a consequence, the nonadaptive e-processes satisfy
Assumption \ref{ass:ind}, but the other two do not. 
All procedures share the common threshold $\gamma=49$.  For WA, OG, and the
average process, Theorem \ref{thm:FDR_barrier} controls the uniform FDR
by $(1+\log 49)/49<0.10$.  The nonadaptive processes have a stronger fixed-time bound, using Theorem \ref{thm:localized-main}: $\mathrm{FDR}_t^{\mathrm{ref}}\le 1/48\approx0.0208.$
Nevertheless, we still use the same threshold for every procedure to keep the comparisons at scale.

{Table \ref{tab:llm-process-comparison} and Figure
\ref{fig:llm-process-tradeoff} compare the resulting procedures and their
power--coverage-FDP tradeoffs.}
All seven procedures are applied
to the same paths, with the same resetting rule and
threshold \(\gamma=49\).  The nonadaptive process with \(\lambda=0.10\) has the
highest power, but includes substantially more null positions.  WA gives the
most balanced performance: it has the highest IoU at \(\vartheta=1\), near-top IoU at \(\vartheta=0.75\), and low coverage FDP at both
temperatures.  OG and the larger fixed values of \(\lambda\) include less null
text, but generally have lower detection power. All empirical uniform realized FDP values are below our theoretical bounds.

\begin{table}[t]
\centering
\small
\setlength{\tabcolsep}{3.2pt}
\caption{Comparison of the seven resetting e-processes {under the Gumbel--max watermark} at $\gamma=49$, based
on 500 paths per temperature. Green bold values indicate the best-performing values
for the metrics.}
\label{tab:llm-process-comparison}
\begin{tabular}{@{}ccccccc@{}}
\toprule
Process & Mean $\lambda$ & Power & Coverage FDP & IoU
& \shortstack{Realized\\FDP}
& \shortstack{Uniform realized\\FDP}\\
\midrule
\multicolumn{7}{c}{$\gamma=49$, $\vartheta=0.75$}\\
Average & -- & 0.723 & 0.042 & 0.699 & 0.0004 & 0.0006\\
OG & -- & 0.652 & 0.035 & 0.637 & 0.0006 & 0.0043\\
WA & 0.424 & 0.751 & 0.034 & 0.731 & 0.0001 & 0.0002\\
Nonadaptive $\lambda=0.10$
  & 0.100 & \green{\textbf{0.844}} & 0.144 & 0.732
  & \green{\textbf{0.0000}} & \green{\textbf{0.0000}}\\
Nonadaptive $\lambda=0.25$
  & 0.250 & 0.810 & 0.102 & \green{\textbf{0.737}} & 0.0001 & 0.0001\\
Nonadaptive $\lambda=0.50$
  & 0.500 & 0.753 & 0.045 & 0.726 & 0.0008 & 0.0027\\
Nonadaptive $\lambda=0.75$
  & 0.750 & 0.693 & \green{\textbf{0.025}} & 0.680 & 0.0014 & 0.0072\\
\midrule
\multicolumn{7}{c}{$\gamma=49$, $\vartheta=1$}\\
Average & -- & 0.877 & 0.039 & 0.847 & 0.0003 & 0.0064\\
OG & -- & 0.800 & 0.023 & 0.786 & 0.0003 & 0.0142\\
WA & 0.457 & 0.927 & 0.035 & \green{\textbf{0.897}} & 0.0004 & 0.0064\\
Nonadaptive $\lambda=0.10$
  & 0.100 & \green{\textbf{0.962}} & 0.140 & 0.831
  & \green{\textbf{0.0000}} & \green{\textbf{0.0000}}\\
Nonadaptive $\lambda=0.25$
  & 0.250 & 0.951 & 0.092 & 0.868 & 0.0001 & 0.0001\\
Nonadaptive $\lambda=0.50$
  & 0.500 & 0.930 & 0.043 & 0.893 & 0.0007 & 0.0184\\
Nonadaptive $\lambda=0.75$
  & 0.750 & 0.913 & \green{\textbf{0.023}} & 0.894 & 0.0007 & 0.0206\\
\bottomrule
\end{tabular}
\end{table}

\begin{figure}[t]
\centering
\includegraphics[width=0.98\textwidth]{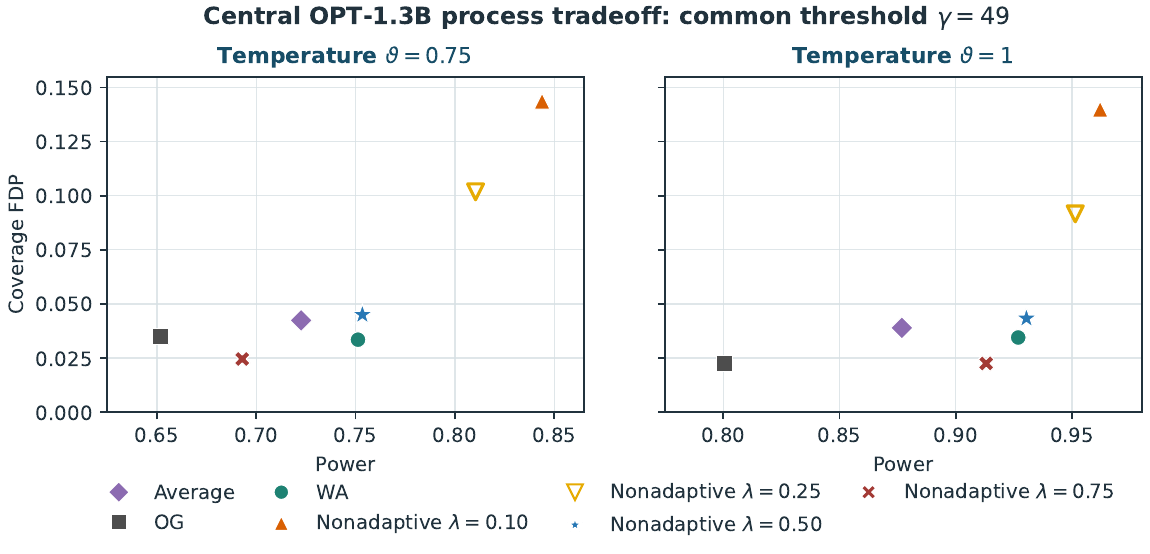}
\caption{{Power--coverage-FDP tradeoff {under the Gumbel--max watermark} at the common rejection threshold $\gamma=49$. }}
\label{fig:llm-process-tradeoff}
\end{figure}

{The Tournament experiment yields a similar tradeoff.  The nonadaptive
process with $\lambda=0.10$ has the highest power at both temperatures, while
$\lambda=0.75$ has the lowest coverage FDP.  WA has the highest IoU at
$\vartheta=1$, and $\lambda=0.25$ has the highest IoU at
$\vartheta=0.75$.}
\begin{table}[t]
\centering
\small
\setlength{\tabcolsep}{3.2pt}
\caption{Comparison of the seven resetting e-processes under the Tournament
watermark at $\gamma=49$, based on 500 paths per temperature. Green bold values
indicate the best-performing values for the metrics.}
\label{tab:llm-tournament-process-comparison}
\begin{tabular}{@{}ccccccc@{}}
\toprule
Process & Mean $\lambda$ & Power & Coverage FDP & IoU
& \shortstack{Realized\\FDP}
& \shortstack{Uniform realized\\FDP}\\
\midrule
\multicolumn{7}{l}{$\gamma=49$, $\vartheta=0.75$}\\
Average & -- & 0.723 & 0.043 & 0.699 & 0.0012 & 0.0090\\
OG & -- & 0.669 & 0.036 & 0.652 & 0.0016 & 0.0129\\
WA & 0.414 & 0.740 & 0.037 & 0.718 & 0.0012 & 0.0071\\
Nonadaptive $\lambda=0.10$
  & 0.100 & \green{\textbf{0.828}} & 0.144 & 0.719
  & \green{\textbf{0.0000}} & \green{\textbf{0.0000}}\\
Nonadaptive $\lambda=0.25$
  & 0.250 & 0.799 & 0.102 & \green{\textbf{0.727}}
  & 0.0005 & 0.0024\\
Nonadaptive $\lambda=0.50$
  & 0.500 & 0.736 & 0.047 & 0.709 & 0.0015 & 0.0189\\
Nonadaptive $\lambda=0.75$
  & 0.750 & 0.662 & \green{\textbf{0.029}} & 0.649
  & 0.0037 & 0.0172\\
\midrule
\multicolumn{7}{l}{$\gamma=49$, $\vartheta=1$}\\
Average & -- & 0.889 & 0.040 & 0.857 & 0.0003 & 0.0005\\
OG & -- & 0.829 & 0.025 & 0.812 & 0.0002 & 0.0042\\
WA & 0.455 & 0.925 & 0.036 & \green{\textbf{0.894}}
  & 0.0002 & 0.0022\\
Nonadaptive $\lambda=0.10$
  & 0.100 & \green{\textbf{0.957}} & 0.137 & 0.830
  & \green{\textbf{0.0000}} & \green{\textbf{0.0000}}\\
Nonadaptive $\lambda=0.25$
  & 0.250 & 0.949 & 0.094 & 0.863 & 0.0002 & 0.0002\\
Nonadaptive $\lambda=0.50$
  & 0.500 & 0.929 & 0.044 & 0.890 & 0.0004 & 0.0103\\
Nonadaptive $\lambda=0.75$
  & 0.750 & 0.908 & \green{\textbf{0.022}} & 0.890
  & 0.0006 & 0.0106\\
\bottomrule
\end{tabular}
\end{table}

\begin{figure}[t]
\centering
\includegraphics[width=0.98\textwidth]{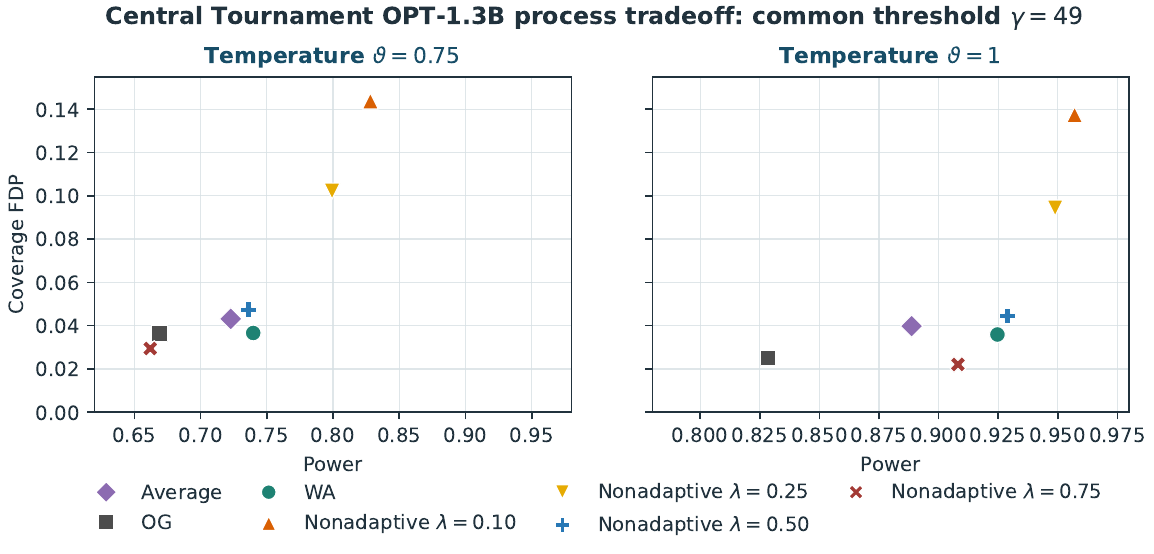}
\caption{{Power--coverage-FDP tradeoff under the Tournament watermark at
the common rejection threshold $\gamma=49$.}}
\label{fig:llm-tournament-process-tradeoff}
\end{figure}

\subsection{Comparison with e-BH and closed e-BH}
\label{subsec:ebh-comparison}
We compare the refined resetting procedures with the e-BH and closed e-BH procedures \citep{WR22,XSF25}.
Not all choices of the closed e-BH uniformly improve e-BH \citep[Proposition 3.4]{SW26}, and here we use the mean closed e-BH, which is a uniform improvement of e-BH shown by \cite{XSF25}. 
We use the Gaussian simulation setting in Section~\ref{sec:sim_ind}, with $T=1500$ and alternating periods of 200 null and 100 alternative observations, but vary the alternative mean over $\delta=1,2,3,4$.
Take e-values $E_t=\exp(\delta X_t-\delta^2/2)$, $t\in\T_{\ge 1}$, which are the likelihood ratios between $\mathrm N(\delta,1)$ and $\mathrm N(0,1)$.
Both the e-BH and closed e-BH procedures are applied to these individual e-values at the terminal horizon with FDR level $0.1$. The refined resetting procedure uses $\lambda=1/2$ and $\gamma=11$ to guarantee a general FDR bound of $0.1$ as in Theorem \ref{thm:localized-main}. All methods receive the same data in each of 1000 repetitions.

Table~\ref{tab:ebh_comparison} shows the comparison of the performance metrics across the three methods.
We recall that the realized FDP is the proportion of false reports: an all-null interval is counted as false for interval reporting, and a null point is counted as false for point reporting.
The e-BH and closed e-BH procedures reject individual data points, whereas the sequential resetting procedure can report both points and intervals. For the point interpretation of the sequential resetting procedure, we reject all time points in the union of the reported refined intervals; its realized FDP equals the coverage FDP in \eqref{eq:metrics} and is not controlled by our theoretical FDR guarantee; so this is a heuristic comparison. The two interpretations cover the same points and thus have the same power according to our definition.
The results show that the sequential resetting procedure rejects more points on average than the e-BH and closed e-BH procedures in all four experiments. Moreover, both e-BH procedures generally have little power until the alternative signal becomes strong enough, whereas the sequential resetting procedure keeps the power above $90\%$ throughout. The realized FDP is low for the e-BH procedures and the interval-reporting sequential resetting procedure.

\begin{table}[t]
    \centering
    \caption{Average performance metrics for e-BH, mean closed e-BH, and refined sequential resetting procedures over 1000 repetitions of independent normal data with alternating periods.}
    \label{tab:ebh_comparison}
    \small
    \setlength{\tabcolsep}{4pt}
    \begin{tabular}{c c c c c}
        \toprule
        Signal  & Method & Reports & Power & Realized FDP\\
        \midrule
 \multirow{4}{*}{$\delta=1$} & e-BH (points) & 0.000 & 0.0000 & 0.0000 \\
         & Mean closed e-BH (points) & 0.000 & 0.0000 & 0.0000 \\
        & Sequential resetting (points) & 515.786 & 0.9064 & 0.1083 \\
        & Sequential resetting (intervals) & 27.808 & 0.9064 & 0.0085 \\
        \midrule
        \multirow{4}{*}{$\delta=2$} & e-BH (points) & 0.043 & 0.0001 & 0.0000 \\
         & Mean closed e-BH (points) & 0.738 & 0.0015 & 0.0000 \\
         & Sequential resetting (points) & 473.191 & 0.9324 & 0.0146 \\
         & Sequential resetting (intervals) & 125.997 & 0.9324 & 0.0013 \\
        \midrule
         \multirow{4}{*}{$\delta=3$} & e-BH (points) & 288.487 & 0.5722 & 0.0083 \\
          & Mean closed e-BH (points) & 347.455 & 0.6833 & 0.0167 \\
         & Sequential resetting (points) & 480.437 & 0.9566 & 0.0044 \\
         & Sequential resetting (intervals) & 263.067 & 0.9566 & 0.0004 \\
        \midrule
        \multirow{4}{*}{$\delta=4$}  & e-BH (points) & 435.706 & 0.8675 & 0.0045 \\
         & Mean closed e-BH (points) & 452.395 & 0.8984 & 0.0070 \\
         & Sequential resetting (points) & 491.346 & 0.9814 & 0.0013 \\
         & Sequential resetting (intervals) & 371.550 & 0.9814 & 0.0001 \\
        \bottomrule
    \end{tabular}
\end{table}

The experiment shows that the two e-BH procedures concern different applications than the sequential resetting procedures. These applications have different reporting units and FDR targets.

\end{document}